\RequirePackage{snapshot}  
\documentclass{llncs}
\usepackage{graphicx}
\usepackage{amssymb}
\usepackage{amsmath}
\usepackage{color}
\usepackage{times}
\usepackage{verbatim}
\usepackage{url}
\usepackage{tikz}
\usetikzlibrary{chains, shapes.misc}
\tikzset{
    nonterminal/.style={
      rectangle,
      minimum size=6mm,
      very thick,
      draw=red!50!black!50,
      top color=white,
      bottom color=red!50!black!20,
      font=\itshape},
    terminal/.style={
      rounded rectangle,
      minimum size=6mm,
      very thick,draw=black!50,
      top color=white,bottom color=black!20,
      font=\ttfamily}
    }
\usepackage{wrapfig}
\usetikzlibrary{arrows,decorations.pathmorphing,backgrounds,positioning,fit,petri}

\usepackage{subcaption}
\usepackage{subfig}

\usepackage[ruled,vlined]{algorithm2e}

\def\C{\mathcal{C}}	
\newcommand{\bigO}{\ensuremath{\mathcal{O}}} 

\def\I{\mathcal{I}}

\def\uri{\mathbf{U}}
\def\bn{\mathbf{B}}
\def\lit{\mathbf{L}}
\def\const{\mathbf{C}}
\def\var{\mathbf{V}}

\def\local{\textrm{local}}
\def\lclosure{\textsf{lclosure}}

\title{ Query Answering over Contextualized RDF/OWL Knowledge with Forall-Existential Bridge Rules:
Decidable Finite Extension Classes (Post Print)
}

\author{
Mathew Joseph$^{1,2}$  \and Gabriel Kuper$^2$ \and Till Mossakowski$^3$ \and Luciano Serafini$^1$
}
\institute{ 
$^1$ DKM, FBK-IRST, Trento, Italy \\
$^2$ DISI, University Of Trento, Trento, Italy\\
$^3$ Faculty of Computer Science,  Otto-von-Guericke University of Magdeburg, 39016, Magdeburg, Germany \\
\medskip
\email{\{joseph, kuper\}@disi.unitn.it, \{serafini\}@fbk.eu, \{mossakow\}@iws.cs.uni-magdeburg.de}
}

\begin{document}
\maketitle
\begin{abstract}
The proliferation of contextualized knowledge in the Semantic Web (SW) has led to the popularity of 
knowledge formats such as \emph{quads} in the SW community. A quad is an extension of an RDF triple with contextual information
of the triple. In this paper, we study the problem of query answering over quads augmented with forall-existential 
bridge rules that enable interoperability of reasoning between triples in various contexts. 
We call a set of quads together with such expressive bridge rules, a quad-system. 
Query answering over quad-systems is undecidable, in general. We derive decidable classes of quad-systems, 
for which query answering can be done using forward chaining. Sound, complete and terminating procedures, which are adaptations
of the well known chase algorithm, are provided for these classes for deciding query entailment. 
Safe, msafe, and csafe class of 
quad-systems restrict the structure of blank nodes generated during the chase computation process
to be directed acyclic graphs (DAGs) of bounded depth. RR and restricted RR classes do not allow the generation of 
blank nodes during the chase computation process. Both data and combined complexity of query entailment has been established for the classes 
derived. We further show that quad-systems are 
equivalent to forall-existential rules whose predicates are restricted to ternary arity, modulo polynomial time translations. 
We subsequently show that the technique of safety, strictly subsumes in expressivity, some of the well known and 
expressive techniques, such as joint acyclicity and model faithful acyclicity, used for decidability guarantees in the realm of
forall-existential rules.  
\end{abstract}

\begin{keywords}
Contextualized Query Answering, Contextualized RDF/OWL knowledge bases, Multi-Context Systems, Quads, 
Query answering, forall-existential rules, Knowledge Representation,  Semantic Web
\end{keywords}

\section{Introduction}
\noindent As the Semantic Web (SW) is getting more and more ubiquitous and its constellation of interlinked ontologies, the web of data, is seamlessly
proliferating at a steady rate, more and more applications have started using SW as a back end, providing their users manifold services, 
leveraging semantic technologies. One of the main reasons why SW enjoys such admirable hospitality 
from its mammoth geographically disparate users is its ``simple'' and ``open'' model. The model is simple, as the only intricacy that 
a creator/consumer of a SW application needs to be equipped with is that of a (RDF) triple. A triple $t$ $=$ $(s,p,o)$ represents the most 
basic piece of knowledge in the SW, where $s$, called the \emph{subject}, is an identifier for a person, place, thing, value, or a resource in general, about which the 
creator of $t$ intended to express his/her knowledge using $t$. $p$, called the \emph{predicate}, is an identifier for a property, attribute, or in general a 
binary relation that relates $s$ with the component $o$, called the object, that is also an identifier for a resource similar to $s$. 
The model is called open, as it allows anybody, anywhere around the 
world to freely create their RDF/OWL ontologies about a domain of their choice, and publish them in (embedded) RDF/OWL formats in their web portals, also linking via 
URIs to the concepts in other similarly published ontologies. Thus the open model, in order to promote reuse and freedom, imposes no arbitration mechanism 
for the ontologies users publish on the SW.

A problem caused by this open model is
that any piece of knowledge which a person publishes is often his/her own perspective about a particular domain, which largely is relative to this person. 
 As a consequence, the truth value of a piece of knowledge in the SW is context-dependent. 
Recently, as a solution to the aforementioned problem, the SW community adopts the use of quads, an extension of 
triples, as the primary carrier of knowledge. A quad $c\colon (s,p,o)$ thus adds a fourth component of 
the context $c$ to the triple $(s,p,o)$,  explicating the identifier of the context in which the triple holds. As a result, more and more 
triple-stores are becoming quad-stores. Some of the popular quad-stores are 4store\footnote{http://4store.org}, Openlink Virtuoso
\footnote{http://virtuoso.openlinksw.com/rdf-quad-store/}, and some
of the currently popular triple-stores like Sesame\footnote{http://www.openrdf.org/}, 
Allegrograph\footnote{http://www.franz.com/agraph/allegrograph/}
internally keep track of the contexts of triples. Some of the recent initiatives
in this direction have also extended existing formats like N-Triples to N-Quads, which the RDF 1.1 has introduced
as a W3C recommendation.  
The latest Billion triple challenge datasets  have all been  released in the N-Quads format.

Other benefits of quads over triples are that they allow knowledge creators to specify various attributes of meta-knowledge 
that further qualify knowledge~\cite{caroll-bizer-named-graphs-data-provenance-wwww2005}, and also allow users to query for this
meta knowledge~\cite{QueryingMetaknowledge}. These attributes, which explicate the various assumptions under which knowledge holds, 
are also called \emph{context dimensions}~\cite{cyc_context_space}. Examples of context dimensions are
provenance, creator, intended user, creation time, validity time, geo-location, and topic. 
Having defined knowledge that is contextualized, as in $c_1 \colon($\textsf{Renzi}, \textsf{primeMinsiterOf}, \textsf{Italy}$)$, 
one can now declare in a meta-context $mc$, statements such as 
$mc\colon(c_1$, \textsf{creator}, \textsf{John}$)$, $mc\colon(c_1$, \textsf{expiryTime}, \textsf{``jun-2016''}$)$ that talk 
about the knowledge in context $c_1$, in this case its creator and expiry time.
Another benefit of such a contextualized approach is that it opens possibilities of interesting ways for
querying a contextualized knowledge base. For instance, if context $c_1$ contains knowledge about football world cup 2014
and context $c_2$ about football euro cup 2012, then the query ``who beat Italy in both world cup 2014 and
euro cup 2012'' can be formalized as the conjunctive query: 
\[
 c_1\text{: }(x,\textsf{beat},\textsf{Italy}) \wedge c_2\text{: }(x,\textsf{beat},\textsf{Italy}), 
\] 
where $x$ is a variable.

When reasoning with knowledge in quad form, since knowledge can 
be grouped and divided context wise and simultaneously be fed to separate reasoning engines, 
this approach  improves both efficiency and scalability~\cite{BozzatoComparing}. Besides the above flexibility, \emph{bridge rules}~\cite{DDL}
can be provided for inter-operating the knowledge in different contexts. Such rules are primarily of the form:
\begin{eqnarray}\label{eqn:bridgeRule}
c:\phi \rightarrow c':\phi' 
\end{eqnarray}
where $\phi,\phi'$ are both atomic concept (role) symbols, $c,c'$ are contexts. The semantics of such a rule is that if, for any $\vec a$, 
$\phi(\vec a)$ holds in context $c$, then $\phi'(\vec a)$ should hold in context $c'$, 
where $\vec a$ is a unary/binary vector depending on whether $\phi, \phi'$ are concept/role symbols.
Although such bridge rules serve the purpose of specifying  knowledge interoperability from a source 
context $c$ to a target context $c'$, in many practical situations there is the need of 
inter-operating  multiple source contexts with multiple target contexts, for which the bridge rules of the form (\ref{eqn:bridgeRule})
are inadequate. Besides, one would also want the ability of creating new values in target contexts for the bridge rules.

In this work, we study contextual reasoning and query answering over contextualized RDF/OWL knowledge bases in the presence of
\emph{forall-existential bridge rules} 
that allow conjunctions and existential quantifiers in them, 
and hence are more expressive than those in DDL~\cite{DDL} and 
McCarthy et al.~\cite{McCarthy95formalizingcontext}.
We provide a basic semantics for contextual reasoning based on which we provide procedures for 
conjunctive query answering. For query answering, we use the notion of a \emph{distributed chase},
which is an extension of the standard \emph{chase}~\cite{JohnsonK84,AbiteboulHV95} that is widely used in the
knowledge representation (KR) and Database (DB) settings for similar purposes. As far as the semantics for reasoning is concerned,
we adopt the approach given in works such as Distributed Description Logics~\cite{DDL}, 
E-connections~\cite{econnections}, and two-dimensional logic of 
contexts~\cite{KlaGutAAAI11}, to use a set of interpretation structures 
as a model for contextualized knowledge. In this way, knowledge in each context is separately interpreted in a different 
interpretation structure.  
\noindent The main contributions of this work are:
\begin{enumerate}
 \item  We formulate a context-based semantics that reuses the standard RDF/OWL semantics, which can be used for reasoning over 
 quad-systems. Studying conjunctive query answering
 over quad-systems, it turns out that the entailment problem of conjunctive queries is undecidable for the most 
 general class of quad-systems, called \emph{unrestricted quad-systems}. 
  \item We derive decidable subclasses of unrestricted quad-systems, namely \emph{csafe}, \emph{msafe}, and \emph{safe} quad-systems, 
  for which we detail both data and combined complexities of conjunctive query entailment. These classes are based on the
  constrained DAG structure of Skolem blank nodes generated during the chase construction. We also provide decision procedures to decide whether
  an input quad-system is safe (csafe, msafe) or not. 
  \item We further derive less expressive classes, RR and restricted RR quad-systems, 
  for which no Skolem blank nodes are generated during the chase construction. 
\item We show that the class of unrestricted quad-systems is equivalent to the class of ternary $\forall\exists$ rule sets. 
We compare the derived classes of quad-systems with well known subclasses of $\forall\exists$ rule sets, such as jointly acyclic and model faithful acyclic
rule sets, and show that the technique of safety we propose, subsumes these other techniques, in expressivity. 
\end{enumerate}
The paper is structured as follows. In section \ref{sec:CQS}, we formalize the idea of contextualized quad-systems, giving
various definitions and notations for setting the background.
In section \ref{sec:Query}, we formalize the problem of query answering for quad-systems, define notions such as distributed chase
that are further used for query answering, and give the undecidability results of query entailment on
unrestricted quad-systems. In section \ref{sec:safe}, we present
\emph{csafe}, \emph{msafe}, and \emph{safe} quad-systems and their computational properties. In section \ref{sec:horn}, RR quad-systems
and restricted RR quad-systems are introduced. In section \ref{sec:Comparison}, we prove the equivalence of quad-systems with ternary $\forall\exists$ 
rule sets, and formally compare a few well known decidable classes in the realm of $\forall\exists$ rules to the classes of quad-systems,
we presented in section \ref{sec:safe}.
We provide a detailed discussion to other relevant related works in section \ref{sec:related work}, and conclude in 
section \ref{section:conclusion}.

Note that parts of the contents of section \ref{sec:CQS} and section \ref{sec:Query} has been taken from conference papers \cite{josephCilc2014}
and \cite{josephRR2014}.
\section{Contextualized Quad-Systems}\label{sec:CQS}
 \vspace{-5pt} 
 \noindent In this section, we formalize the notion of a
 quad-system and its semantics.  For any vector or sequence $\vec x$, we denote
 by $\|\vec x\|$ the number of symbols in $\vec x$, and by $\{\vec x\}$ the set of symbols
 in $\vec x$. For any sets $A$ and $B$, $A \rightarrow B$ denotes the
 set of all functions from set $A$ to set $B$.
 
Given  the set of URIs $\uri$, the set of blank nodes $\bn$, and  the set of literals $\lit$, the set 
$\const=\uri \uplus \bn \uplus \lit$ is called the set of (RDF) constants. Any $(s,p,o) \in \const \times 
\const \times \const$ is called a generalized RDF triple (from now on,
just triple). A graph is a set of triples.
A \emph{quad} is a tuple of the form $c\colon(s,p,o)$, where $(s,p,o)$ is a triple and $c$ is a URI\footnote{Although,
in general a context identifier can be a constant, for the ease of notation, we restrict them to be a URI}, 
called the \emph{context identifier} that denotes the context of the RDF triple.
A \emph{quad-graph} is defined as a set of quads.
For any quad-graph $Q$ and any context identifier $c$, we denote by $graph_Q(c)$ the set $\{(s,p,o)|c\colon(s,p,o) \in Q\}$.
We denote by $Q_{\C}$ the quad-graph whose set of context identifiers is $\C$. 
The set of constants occurring in $Q_{\C}$ is 
given as $\const(Q_{\C})$ $=$ $\{c, s, p, o \ | \ c\colon(s,p,o) \in Q_{\C}\}$.
 The set of URIs in $Q_{\C}$ is given by $\uri(Q_{\C})$ $=$ $\const(Q_{\C})$ $\cap$ $\uri$. The set of blank nodes
 $\bn(Q_{\C})$ and the set of literals $\lit(Q_{\C})$ are similarly defined.
 Let $\var$ be the set of variables, any element of the set ${\const}^{\var}=\var \cup \const$ is a \emph{term}. 
 Any $(s,p,o) \in {\const}^{\var} \times {\const}^{\var} \times
 {\const}^{\var}$ is called a \emph{triple pattern}, and an expression
 of the form $c\colon(s,p,o)$, where $(s,p,o)$ is a triple pattern,
 $c$ a context identifier, is called a \emph{quad pattern}.  A triple
 pattern $t$, whose variables are elements of the vector $\vec x$ or elements of the vector 
 $\vec y$ is written as $t(\vec x, \vec y)$.  For any function $f\colon A \rightarrow B$, the \emph{restriction} of $f$ to a set $A'$,
is the mapping $f|_{A'}$ from $A' \cap A$ to $B$ such that $f|_{A'}(a)=f(a)$, for each $a \in A \cap A'$.
 For any triple pattern $t=(s,p,o)$ and function $\mu$ from  $\var$ to a set $A$,
 $t[\mu]$ denotes $(\mu'(s), \mu'(p), \mu'(o))$, where $\mu'$ is an extension of $\mu$ to $\const$ such that $\mu'|_{\const}$
is the identity function. For any set of  triple patterns $G$, $G[\mu]$ denotes $\bigcup_{t \in G} t[\mu]$.
For any vector of constants $\vec a=\langle a_1,\dots,a_{\|\vec a\|}\rangle$,
 and vector of variables $\vec x$ of the same length, $\vec x/\vec
 a$ is the function $\mu$ such that 
 $\mu(x_i)=a_i$, for $1\leq i\leq \|\vec a\|$. We use the notation $t(\vec a, \vec y)$ to
 denote $t(\vec x, \vec y)[\vec x/\vec a]$. 
 Similarly, the above notations are also extended to sets of
 quad-patterns. For instance $Q(\vec x, \vec y)$
 denotes a  set of quad-patterns, whose variables
 are from $\vec x$ or $\vec y$, and $Q(\vec a, \vec y)$ is written for $Q(\vec x, \vec y)[\vec x/\vec a]$.
For the sake of  interoperating  knowledge in different contexts, bridge rules need to be provided:
\paragraph{Bridge rules (BRs)} 
 Formally, a BR is of the form:
\begin{eqnarray}\label{eqn:intraContextualRuleWithQuantifiers}
  \forall \vec x \forall \vec z \ [c_1\text{: }t_1(\vec x, \vec z) \wedge ... \wedge c_n\text{: }t_n(\vec x, \vec z) 
 \rightarrow \exists \vec y \ c'_1\text{: }t'_1(\vec x, \vec y) \wedge ... \wedge c'_m\text{: }t'_m(\vec x, \vec y)]
\end{eqnarray}  
where $c_1, ..., c_n, c'_1,...,c'_m$ are context identifiers, 
$\vec x$, $\vec y$, $\vec z$ are vectors of variables such that $\{\vec x\}, \{\vec y\}$,
and $\{\vec z\}$ are pairwise disjoint. $t_1(\vec x$, $\vec z)$, 
..., $t_n(\vec x$, $\vec z)$ are triple patterns which do not contain blank-nodes, and whose set of
variables are from $\vec x$ or $\vec z$. $t'_1(\vec x$, $\vec y)$, ..., $t'_m(\vec x, \vec y)$ are triple patterns, whose set of
variables are from $\vec x$ or $\vec y$, and also does not contain blank-nodes. 
For any BR $r$ of the form (\ref{eqn:intraContextualRuleWithQuantifiers}), 
$body(r)$ is the set of quad patterns $\{c_1\text{: }t_1(\vec x$, $\vec z)$,...,$c_n\text{: }t_n(\vec x$, $\vec z)\}$,
and $head(r)$ is the set of quad patterns $\{c'_1\text{: }t'_1(\vec x, \vec y)$, ...  $c'_m\text{: }t'_m(\vec x$, $\vec y)\}$, 
and the frontier of $r$, $\mathit{fr}(r)$ $=$ $\{\vec x\}$. Occasionally, we also note the BR $r$ above as $body(r)(\vec x$, $\vec z)$ $\rightarrow$
$head(r)(\vec x$, $\vec y)$. The set of terms in a BR $r$ is: 
  \[\const^{\var}(r)=\{c,s,p,o \ | \ c:(s,p,o) \in body(r) \cup head(r) \}\]
 The set of terms for a set of BRs $R$ is $\const^{\var}(R)=\bigcup_{r \in R} \const^{\var}(r)$.
The URIs, blank nodes, literals, variables of a BR $r$ (resp. set of
BRs $R$) are similarly defined, and
  are denoted as $\uri(r)$, $\bn(r)$, $\lit(r)$, $\var(r)$ (resp. $\uri(R)$, $\bn(R)$, $\lit(R)$, $\var(R)$), respectively.
\begin{definition}[Quad-System]
A \emph{quad-system} $QS_{\C}$ is defined as a pair $\langle Q_{\C}, R\rangle$, where $Q_{\C}$ is a quad-graph, whose
set of context identifiers is $\C$, and $R$ is a set
of BRs.
\end{definition}
\noindent  
For any quad-system, $QS_{\C}$ $=$ $\langle Q_{\C},R \rangle$, the set of constants in $QS_{\C}$ is given by
$ \const(QS_{\C})=\const(Q_{\C}) \cup \const(R) \nonumber$.
The sets $\uri(QS_{\C})$, $\bn(QS_{\C})$, $\lit(QS_{\C})$, 
and $\var(QS_{\C})$ are similarly defined for any quad-system $QS_{\C}$.
For any quad-graph  $Q_{\C}$ (BR $r$), its symbol size $\|Q_{\C}\|$ $(\|r\|)$ is the number of symbols
 required to print $Q_{\C}$ $(r)$. Hence, $\|Q_{\C}\|\approx 4*|Q_{\C}|$, where $|Q_{\C}|$ denotes the cardinality of
 the set $Q_{\C}$. Note that $|Q_{\C}|$ equals the number of quads in $Q_{\C}$. 
 For a BR $r$, $\|r\|\approx 4*k$, where $k$ is the number of quad-patterns in $r$.
 For a set of BRs $R$, $\|R\|$ is given as $\Sigma_{r \in R} \|r\|$.   
 For any quad-system $QS_{\C}$ $=$ $\langle Q_{\C},R \rangle$, its size $\|QS_{\C}\|=\|Q_{\C}\|+\|R\|$. 
\paragraph{Semantics}
In order to provide a semantics for enabling reasoning over a quad-system, 
we need to use a local semantics for each context to interpret the knowledge pertaining to it. 
Since the primary goal of this paper is a decision procedure for query answering over quad-systems based on 
forward chaining, we consider the following desiderata for the choice of the local semantics and its deductive machinery:
\begin{itemize}
 \item there exists 
 an operation $\lclosure()$ 
 that computes the deductive closure of a graph w.r.t to the local semantics using the local inference rules in a set LIR,
 \item each inference rule in LIR is \emph{range restricted}, i.e. non value-generating,
 \item given a finite graph as input, the $\lclosure()$ operation terminates with a finite graph as output in polynomial time
 whose size is polynomial w.r.t. to the input set. 
\end{itemize}
Some of the alternatives for the local semantics satisfying the above mentioned criterion
 are  Simple, RDF, RDFS~\cite{Hayes04rdfsemantics}, OWL-Horst~\cite{terHorst200579} etc. 
 Assuming that a local semantics has been fixed, for any context $c$, 
 we denote by $I^c=\langle\Delta^c, \cdot^c\rangle$ an interpretation structure for the local semantics, 
 where $\Delta^c$ is the interpretation domain, $\cdot^c$ the corresponding interpretation function.
 Also $\models_{\local}$ denotes the local satisfaction relation between a local interpretation structure and
 a graph.
Given a quad graph $Q_{\C}$, a \emph{distributed interpretation structure} is an 
indexed set $\I^{\C}=\{I^c\}_{c\in\C}$, 
 where $I^c$ is a local interpretation structure, for each $c \in \C$.
We define the satisfaction relation $\models$ between a distributed interpretation
structure $\I^{\C}$ and a quad-system $QS_{\C}$ as: 
\begin{definition}[Model of a Quad-System]\label{def:model-quad-system}
 A distributed interpretation structure $\I^{\C}=\{I^c\}_{c\in\C}$  satisfies
 a quad-system $QS_{\C}$ $=$ $\langle Q_{\C}$, $R\rangle$, in symbols $\I^{\C}$ $\models$ $QS_{\C}$, iff all the following conditions are satisfied:
 \vspace{-5pt}
\begin{enumerate}
\item\label{item:graphSatisfaction}  
$I^c \models_{\local} graph_{Q_{\C}}(c)$, 
for each $c \in \C$;
\item\label{item:rigid} 
$a^{c_i}=a^{c_j}$, for any $a \in \const$, 
$c_i, c_j \in \C$; 
 \item\label{item:BRSatisfaction} 
 for each BR $r \in R$ of the form~(\ref{eqn:intraContextualRuleWithQuantifiers}) and
for each $\sigma \in \var \rightarrow \Delta^{\C}$, where $\Delta^{\C}$ $=$ $\bigcup_{c\in \C} \Delta^c$, if 
\vspace{-5pt}
\[
I^{c_1} \models_{\local} t_1(\vec x, \vec z)[\sigma], ..., I^{c_n} \models_{\local} 
t_n(\vec x, \vec z)[\sigma], 
\]
then there exists a function $\sigma'\supseteq \sigma$, such that
\vspace{-5pt}
 \[
I^{c'_1} \models_{\local} t'_1(\vec x, \vec y)[\sigma'], ..., I^{c'_m} \models_{\local} t'_m(\vec x, \vec y)[\sigma'].
 \]
\end{enumerate}
\end{definition}
\vspace{-5pt}
\noindent
Condition \ref{item:graphSatisfaction} in the above definition ensures that for any model $\I^{\C}$
of a quad-graph, each $I^c \in \I^{\C}$ is a local model of
the set of triples in context $c$. Condition \ref{item:rigid} ensures that
any constant $c$ is rigid, i.e. represents the same resource across a quad-graph, irrespective of the context in which it occurs.
Condition \ref{item:BRSatisfaction} ensures that any model of a quad-system satisfies each BR in it.
Any $\I^{\C}$ such that $\I^{\C} \models QS_{\C}$ is said to be a model of $QS_{\C}$.
A quad-system $QS_{\C}$ is said to be \emph{consistent} if there exists a model $\I^{\C}$, such that 
$\I^{\C} \models QS_{\C}$, and otherwise said to be \emph{inconsistent}. 
For any quad-system $QS_{\C}=\langle Q_{\C},R \rangle$, it can be the case that $graph_{Q_{\C}}(c)$ is locally consistent,
i.e. there exists an $I^c$ such that $I^c \models_{\local} graph_{Q_{\C}}(c)$,
 for each $c \in \C$, whereas $QS_{\C}$ is not consistent. This is because the set of BRs $R$ adds more 
knowledge to the quad-system, and restricts the set of models that satisfy the quad-system. 
\begin{definition}[Quad-system entailment]
(a) A quad-system $QS_{\C}$ entails a quad $c\colon(s$, $p$, $o)$, in symbols $QS_{\C}$ $\models$ $c\colon(s,p,o)$,
iff for any distributed interpretation structure
$\I^{\C}$, if $\I^{\C} \models QS_{\C}$ then $\I^{\C} \models \langle \{c\colon(s,p,o)\}, \emptyset \rangle$.
 (b) A quad-system $QS_{\C}$ entails a quad-graph $Q'_{\C'}$, in symbols $QS_{\C} \models Q'_{\C'}$ iff $QS_{\C} \models c\colon(s,p,o)$
for any $c\colon(s,p,o) \in Q'_{\C'}$. 
(c)  A quad-system $QS_{\C}$ entails a BR $r$ iff 
for any 
$\I^{\C}$, if $\I^{\C} \models QS_{\C}$ then $\I^{\C} \models \langle \emptyset, \{r\} \rangle$. 
(d) For a set of BRs
$R$, $QS_{\C} \models R$ iff $QS_{\C} \models r$, for every $r \in R$. 
(e)  Finally, a quad-system $QS_{\C}$ entails another 
quad-system $QS'_{\C'}$ $=$ $\langle Q'_{\C'}, R' \rangle$, in symbols $QS_{\C} \models QS'_{\C'}$ iff $QS_{\C} \models Q'_{\C'}$ and
$QS_{\C} \models R'$. 
\end{definition}
\vspace{-5pt}
\noindent We call  the decision problems  corresponding to the entailment problems (EPs) in  (a), (b), (c), (d), and (e) as
\emph{quad EP, quad-graph EP, BR EP, BRs EP, and quad-system EP}, respectively.
\vspace{-5pt}
\vspace{-5pt}
\section{Query Answering on Quad-Systems}\label{sec:Query}
\noindent 
In the realm of quad-systems, the classical conjunctive queries or select-project-join queries are 
slightly extended to what we call \emph{Contextualized
Conjunctive Queries} (CCQs). A CCQ $CQ(\vec x)$ is an expression of the form:
\begin{equation}\label{eqn:CCQ}
 \exists \vec y \ q_1(\vec x, \vec y) \wedge ... \wedge q_p(\vec x, \vec y)
\end{equation}
where $q_i$, for $i=1,...,p$ are quad patterns over vectors of \emph{free variables} $\vec x$ 
and \emph{quantified variables} $\vec y$. A CCQ is called a boolean CCQ if it does not have any free variables.
With some abuse, we sometimes discard the logical symbols in a CCQ and consider it as a set of quad-patterns. 
For any CCQ $CQ(\vec x)$ and a vector $\vec a$ of constants such that  $\|\vec x\|=\|\vec a\|$, $CQ(\vec a)$ is boolean.  
A vector $\vec a$ is an \emph{answer} for a CCQ $CQ(\vec x)$ 
w.r.t. structure $\I^{\C}$, 
in symbols $\I^{\C} \models CQ(\vec a)$, iff there exists assignment $\mu\colon \{\vec y\}\rightarrow  \bn$ such that 
$\I^{\C} \models \bigcup_{i=1, \ldots, p} q_i(\vec a, \vec y)[\mu]$. 
A vector $\vec a$ is a \emph{certain answer} for a CCQ $CQ(\vec x)$ over a quad-system $QS_{\C}$, iff
$\I^{\C} \models CQ(\vec a)$, for every model $\I^{\C}$ of $QS_{\C}$.
 Given a quad-system $QS_{\C}$, a CCQ $CQ(\vec x)$, and a vector $\vec a$,
 decision problem of determining whether $QS_{\C} \models CQ(\vec a)$ is called the \emph{CCQ EP}.
It can be noted that the other decision problems over quad-systems, namely Quad/Quad-graph EP, BR(s) EP, Quad-system EP, are reducible to the 
CCQ EP (See Property~\ref{prop:QuadsystemEPsReductiontoCCQeps}). Hence, in this paper, we primarily focus on the CCQ EP. 
\subsection{dChase of a Quad-System}
\noindent In order to build a procedure for query answering over a quad-system, we employ what has been called
in the literature a \emph{chase}~\cite{JohnsonK84,AbiteboulHV95}. Specifically,  
we adopt notions of the \emph{restricted chase}  in Fagin et al.~\cite{Fagin05dataexchange} 
(also called non-oblivious chase).
In order to fit the framework of quad-systems, we extend the standard notion of chase to a 
 \emph{distributed chase}, abbreviated \emph{dChase}. In the following, we show how the dChase of a quad-system can be constructed.
 
For a set of  quad-patterns $S$ and a set of terms $T$, we define the relation $T$-\emph{connectedness} between quad-patterns in $S$ as the least relation with:
\begin{itemize}
 \item  $q_1$ and $q_2$ are $T$-connected, if $\const^{\var}(q_1) \cap \const^{\var}(q_2) \cap T \neq \emptyset$,
 for any two quad-patterns $q_1$, $q_2 \in S$,
 \item  if $q_1$ and $q_2$ are $T$-connected,  and $q_2$ and $q_3$ are $T$-connected, then  $q_1$ and $q_3$
 are also $T$-connected, for any quad-patterns $q_1, q_2, q_3 \in S$.
\end{itemize}
It can be noted that $T$-connectedness is an equivalence relation and partitions $S$ into a set of $T$-components 
(similar notion is called a \emph{piece} in Baget et al.~\cite{BagetLMS11}).
Note that for two distinct $T$-components $P_1$, $P_2$ of $S$, 
$\const^{\var}(P_1) \cap \const^{\var}(P_2) \cap T = \emptyset$.
For any BR $r=body(r)(\vec x, \vec z) \rightarrow head(r)(\vec x, \vec y)$, suppose $P_1, P_2, \ldots, P_k$ are the pairwise distinct 
$\{\vec y\}$-components of $head(r)(\vec x, \vec y)$, then $r$ can be replaced by the semantically equivalent
set of BRs $\{body(r)(\vec x$, $\vec z)$ $\rightarrow$ $P_1$, \ldots, $body(r)(\vec x$, $\vec z)$ $\rightarrow$ $P_k\}$ whose symbol size
is worst case quadratic w.r.t. the symbol size of $r$. Hence, w.l.o.g. we assume that for any BR $r$, the set of quad-patterns
$head(r)$ is a single component w.r.t. the set of existentially quantified variables in $r$. 

Considering the fact that the local semantics for contexts are fixed a priori (for instance RDFS), both the number of rules in the set of local inference rules LIR and 
the size of each rule in LIR can be assumed to be a constant. Note that each local inference rule is range restricted and does not contain
existentially quantified variables in its head. Any $ir \in$ LIR is of the form:
\begin{eqnarray}\label{eqn:LocalIR}
\forall \vec x \forall \vec z \ [t_1(\vec x, \vec z) \wedge \ldots \wedge t_k(\vec x, \vec z) \rightarrow t'_1(\vec x) 
],   
\end{eqnarray}
where $t_i(\vec x, \vec z)$, for $i=1$, \ldots, $n$ are triple patterns, whose variables are from $\{\vec x\}$ or $\{\vec z\}$, and 
$t'_1(\vec x)$ is a triple pattern, whose variables are from $\{\vec x\}$. Hence, for any quad-system $QS_{\C}$ $=$ $\langle Q_{\C}$, $R\rangle$
in order to accomplish the effect of local inferencing in each context $c\in \C$, for each $ir \in$ LIR of the form (\ref{eqn:LocalIR}), we could augment $R$ with a 
BR $ir_c$ of the form:
\begin{eqnarray}
\forall \vec x \forall \vec z \ [c\colon t_1(\vec x, \vec z) \wedge \ldots \wedge c\colon t_k(\vec x, \vec z) \rightarrow c\colon t'_1(\vec x) 
] \nonumber   
\end{eqnarray}
Since $\|$LIR$\|$ is a constant and the size of the augmentation is linear in $|\C|$, 
w.l.o.g we assume that the set $R$ contains a BR $ir_c$, for each $ir \in $ LIR, $c\in \C$.

Given a quad-system $QS_{\C}$, we denote by $\bn_{sk} \subseteq \bn$, a set of blank nodes 
called \emph{Skolem blank nodes}, such that  $\bn_{sk} \cap \bn(QS_{\C})=\emptyset$. For any BR $r$ $=$ $body(r)(\vec x$, $\vec z)$
$\rightarrow$ $head(r)(\vec x$, $\vec y)$ and 
an assignment $\mu\colon \{\vec x\}$ $\cup$ 
 $\{\vec z\}$ $\rightarrow$ $\const$,  the \emph{application} of $\mu$ on $r$ is defined as: 
 \[
apply(r, \mu)=head(r)[\mu^{ext(\vec y)}]  
 \]
where $\mu^{ext(\vec y)}\supseteq \mu$ such that  
 $\mu^{ext(\vec y)}(y_i)=\_\colon b$ is a fresh blank node from $\bn_{sk}$, for each $y_i \in \{\vec y\}$. 

We assume that there exists an order $\prec_l$ (for instance, lexicographic order) on the 
set of constants. 
We extend $\prec_l$ to the set of quads such that  for any two quads $c\colon(s,p,o)$ and $c'\colon(s',p',o')$, 
$c\colon(s,p,o) \prec_l c'\colon(s',p',o')$, iff $ c \prec_l c'$, or 
$c=c', s\prec_l s'$, or $c=c', s=s', p \prec_l p'$, or $c=c', s=s', p=p', o\prec_l o'$.
It can be noted that $\prec_l$ is a strict linear order over the set of all quads. 
For any finite quad-graph $Q_{\C}$, the $\prec_l$-greatest quad of $Q_{\C}$, denoted greatestQuad$_{\prec_l}(Q_{\C})$, is the quad $q \in Q_{\C}$ such that 
$q' \prec_l q$, for every other $q' \in Q_{\C}$. Also, the order $\prec_q$ is defined over the set of finite quad-graphs as follows:
for any two finite quad-graphs $Q_{\C}$, $Q'_{\C'}$, 
\begin{description}
 \item $Q_{\C}$ $\prec_q$ $Q'_{\C'}$, if (i) $Q_{\C}$ $\subset$ $Q'_{\C'}$; 
 \item $Q_{\C}$ $\prec_q$ $Q'_{\C'}$,  if (i) does not hold and   
(ii) greatestQu-\linebreak -ad$_{\prec_l}(Q_{\C}$ $\setminus$ $Q'_{\C'})$ $\prec_l$ 
 greatestQuad$_{\prec_l}(Q'_{\C'}$ $\setminus$ $Q_{\C})$; 
 \item $Q_{\C}$ $\not \prec_q$ $Q'_{\C'}$, if both (i) and (ii) are not satisfied;
\end{description}
A relation $R$ over a set $A$ is called a \emph{strict linear order} iff 
$R$ is irreflexive, transitive, and $R(a,b)$ or $R(b,a)$ holds, for every distinct $a, b \in A$.
\begin{property}\label{prop:precLinearOrder}
 Let $\mathcal{Q}$ be the set of all finite quad-graphs; $\prec_q$ is a strict linear order 
 over $\mathcal{Q}$.
\end{property}
\noindent Also, we now define in parallel the dChase of a quad system $QS_{\C}$ $=$ $\langle Q_{\C}$, $R\rangle$ and the level of a quad in the dChase
 of $QS_{\C}$ as follows: 
any quad in $Q_{\C}$ is of level 0. The level of a set 
of quads is the largest among levels of quads in the set. The level of any quad that results from the application of a BR $r$
w.r.t. an assignment $\mu$ is one more than the level of the set $body(r)[\mu]$, if it has not already been assigned a level.  
Let $\prec$ be an ordering on the quad-graphs such that for any two quad-graphs $Q'_{\C'}$ and $Q''_{\C''}$ of the same level,
$Q'_{\C'} \prec Q''_{\C''}$, iff $Q'_{\C'} \prec_q Q''_{\C''}$. 
For $Q'_{\C'}$ and $Q''_{\C''}$ of different levels, $Q'_{\C'} \prec Q''_{\C''}$, iff level 
of $Q'_{\C'}$ is less than level of $Q''_{\C''}$. It can easily be seen that $\prec$ is a strict linear
order over the set of quad-graphs. 
For any  BRs $r, r'$ and assignments $\mu, \mu'$ over $\var(body(r)), \var(body(r'))$, respectively,
$(r, \mu)$ $\prec$ $(r', \mu')$ iff $body(r)[\mu]$ $\prec$ $body(r')[\mu']$. 
For any quad-graph $Q'_{\C'}$, a set of BRs $R$, a BR $r\in R$, an assignment 
$\mu \in \var(body(r)) \rightarrow \const$, let $applicable_R$ be the least ternary predicate defined inductively as:
\begin{eqnarray}
&& applicable_R(r, \mu, Q'_{\C'})  \text{ holds, if } (a) \ body(r)[\mu] \subseteq Q'_{\C'}, 
head(r)[\mu''] \not \subseteq Q'_{\C'}, \forall \mu''  \nonumber \\
&&   \supseteq \mu,  \text{ and } (b) \not \exists r' \in R, \not \exists \mu'  \text{ such that }  r' \neq r \text{ or } 
\mu' \neq \mu  \text{ with } (r',\mu') \prec (r,\mu) \text{ and  }  \nonumber \\
&&   applicable_R(r', \mu', Q'_{\C'}) ; \nonumber
\end{eqnarray}

For any quad-system $QS_{\C}$ $=$ $\langle Q_{\C},R \rangle$, let

 $dChase_0(QS_{\C})$ $=$ $Q_{\C}$;  

 $dChase_{i+1}(QS_{\C})$ $=$
$ dChase_i(QS_{\C})$ $\cup$ $apply(r$, $\mu)$, if there exists $r$ $=$ $body(r)(\vec x$, $\vec z)$ 
 $\rightarrow$ $head(r)(\vec x$, $\vec y)$ $\in$ $R$, assignment $\mu\colon \{\vec x\}$ $\cup$ 
 $\{\vec z\}$ $\rightarrow$ $\const$ such that  $applicable_R(r$, $\mu$, $dChase_i(QS_{\C}))$;
 
$dChase_{i+1}(QS_{\C})$ $=$ $dChase_i(QS_{\C})$, \text{otherwise}; for any $i \in \mathbb{N}$.
The dChase of $QS_{\C}$, noted $dChase(QS_{\C})$, is given as:
\[
 dChase(QS_{\C})= \bigcup_{i\in \mathbb{N}} dChase_i(QS_{\C})
\]
Intuitively, $dChase_i(QS_{\C})$ can be thought of as the state of $dChase(QS_{\C})$ at the end of iteration $i$. 
It can be noted that, if there exists $i$ such that  $dChase_i(QS_{\C})$ $=$ $dChase_{i+1}(QS_{\C})$, 
then $dChase(QS_{\C})$ is equal to $dChase_i($ $QS_{\C})$. 
A model $\I^{\C}$ of a quad-system $QS_{\C}$ is called \emph{universal}~\cite{Deutsch:2008}, iff the following holds:
$\I^{\C}$ is a model of $QS_{\C}$, and for any model $\I'^{\C}$ of $QS_{\C}$ there exists a homomorphism from $\I^{\C}$ to $\I'^{\C}$.
\begin{theorem}\label{dChaseUniversalModelProperty}
 For any consistent quad-system $QS_{\C}$, the following holds:
 (i) $dChase(QS_{\C})$ is a universal model of $QS_{\C}$.\footnote{Though $dChase(QS_{\C})$ is not an 
interpretation in a strict model theoretic sense,  one can easily create the corresponding interpretation $\I_{dChase(QS_{\C})}$ $=$ $\{I^c$ $=$ $\langle \Delta^c$, 
$.^c \rangle\}_{c\in \C}$, s.t. for every $c \in \C$, $\Delta^c$ is equal to set of constants in $graph_{dChase(QS_{\C})}(c)$, and 
$.^c$ is s.t  $(s,p,o) \in graph_{dChase(QS_{\C})}(c)$ iff $(s^c, o^c) \in p^c$.}, and
(ii) for any boolean CCQ $CQ()$, $QS_{\C} \models CQ()$ iff there exists a map $\mu\colon \var(CQ) \rightarrow \const$ such 
that  $\{CQ()\}[\mu]\subseteq dChase(QS_{\C})$.
\end{theorem}
\noindent An anolog of the above theorem for DLs and Databases is stated and proved in \cite{Calvanese2007}. Since the proof in \cite{Calvanese2007} can easily be 
adapted to our case, we refer the reader to \cite{Calvanese2007} for the proof. 
\noindent We call the sequence $dChase_0(QS_{\C})$, $dChase_1(QS_{\C})$, ..., the \emph{dChase sequence} of $QS_{\C}$. 
It should be noted that at each iteration $i$, after the application of a BR, any new quad added is assigned a level, and as a result any subset of
the set of quads in $dChase_i(QS_{\C})$ has a level. This assignment of levels guarantees that 
$applicable_R(r$, $\mu$, $dChase_i(QS_{\C}))$ is either true or false, for any $r\in R$, assignment $\mu\colon \var(body(r))$ $\rightarrow$ $\const$.
The following lemma shows that in a dChase sequence of a quad-system, any dChase
iteration can be performed in time exponential w.r.t. the size of the largest BR.
\begin{lemma}\label{lemma:chaseSizeIncrease}
 For a quad-system $QS_{\C}=\langle Q_{\C}, R \rangle$, for any  $i \in \mathbb{N}^+$, the following holds:   
 (i)  $dChase_i(QS_{\C})$ can be computed in time 
$\bigO($ $|R|*\|dChase_{i-1}(QS_{\C})\|^{rs})$, where $rs=max_{r \in R} \|r\|$,
(ii)  $\|dChase_i(QS_{\C})\|$ $=$ $\bigO(\|dChase_{i-1}(QS_{\C})\|+\|R\|)$.
\end{lemma}
\begin{proof}
(i)
We can first find, if there exists an $r$ among the set of BRs $R$, assignment $\mu$ such that  $applicable_R(r$, $\mu$, $dChase_{i-1}(QS_{\C}))$ holds, in the following naive way:
(1) bind the set of variables in all rules in $R$ with the set of constants in $dChase_{i-1}(QS_{\C})$. Let this set be called 
$S$. Note that $|S|=\bigO(|R|*\|dChase_{i-1}(QS_{\C}$ $)\|^{\|rs\|})$, where $rs=max_{r \in R} \|r\|$. Also, note that each of the binding in $S$ is of the form 
$body(r)(\vec x$, $\vec z)(\mu)$ $\rightarrow$ $head(r)(\vec x$, $\vec y)(\mu')$ ($\heartsuit$), where $r \in R$. (2) From the set $S$ we filter out every binding
of the form ($\heartsuit$) in which $\vec x[\mu]\neq \vec x[\mu']$. Let $S'$ be the resulting set after the above filtering operation. (3) From the set $S'$,
we now filter out all the bindings of the form ($\heartsuit$) with  $head(r)(\vec x$, $\vec y)(\mu')$ $\subseteq$ $dChase_{i-1}(QS_{\C})$, with resulting set $S''$. 
(4) If $S''=\emptyset$, then there is no $r \in R$, assignment $\mu$ such that  $applicable_R(r$, $\mu$, $dChase_{i-1}(QS_{\C}))$ is True.
Otherwise if $S'' \neq \emptyset$, then 
note that each binding of the form ($\heartsuit$) in $S''$ is such that  condition (a) of the true $applicable_R(r$, $\mu$, $dChase_{i-1}(QS_{\C}))$ is
satisfied. Now, we can sort $S''$ w.r.t. $\prec$ and select the least binding $b$ of the form ($\heartsuit$), so that  condition (b) in True condition of
$applicable_R()$ is satisfied for $b$. It can easily be seen that
$applicable_R(r$, $\mu$, $dChase_{i-1}(QS_{\C}))$ holds for the $r$, $\mu$ extracted from $b$. Since the size of each binding is at most $\|rs\|$, 
the operations (1)-(4) can be performed in time $\bigO(|R|*\|dChase_{i-1}(QS_{\C})\|^{rs})$. 
Since $dChase_i(QS_{\C})$ $=$ $dChase_{i-1}(QS_{\C})$ $\cup$ $head(r)[\mu]$, for $r$, $\mu$ with $applicable_R(r$, $\mu$, $dChase_{i-1}(QS_{\C}))$, 
$dChase_i(QS_{\C})$ can be computed in time $\bigO(\|dChase_{i-1}(QS_{\C}$ $)\|^{rs})$.

(ii) Trivially holds, since at worst
$dChase_i(QS_{\C})$ $=$  $dChase_{i-1}(QS_{\C})$ $\cup$ $head(r)[\mu]$, for $r \in R$.
\end{proof}
\begin{lemma}\label{lemma:BnodeUniquenessLemma}
 For any quad-system $QS_{\C}$, If $\_\colon b$ is a Skolem blank node in $dChase(QS_{\C})$, 
 generated by the application of assignment $\mu$ on $r$ $=$ $body(r)(\vec x$, $\vec z)$ $\rightarrow$ $head(r)(\vec x$, $\vec y)$, 
 with $\mu^{ext(\vec y)}(y_j)=\_\colon b$, $y_j \in \{\vec y\}$, then $\_\colon b$ is unique for $(r, y_j, \vec x[\mu^{ext(\vec y)}])$.
\end{lemma}
\begin{proof}
 By contradiction, suppose if $\_\colon b$ is not unique for $(r, y_j, \vec x[\mu^{ext(\vec y)}])$, i.e.
 there exists $\_\colon b'$ $\neq$ $\_\colon b$ in $dChase(QS_{\C})$, with $\_\colon b'$ generated by $r$ such that  
 $\_\colon b'$ $=$ $\mu'^{ext(\vec y)}(y_j)$ and $\vec x[\mu^{ext(\vec y)}]$ $=$ $\vec x[\mu'^{ext(\vec y)}]$.
 W.l.o.g. suppose $\_\colon b$ was generated in an iteration $l\in \mathbb{N}$ and $\_\colon b'$ in an iteration $m>l$.
 This means that $head(r)(\vec x$, $\vec y)[\mu^{ext(\vec y)}]$ $\subseteq$ $dChase_l(QS_{\C})$, and
hence $head(r)(\vec x$, $\vec y)[\mu^{ext(\vec y)}]$ $\subseteq$ $dChase_{m-1}(QS_{\C})$. 
Also, since $\mu|_{\vec x}$ $=$ $\mu'|_{\vec x}$, there $\exists \mu''$ $\supseteq$ $\mu'$ s.t. $head(r)(\vec x$, $\vec y)[\mu'']$ $\subseteq$ $dChase_{m-1}(QS_{\C})$.
This means that (a) part of the function $applicable_R$ is false, for $applicable_R(r$, $\mu'$, $dChase_{m-1}(QS_{\C}))$ to be true, and 
hence $applicable_R(r$, $\mu'$, $dChase_{m-1}(QS_{\C}))$ 
is false. Hence, our assumption
 that $\_\colon b'$ $=$ $y_j[\mu'^{ext(\vec y)}]$ is false.
Hence, $\_\colon b$ is unique for $(r$, $y_j$, $\vec x[\mu^{ext(\vec y)}])$.
\end{proof}

\noindent Although we now know how to compute the dChase of a quad-system, which can be used for deciding CCQ EP, 
the following proposition reveals that for the class of quad-systems whose BRs are of the form 
(\ref{eqn:intraContextualRuleWithQuantifiers}), which we call \emph{unrestricted quad-systems}, 
the dChase can be infinite.
\begin{proposition}\label{prop:chase-infinite}
 There exists unrestricted quad-systems whose dChase is infinite.
\end{proposition}
\vspace{-2pt}
\begin{proof}
Consider an example of a quad-system $QS_c=\langle Q_c,r \rangle$, where
$Q_c=$ $\{c\colon (a$, $\texttt{rdf:type}$, $C)\}$, and the BR $r=$
$c\colon(x$, \texttt{rdf:type}, $C)$ $\rightarrow$ $\exists y$ $c\colon(x$, $P$, $y)$, $c\colon(y$, \texttt{rdf:type}, $C)$. 
The dChase computation starts with $dChase_0(QS_c)=\{c\colon(a$, \linebreak \texttt{rdf:type}, $C)\}$, now the 
rule $r$ is applicable, and its application leads to 
$dChase_1(QS_c)$ $=$ $\{c\colon(a$, $\texttt{rdf:type}$, $C)$, $c\colon(a,P,\_\colon b_1)$, $c\colon(\_\colon b_1$, 
$\texttt{rdf:type}$, $C)\}$, where $\_\colon b_1$ is a fresh Skolem blank node. It can be noted that $r$ is yet again
applicable on $dChase_1(QS_c)$, for $c\colon(\_\colon b_1$, \texttt{rdf:type}, $C)$, which leads to the generation of another Skolem blank node, 
and so on. Hence, $dChase(QS_c)$ does not have a finite fix-point, and $dChase(QS_c)$ is infinite. 
\end{proof}
\noindent A class $\mathfrak{C}$ of quad-systems is called a \emph{finite extension class} (FEC), iff for every
 member $QS_{\C} \in $ $\mathfrak{C}$, $dChase(QS_{\C})$ is a finite set. Therefore, the class of unrestricted quad-systems is 
 not a FEC. This raises the question if there are other approaches that can be used, for instance, a
similar problem of non-finite chase is manifested in description logics (DLs) with value creation, due to the presence of existential quantifiers, 
whereas the approaches like the one in Glimm et al.~\cite{GliHoLuSa-07} provides an algorithm for CQ entailment 
based on query rewriting.  
Theorem~\ref{theorem:undecidable} below establishes the fact that the CCQ EP for unrestricted quad-systems
is undecidable. Despite this, the reader should note that the following undecidability result and its proof is only 
provided for the sake of self containedness, 
and we do not claim the undecidability theorem nor its proof to be a novel contribution, as we will show in section \ref{sec:Comparison}, 
ternary $\forall\exists$ rule sets are polynomially reducible to unrestricted quad-systems. Hence, the undecidability results provided
in Baget et al.~\cite{BagetLMS11}, Kr{\"o}tzsch et al.~\cite{DL07ELwithRoleComposition}, or Beeri et al.~\cite{BeeriVImplicationProblem81} 
can trivially be applied in our setting to obtain the undecidability result for unrestricted quad-systems.
\begin{theorem}\label{theorem:undecidable}
 The  CCQ entailment problem over unrestricted quad-systems is undecidable.
\end{theorem}
\begin{proof}(sketch)
We show that the well known undecidable problem of 
non-emptiness of intersection of context-free grammars (CFGs) is reducible to the CCQ entailment
problem. Given two CFGs, $G_1=\langle V_1, T, S_1, P_1 \rangle$ and $G_2=\langle V_2, T, S_2, P_2 \rangle$, 
where $V_1, V_2$ are the set of variables, 
$T$ such that  $T \cap (V_1 \cup V_2)=\emptyset$ is the set of terminals. $S_1 \in V_1$ is the start symbol of $G_1$,
and $P_1$ are the set of PRs of the form $v \rightarrow \vec w$, where $v \in V$, $\vec w$ is a sequence of the form
$w_1...w_n$, where $w_i \in V_1 \cup T$.  $s_2, P_2$ are defined similarly.  Deciding whether the language
generated by the grammars $L(G_1)$ and $L(G_2)$ have non-empty intersection is known to be 
undecidable~\cite{Harrison-formal-language}.

Given two CFGs $G_1=\langle V_1, T, S_1, P_1 \rangle$ and $G_2=\langle V_2, T, S_2, P_2 \rangle$,
we encode grammars $G_1, G_2$ into a quad-system $QS_c=\langle Q_{c},R\rangle$, with only a single context identifier $c$.
Each PR $r$ $=$ $v$ $\rightarrow$ $\vec w$ $\in P_1$ $\cup$ $P_2$, with $\vec w$ $=$ $w_1w_2w_3..w_n$, is encoded as a BR of the form:
$c\colon(x_1,w_1,x_2), c\colon(x_2, w_2, x_3)$, ..., $c\colon(x_n,w_n,x_{n+1})$ 
$\rightarrow$ $c\colon(x_1,v,x_{n+1})$, 
where $x_1,..,x_{n+1}$ are variables. For each terminal symbol $t_i \in T$, $R$ contains
a BR of the form:
$ c\colon(x, \texttt{rdf:type}, C)$ $\rightarrow$ $\exists y$ \ $c\colon(x,t_i,y)$, 
 $c\colon(y$, \texttt{rdf:type}, $C) \nonumber$
and $Q_c$ is the singleton:
$\{$ $c\colon(a$, \texttt{rdf:type}, $C)\}$.
It can be observed that:
\begin{eqnarray}
 QS_c \models \exists y \ c\colon(a , S_1, y) \ \wedge \ c\colon(a, S_2, y) \Leftrightarrow \nonumber \\
L(G_1)\cap L(G_2)\neq \emptyset \nonumber
\end{eqnarray}
We refer the reader to Appendix for the complete proof.
\end{proof} 
Having shown the undecidability results of query answering of unrestricted quad-systems, the rest of the paper 
focuses on defining subclasses of unrestricted quad-systems for which query answering is decidable, and establishing their relationships 
with similar classes in the realm of $\forall\exists$ rules. While defining decidable classes for quad-systems, one 
mainly has two fundamentally distinct options: 
(i) is to define notions that solely use the structure/properties of the BR part, ignoring the quad-graph part, or (ii) to define
notions that takes into account both the BR and quad-graph part. 
The decidability notions which we define in section~\ref{sec:safe}, namely safety, msafety, and csafety 
belong to type (ii), as these techniques take into account the property of the dChase of a quad-system, which is determined by both 
the quad-graph and BRs of the quad-system. Whereas the ones which we define in section~\ref{sec:horn}, namely RR and restricted RR quad-systems 
fall into type (i), as the properties of BRs alone are used. 
With an analogy between a set of BRs and a set of $\forall\exists$ rules, 
and between a quad-graph and a set of $\forall\exists$ instances, the reader should note that
such distinctions can also be made for the decidability notions in the realm of $\forall\exists$ rule sets. Techniques such as
Weak acyclicity~\cite{Fagin05dataexchange}, 
Joint acyclicity~\cite{KR11jointacyc}, and Acyclic graph of rule dependencies~\cite{BagetLMS11} belong to type (ii), as these notions 
ignore the instance part. Whereas techniques such as model faithful acyclicity~\cite{Bernardo:dlacyclicity:2012} and model 
summarizing acyclicity~\cite{Bernardo:dlacyclicity:2012} are of type (i) as both the rules and instance part is considered. 
\section{Safe, Msafe and Csafe Quad-Systems: Decidable FECs}\label{sec:safe}
\noindent In the previous section, we saw that the query answering problem over unrestricted quad-systems is undecidable, in general. 
We will also see in section~\ref{sec:Comparison} that any quad-system is polynomially translatable to a $\forall \exists$ rule set, which is also a 
first order logic  theory. 
Hence, a possible solution approach is to translate to these more expressive languages, 
and apply well known tests (see related work for details on such tests) available in these languages to check if query answering is decidable. 
If the translated quad-system passes one of these tests, then query answering can be performed on this translation using available 
algorithms in these expressive languages. But such an approach is often discouraged, because of the non-applicability of the already available
tools and techniques available for reasoning over quads. Instead, we in the following define three classes of quad-systems, 
namely \textsc{safe}, \textsc{msafe} and \textsc{csafe}, that are FECs and for which query entailment is decidable. 
Finiteness/decidability is achieved by putting certain restrictions (explained below) on the blank nodes generated in the dChase. 

Recall that, for any quad-system $QS_{\C}$, the set of blank-nodes $\bn(dChase(QS_{\C}))$
in its $dChase(QS_{\C})$ not only contains blank nodes present in $QS_{\C}$, i.e. $\bn(QS_{\C})$,
but also contains Skolem blank nodes that are generated 
during the dChase construction process.
Note that the following holds: 
$\bn_{sk}(dChase(QS_{\C}))$ $=$ $\bn(dChase(QS_{\C}))$ $\setminus$ $\bn(QS_{\C})$.
We assume w.l.o.g. that for any set of BRs $R$,  any
BR in $R$ has a unique rule identifier, and we often write $r_i$ for the BR in $R$, whose identifier is $i$.
\begin{definition}[Origin RuleId/Vector]
For any Skolem blank node $\_\colon b$, generated in the dChase by 
the application of a BR $r_i=body(r_i)(\vec x, \vec z)\rightarrow head(r_i)(\vec x, \vec y)$ using
assignment $\mu\colon \{\vec x\} \cup \{\vec z\}\rightarrow \const$, i.e. $\_\colon b=\mu^{ext(\vec y)}(y_j)$, for some $y_j\in \vec y$,
we say that the origin ruleId of $\_\colon b$ is $i$, denoted  
$originRuleId(\_\colon b)$ $=$ $i$. Moreover $\vec w$ $=$ $\vec x[\mu]$ is said to 
be the origin vector of $\_\colon b$, denoted    
 $originVector(\_\colon b)$ $=$ $\vec w$.  
\end{definition}
\noindent As we saw in Lemma~\ref{lemma:BnodeUniquenessLemma}, any such Skolem blank node $\_\colon b$,
generated in the dChase can uniquely be represented by the expression $(i, j,\vec w)$, where $i$ is
rule id, $j$ is identifier of the existentially quantified variable $y_j$ in $r_i$ substituted by $\_\colon b$ during the application of $\mu$ on $r_i$.
Also in the above case, we denote relation between each constant
$k$ $=$ $\mu^{ext(\vec y)}(x_h)$, $x_h \in \{\vec x\}$, and $\_\colon b$ with the relation \emph{childOf}.  
Moreover, since children of a Skolem blank node can be Skolem blank nodes, which themselves can have children, one can naturally define
 relation \emph{descendantOf}=childOf$^+$ as the transitive closure of childOf. Note that 
 according to the above definition, `descendantOf' is not reflexive.   
In addition, we could keep  track of the set of contexts in which a blank-node was first generated, using the following notion:
 \begin{definition}[Origin-contexts] For any quad-system $QS_{\C}$ and for any Skolem blank node $\_\colon b$ 
 $\in$ $\bn_{sk}(dCha$- $se(QS_{\C}))$, the set of origin-contexts of $\_\colon b$  
 is given by
 $origin$- -$Contexts(\_\colon b)$ $=$ $\{c \ | \ \exists i. \ c\text{:}(s,p,o)$ $\in$ $dChase_i(QS_{\C})$, $s$ $=$ $\_\colon b$ or $p$ $=$ $\_ \colon b$ or 
 $o$ $=$ $\_\colon b,$ and
$\nexists j<i$ with $c'\text{:}(s',p',o')$ $\in$ $dChase_j(QS_{\C})$, $s'$ $=$ $\_\colon b$ or $p'$ $=$ $\_\colon b$ or 
$o'$ $=$ $\_\colon b$, \text{ for any } $c' \in \C \}$.
 \end{definition}
\noindent Intuitively, origin-contexts for a Skolem blank node $\_\colon b$ is the set of contexts in which triples containing
$\_\colon b$ are first generated, during the dChase construction. Note that there can be multiple contexts in which $\_\colon b$
can simultaneously be generated. By setting $originRuleId(k)$ $=$ $n.d.$, (resp. $originVector(k)$ $=$ $n.d.$, resp. $originContexts(k)$ $=$ $n.d.$,) 
where $n.d.$ is an ad hoc constant,  $\forall k \not \in$ $\bn_{sk}(dChase(QS_{\C}))$, 
we extend the definition of origin ruleId, (resp. origin vector, resp. origin-contexts) to all the constants in the dChase of a quad-system.
\begin{example}\label{eg:descendance}
 Consider the quad-system $\langle Q_{\C}, R \rangle$, where $Q_{\C}=\{c_1\colon(a,b,c)\}$.
  Suppose $R$ is the following set:
  \[  R= \left \{
  \begin{array}{r}
     c_1\colon (x_{11},x_{12},z_1) \rightarrow c_2\colon(x_{11},x_{12},y_1) \hspace{0.2cm} (r_1)\\
     c_2\colon (a,z_2,x_{22}) \rightarrow c_3\colon(a,x_{22},y_2) \hspace{0.85cm} (r_2) \\
     c_2\colon (z_{3},b,x_{32}) \rightarrow c_3\colon(b,x_{32},y_3) \hspace{0.9cm} (r_3) \\
     c_3\colon (a,z_{41},x_{41}), c_3\colon (b,z_{42},x_{42}) \hspace{1.6cm} \  \\
      \rightarrow c_2\colon(y_4, x_{41}, a),c_2\colon(y_4, x_{42}, b)\hspace{.5cm} (r_4)
  \end{array}
  \right \}\]
 Suppose that for brevity quantifiers have been omitted, and variables of the form $y_i$ or $y_{ij}$ are implicitly existentially quantified.
 Iterations during the dChase construction are:
  \vspace{-2pt}
 \begin{eqnarray}
&&   dChase_0(QS_{\C})=\{c_1\text{:}(a,b,c)\} \nonumber \\
&&  dChase_1(QS_{\C})=\{c_1\colon (a,b,c), c_2\colon (a,b,\_\colon b_1)\} \nonumber \\
&&  dChase_2(QS_{\C})=\{c_1\text{:}(a,b,c), c_2\colon(a,b,\_\colon b_1), 
   c_3\colon (a,\_\colon b_1,\_\colon b_2)\} \nonumber \\
&&   dChase_3(QS_{\C})=\{c_1\text{:}(a,b,c), c_2\colon(a,b,\_\colon b_1), c_3\colon (a,\_\colon b_1,\_\colon b_2),  \nonumber \\
&&   c_3\colon (b,\_\colon b_1,\_\colon b_3)\} \nonumber \\ 
&&        dChase_4(QS_{\C})=\{c_1\text{:}(a,b,c), c_2\colon(a,b,\_\colon b_1), c_3\colon (a,\_\colon b_1,\_\colon b_2),  c_3\colon (b,\_\colon b_1, \nonumber \\
&&   \_\colon b_3),   c_2\colon ( \_\colon b_4,\_\colon b_2,a), c_2\colon (\_\colon b_4,\_\colon b_3,b)\} \hspace{1.5cm}\nonumber \\
&&  dChase_5(QS_{\C})=dChase_4(QS_{\C}), \hspace{1.7cm} \nonumber
 \end{eqnarray}
 Also note: \\ 
 $originRuleId($ $\_\colon b_1)$ $=$ $1$, $originRuleId(\_\colon b_2)$ $=$ $2$, $originRuleId(\_\colon b_3)$ $=$
 $3$, $originRuleId($ $\_\colon b_4)$ $=$ $4$, \\
 $originVector($ $\_\text{ :} b_1)$ $=$ $\langle a$, $b\rangle$, $originVector(\_\text{ :} b_2)$ $=$ $originVector(\_\text{ :} b_3)$ $=$
 $\langle \_\colon b_1 \rangle$, $originVector(\_\text{ :} b_4)$ $=$ $\langle \_\text{ :} b_2$, $\_\text{ :} b_3\rangle$,
 \\
also $originContexts(\_\text{ :}b_1)$ $=$ $\{c_2\}$, $originConte$- $xts(\_\colon b_2)$ $=$ 
$originContexts($ $\_\colon b_3)$ $=$ $\{c_3\}$, 
$origin$- $Contexts(\_\colon b_4)$ $=$ $\{c_2\}$, \\
also $\_\colon b_1$ descendantOf $\_\colon b_3$, $\_\colon b_1$ descendantOf $\_\colon b_2$, 
$\_\colon b_2$ descendantOf $\_\colon b_4$,
$\_\colon b_3$ descendantOf $\_\colon b_4$, $\_\colon b_1$ descendantOf $\_\colon b_4$. 
\end{example}
\noindent  For any Skolem blank node $\_\colon b$ (in dChase), its descendant hierarchy can be analyzed using a \emph{descendance graph} 
$\langle V, E, \lambda_r, \lambda_v, \lambda_c\rangle$, which is a labeled graph rooted at $\_\colon b$, 
whose set of nodes $V$ are constants in the dChase, the set of edges $E$ is such that $(k, k') \in E$, iff $k'$ is a descendant of $k$. 
$\lambda_r$,  $\lambda_v$, $\lambda_c$
are node labeling functions, such that $\lambda_r(k)$ $=$ $originRuleId(k)$, $\lambda_v(k)$ $=$ $originVector(k)$,
 and  $\lambda_c(k)$ $=$ $originContexts(k)$, for any $k \in V$. 
The descendance graph for $\_\text{ :} b_4$ of Example~\ref{eg:descendance} is shown in Fig.\ref{fig:descendance2}.  
\begin{figure}[t]
\centering
\begin{tikzpicture}
[place/.style={circle,draw=black!50,fill=white!20,thick,
inner sep=0pt,minimum size=8mm}
]
\node(b4)[place,,label=above:$4\mathpunct{,}\langle\_\text{:}b_2\mathpunct{,}\_\text{:}b_3\rangle\mathpunct{,}\{c_2\}$] at ( 0,2) [place] {$\_\text{:}b_4$};
\node(b3) at (-1.5,0) [place, label=below:$
\hspace{-0.5cm}
\begin{array}{c}
\hspace{1cm} 3\mathpunct{,}\langle \_\text{:}b_1\rangle\mathpunct{,} \\
\hspace{0.25cm}\{c_3\}
\end{array}
$] {$\_\text{:}b_3$};
\node(b2)[place,,label=below:$
\begin{array}{c}
\hspace{-0.25cm} 2\mathpunct{,}\langle \_\text{:}b_1 \rangle \mathpunct{,} \\
\{c_3\}
\end{array}
$] at ( 1.5,0) [place] {$\_\text{:}b_2$};
\node(b1)[place,,label=below:$
\begin{array}{c}
1\mathpunct{,}\langle a\mathpunct{,} b \rangle \mathpunct{,}\\
\{c_2\}
\end{array}
$] at ( 0,-2) [place] {$\_\text{:}b_1$};
\node(a)[place,,label=below:] at ( -2,-4) [place] {$a$};
\node(b)[place,,label=below:] at ( 2,-4) [place] {$b$};
\draw [->] (b4.south) to node[auto,swap]{} (b3.north);
\draw [->] (b1.east) to  node[auto,swap]{} (b.west);
\draw [->] (b1.west) to  node[auto,swap]{} (a.east);
\draw [->] (b2.west) to  node[auto,swap]{} (b1.north);
\draw [->] (b2.east) to[out=290,in=80]  node[auto,swap]{} (b.north);
\draw [->] (b3.east) to  node[auto,swap]{} (b1.north);
\draw [->] (b3.west) to[out=250,in=100]  node[auto,swap]{} (a.north);
\draw [->] (b4.south) to node[auto,swap]{} (b1.north);
\draw [->] (b4.south) to node[auto,swap]{} (b2.north);
\draw [->] (b4.east) to[out=-20,in=45] node[auto]{} (b.east);
\draw [->] (b4.west) to[out=200,in=135] node[auto]{} (a.west);
\end{tikzpicture}
\caption{descendance graph of $\_\text{ :}b_4$ in example \ref{eg:descendance}. Note: n.d. labels not shown}
\label{fig:descendance2}
\end{figure}
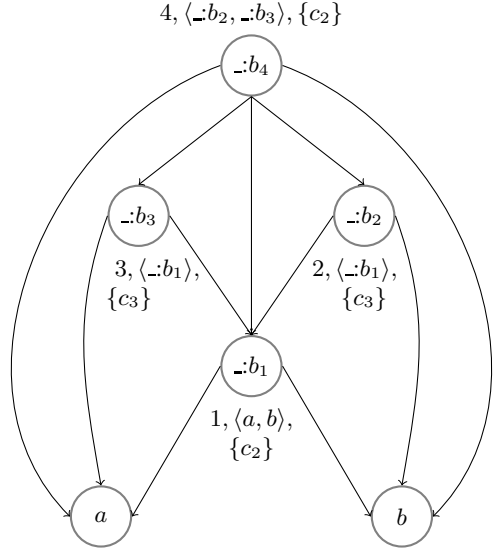
For any two vectors of constants $\vec v, \vec w$, we note $\vec v$ $\cong$ $\vec w$, iff 
there exists a bijection $\mu\colon\bn(\vec v)$ $\rightarrow$ $\bn(\vec w)$ such that 
$\vec w=\vec v[\mu]$.
\begin{definition}[safe, msafe, csafe quad-systems] A quad-system $QS_{\C}$ is said to be \emph{unsafe} (resp. \emph{unmsafe}, resp. \emph{uncsafe}), iff there 
exist Skolem blank nodes $\_\colon b$ $\neq$ $\_\colon b'$ in
$dChase(QS_{\C})$ such that $\_\colon b$ is a descendant of $\_\colon b'$, with $originRuleId(\_\colon b)$ $=$ $originRuleId(\_\colon b')$ 
and $originVector(\_\colon b)$ $\cong$ $originVector(\_\colon b')$
(resp. $origin$- $RuleId(\_\colon b)$ $=$ $originRuleId(\_\colon b')$, resp. $originContexts(\_\colon b)$ $=$ $originConte$- $xts(\_\colon b')$). 
A quad-system is \emph{safe} (resp. \emph{msafe}, resp. \emph{csafe}) iff it is not unsafe (resp. unmsafe, resp. uncsafe).
\end{definition}
\noindent Intuitively, safe, msafe and csafe quad-systems,  does not allow
 repetitive  generation of Skolem blank-nodes with a certain set of attributes in its dChase.  
The containment relation between the class of safe, msafe, and csafe quad-systems are established by the following 
theorem:
\begin{theorem}\label{theorem:StrictContainment}
 Let $\textsc{safe}, \textsc{msafe}$, and $\textsc{csafe}$ denote the class 
 of safe, msafe, and csafe quad-systems, respectively, then the following holds:
 \[
\textsc{csafe} \subset  \textsc{msafe} \subset \textsc{safe} 
 \]
\end{theorem}
\begin{proof}
We first show \textsc{msafe} $\subseteq$ \textsc{safe}, by showing the inverse inclusion of their compliments, i.e.
 $\textsc{unsafe} \subseteq \textsc{unmsafe}$. Suppose a given quad-system $QS_{\C}$ is unsafe, then by definition
 its dChase contains two distinct Skolem blank nodes $\_\colon b$, $\_\colon b'$ such that $\_\colon b$ is a descendant of $\_\colon b'$, with
  $originRuleId(\_\colon b)$ $=$  $originRuleId(\_\colon b')$ and  $originVector($ $\_\colon b)$ $\cong$  $originVector($ $\_\colon b')$.
 But this will imply that 
 $originRuleId($ $\_\colon b)$ $=$  $originRuleId($ $\_\colon b')$. Hence, by definition, $QS_{\C}$ is unmsafe. 
 Hence $\textsc{unsafe} \subseteq \textsc{unmsafe}$ ($\dagger$).
 
 Now, we show that \textsc{csafe} $\subseteq$ \textsc{msafe} by showing
 $\textsc{unmsafe} \subseteq \textsc{uncsafe}$. Suppose a given quad-system $QS_{\C}=\langle Q_{\C}, R\rangle$ is unmsafe, then by definition
 its dChase contains two distinct Skolem blank nodes $\_\colon b$, $\_\colon b'$ such that $\_\colon b$ is a descendant of $\_\colon b'$, with
 $originRuleId(\_\colon b)$ $=$ $originRuleId(\_\colon b')$. But this implies that there exists a BR $r_i$ $=$ 
 $body(r_i)(\vec x$, $\vec z)$ $\rightarrow$ $head(r_i)(\vec x$, $\vec y)$, assignment $\mu$, (resp. $\mu'$,) s.t.
 $\_\colon b$ (resp. $\_\colon b'$) was generated in $dChase(QS_{\C})$ as result of 
 application of $\mu$ (resp. $\mu'$) on $r_i$. That is $\_\colon b$ $=$ $y_j[\mu^{ext(\vec y)}]$, and
 $\_\colon b'$ $=$ $y_k[\mu'^{ext(\vec y)}]$, where $y_j, y_k \in \{\vec y\}$. We have the following two
 subcases (i) $j=k$, (ii) $j\neq k$.
 Suppose (i) $j=k$, then it immediately follows that $originContexts(\_\colon b)$ $=$ $originContexts(\_\colon b')$. Hence, $QS_{\C}$ is uncsafe.
Suppose (ii) $j\neq k$, then
 by construction of dChase, on application 
 of $\mu'$ to $r_i$,  along with $\_\colon b'$,  there gets also generated a Skolem blank node
 $\_\colon b''$ $=$ $y_j[\mu'^{ext(\vec y)}]$, with $y_j \in \{\vec y\}$. Since $\_\colon b$ and $\_\colon b''$ are 
 generated by substitutions of the same variable $y_j \in \{\vec y\}$ of BR $r_i$, $originContexts(\_\colon b)$ $=$ $originContexts(\_\colon b'')$. 
 Also considering that childOf$(\_\colon b')$ $=$ childOf$(\_\colon b'')$ $=$ $\{\vec x[\mu'^{ext(\vec y)}]\}$,  
we can deduce that $\_\colon b$ is a descendant of $\_\colon b''$. Hence, 
 by definition, it holds that $QS_{\C}$ is uncsafe. 
 Hence $\textsc{unmsafe} \subseteq \textsc{uncsafe}$ ($\ddagger$). 
 
  From $\dagger$ and $\ddagger$, it follows that 
\textsc{csafe} $\subseteq$  \textsc{msafe} $\subseteq$ \textsc{safe}. 
 To show that the containments are strict, consider the quad-system $QS_{\C}$ in example~\ref{eg:descendance}. 
 By definition, $QS_{\C}$ is msafe, however uncsafe,
as the Skolem blank nodes $\_\colon b_1$, $\_\colon b_4$, which have the same origin contexts are s.t.
$\_\colon b_1$ is a descendant of $\_\colon b_4$. Hence, \textsc{csafe} $\subset$  \textsc{msafe}.
For \textsc{msafe} $\subset$  \textsc{safe}, the following example shows an instance of a quad-system that is
unmsafe, yet is safe. 
\end{proof}
 \begin{example}\label{eg:strictContainmentMsafeToSafe}  
 Consider the quad-system $QS_{\C}$ $=$ $\langle Q_{\C}$, $R\rangle$,
where $Q_{\C}$ $=$ $\{c_1\colon (a$, $b$, $c)$, $c_2\colon (c$, $d$, $e)\}$, $R$ is given by:
\begin{eqnarray}
     c_1\colon (x_{11},x_{12},x_{13}), c_2\colon(x_{13}, x_{14}, z_{1}) \rightarrow c_3\colon(y_1, 
     x_{11}, x_{12}), c_4\colon(x_{12},x_{13},x_{14}) \hspace{.4cm} (r_1) \nonumber \\
          c_3\colon (x_{21},a,x_{22}), c_4\colon(x_{22}, x_{23}, x_{24}) \rightarrow c_1\colon(x_{21}, a, 
     x_{22}), c_2\colon(x_{22}, x_{23}, x_{24}) \hspace{.4cm} (r_2) \nonumber \\
                c_3\colon (x_{21}, x_{22}, a), c_4\colon(a, x_{23}, x_{24}) \rightarrow c_1\colon(x_{21},  
     x_{22}, a), c_2\colon(a, x_{23}, x_{24}) \hspace{.3cm} (r_3) \nonumber \\
               c_3\colon (x_{21}, x_{22}, x_{23}), c_4\colon(x_{23}, a, x_{24}) \rightarrow c_1\colon(x_{21},  
     x_{22}, x_{23}), c_2\colon(x_{23}, a, x_{24}) \hspace{.3cm} (r_4) \nonumber 
 \end{eqnarray}
 \begin{eqnarray}
        c_3\colon (x_{21}, x_{22}, x_{23}), c_4\colon(x_{23}, x_{24}, a) \rightarrow c_1\colon(x_{21},  
     x_{22}, x_{23}), c_2\colon(x_{23}, x_{24}, a) \hspace{.3cm} (r_5) \nonumber
\end{eqnarray}
Note that for brevity quantifiers have been omitted, and variables of the form $y_i$ or $y_{ij}$ are implicitly existentially quantified.
 Iterations during the dChase construction are:
  \vspace{-2pt}
 \begin{eqnarray}
&&  dChase_0(QS_{\C})=\{c_1\text{:}(a,b,c), c_2\text{:}(c,d,e)\} \nonumber \\
&&  dChase_1(QS_{\C})= dChase_0(QS_{\C}) \cup   \{c_3\colon (\_\colon b_1,  
 a, b), c_4\colon (b,c,d)\} \nonumber \\
&&      dChase_2(QS_{\C})=dChase_1(QS_{\C}) \cup   \{c_1\colon (\_\colon b_1,  
 a, b), c_2\colon (b,c,d)\} \nonumber \\
&&    dChase_3(QS_{\C})= dChase_2(QS_{\C}) \cup   \{c_3\colon (\_\colon b_2,  
 \_\colon b_1, a), c_4\colon (a,b,c)\} \nonumber \\
&&      dChase_4(QS_{\C})=dChase_3(QS_{\C}) \cup   \{c_1\colon (\_\colon b_2,  
 \_\colon b_1, a), c_2\colon (a,b,c)\} \nonumber \\
&&      dChase_5(QS_{\C})= dChase_4(QS_{\C}) \cup   \{c_3\colon (\_\colon b_3,  
 \_\colon b_2, \_\colon b_1), c_4\colon (\_\colon b_1,a,b)\} \nonumber \\
&&       dChase_6(QS_{\C})=dChase_5(QS_{\C}) \cup   \{c_1\colon (\_\colon b_3,  
 \_\colon b_2, \_\colon b_1), c_2\colon (\_\colon b_1,a,b)\} \nonumber \\
&&      dChase_7(QS_{\C})= dChase_6(QS_{\C}) \cup   \{c_3\colon (\_\colon b_4,  
 \_\colon b_3, \_\colon b_2), c_4\colon (\_\colon b_2, \_\colon b_1, a)\} \nonumber \\ 
&&     dChase_8(QS_{\C})=dChase_7(QS_{\C}) \cup    \{c_1\colon (\_\colon b_4,  
 \_\colon b_3, \_\colon b_2), c_2\colon (\_\colon b_2, \_\colon b_1, a)\} \nonumber \\
&&   dChase_9(QS_{\C})= dChase_8(QS_{\C}) \cup   \{c_3\colon (\_\colon b_5,  
 \_\colon b_4, \_\colon b_3), c_4\colon (\_\colon b_3, \_\colon b_2, \nonumber \\
 &&  \_\colon b_1)\} \nonumber \\
&&   dChase(QS_{\C})=  dChase_9(QS_{\C}) \nonumber 
 \end{eqnarray}
 It can be seen that $\_\colon b_1$, $\_\colon b_2$, $\_\colon b_3$, $\_\colon b_4$, $\_\colon b_5$ form 
 a descendant chain, since $\_\colon b_i$ descendantOf $\_\colon b_{i+1}$, for each $i=1, \ldots, 4$.
 Also, $originRuleId(\_\colon b_i)$ $=$ $originRule$- $Id(\_\colon b_{i+1})$, for each $i=1, \ldots, 4$. 
 Hence it turns out that $QS_{\C}$ is unmsafe. However, it can be seen that $originVector(\_\colon b_1)$ $=$ $\langle a$, $b$, $c$, $d \rangle$, 
and $originVector(\_\colon b_2)$ $=$ $\langle \_\colon b_1$, $a$, $b$, $c \rangle$, and
 $originVector(\_\colon b_3)$ $=$ $\langle \_\colon b_2$, $\_\colon b_1$, $a$, $b\rangle$, and
 $originVector(\_\colon b_4)$ $=$ $\langle \_\colon b_3$, $\_\colon b_2$, $\_\colon b_1$, $a\rangle$,
and $originVector(\_\colon b_5)$ $=$ $\langle \_\colon b_4$, $\_\colon b_3$, $\_\colon b_2$, $\_\colon b_1\rangle$, and 
 $originVector(\_\colon b_i) \not \cong  originVector(\_\colon b_j)$, for $1 \leq$ $i$ $\neq$ $j$ $\leq$ $5$, 
 and hence, by definition, $QS_{\C}$ is safe with a terminating dChase. It can be noticed that 
 during each distinct application of $r_1$, the vector of constants bound to the vector of 
 variables $\langle x_{11}, \ldots, x_{14} \rangle$ are different w.r.t $\cong$. 
 Safe quad-systems in this way are capable of recognizing such positive cases of finite dChases (which are classified as
negative cases by msafe quad-systems) by also keeping track of the origin vectors of Skolem blank-nodes in their dChases. 
 \end{example}
\noindent The following property shows that for a safe quad-system, the descendance graph of any Skolem blank node in its dChase 
is a directed acyclic graph (DAG):
\begin{property}[DAG property]
 For a safe (csafe, msafe) quad-system $QS_{\C}$, and for any blank node $b \in \bn_{sk}(dChase(QS_{\C}))$, its descendance graph is
 a DAG.
\end{property}
\begin{proof}
  By construction, as there exists no descendant for any constant $k \in \const(QS_{\C})$, there cannot be any out-going
 edge from any such $k$. Hence, no member of $\const(QS_{\C})$ can be involved in cycles. Therefore, the only members that
 can be involved in cycles are the members of $\const(dChase(QS_{\C}))-\const(QS_{\C})$ $=$ $\bn_{sk}(dChase(QS_{\C}))$. But if there
 exists $\_\colon b \in \bn_{sk}(dChase(QS_{\C}))$, such that there exists a cycle through $\_\colon b$, then this 
 implies that $\_\colon b$ is a descendant of  $\_\colon b$. This would violate the prerequisites of being safe (resp. csafe, resp. msafe), 
 and imply that $QS_{\C}$ is unsafe (resp. uncsafe, resp. unmsafe), which is a contradiction.
\end{proof}

\begin{algorithm}
\scriptsize
\SetKwFunction{unRavel}{UnRavel}

\SetKwData{True}{True}
\SetKwData{dChase}{dChase}
\SetKwData{CQ}{CQ}
\SetKwData{False}{False}
\SetKwInOut{Input}{Input}\SetKwInOut{Output}{Output}
 \unRavel(Descendance Graph $G$)\\
\tcc{procedure to unravel, a descendance graph into a tree}
\Input{descendance graph $G=\langle V,E, \lambda_r, \lambda_v, \lambda_c \rangle$ }
\Output{A labeled Tree $G$}
\Begin{
$G=\langle V,E, \lambda_r, \lambda_v, \lambda_c \rangle$ $:=$ RemoveTranstiveEdges($G$)\;
\ForEach{Node $v_o \in$ preOrder($G)$}{
\If{($k=$ indegree($v_o$)$)>1$}{
$\{v_1,...,v_k\}:=$getFreshNodes();\tcc{each $v_i\not \in V$ is fresh}
\tcc{replace old node $v_o$ by the fresh nodes in $V$}
removeNodeFrom($v_o, V$)\; 
addNodesTo($\{v_1,...,v_k\}, V$)\;

\ForEach{$(v_o,v')\in E$}{\tcc{replace each outgoing edge from $v_o$ with a fresh outgoing edges from each fresh node $v_i$}
removeEdgeFrom($(v_o, v'), E$)\;
addEdgesTo($\{(v_1,v'),...,(v_k,v')\}, E$);
}
$i:=1$\;
\ForEach{$(v',v_o)\in E$}{\tcc{replace each incoming edge of $v_o$ with an incoming edge for a unique $v_i$}
removeEdgeFrom($(v',v_o), E$)\;
addEdgeTo($(v',v_i),E$)\;
$i$++;
}
}
}
\tcc{restrict node labels to the updated set of nodes in $V$}
 $\lambda_r:=\lambda_r|_{V}$, 
 $\lambda_v:=\lambda_v|_{V}$, 
 $\lambda_c:=\lambda_c|_{V}$\;
\Return $G$;
}
\caption{}
\label{alg:unRavel}
\end{algorithm}
\begin{figure}[t]
\centering
\begin{tikzpicture}
[place/.style={circle,draw=black!50,fill=white!20,thick,
inner sep=0pt,minimum size=8mm}
]
\node(b4)[place,,label=above:$4\mathpunct{,}\langle\_\text{:}b_2\mathpunct{,}\_\text{:}b_3\rangle\mathpunct{,}\{c_2\}$] at ( 0,2) [place] {$\_\text{:}b_4$};

\node(b3) at (-0.85,0) [place, label=below:$
\hspace{0.2cm}
\begin{array}{c}
3\mathpunct{,}\langle \_\text{:}b_1\rangle\mathpunct{,}\\
\{c_3\}
\end{array}
$] {$\_\text{:}b_3$};

\node(b2) at (0.85,0) [place, label=below:$
\hspace{0.2cm}
\begin{array}{c}
2\mathpunct{,}\langle \_\text{:}b_1\rangle\mathpunct{,}\\
\{c_3\}
\end{array}
$] {$\_\text{:}b_2$};

\node(b1)[place,,label=below:$
\begin{array}{c}
1\mathpunct{,}\langle a\mathpunct{,} b \rangle \mathpunct{,}\\
\{c_2\}
\end{array}
$] at ( 2,-2) [place] {$\_\text{:}b_1$};

\node(b1a)[place,,label=below:$
\begin{array}{c}
1\mathpunct{,}\langle a\mathpunct{,} b \rangle \mathpunct{,}\\
\{c_2\}
\end{array}
$] at ( -2,-2) [place] {$\_\text{:}b_1$};

\node(a)[place,,label=below:] at ( 1,-4) [place] {$a$};
\node(b)[place,,label=below:] at ( 3,-4) [place] {$b$};
\node(aa)[place,,label=below:] at ( -3,-4) [place] {$a$};
\node(bb)[place,,label=below:] at ( -1,-4) [place] {$b$};

\draw [->] (b4.south) to node[auto,swap]{} (b3.north);
\draw [->] (b1.east) to  node[auto,swap]{} (b.north);
\draw [->] (b1.west) to  node[auto,swap]{} (a.north);
\draw [->] (b2.east) to  node[auto,swap]{} (b1.north);
\draw [->] (b4.south) to node[auto,swap]{} (b2.north);
\draw [->] (b3.west) to  node[auto,swap]{} (b1a.north);
\draw [->] (b1a.east) to  node[auto,swap]{} (bb.north);
\draw [->] (b1a.west) to  node[auto,swap]{} (aa.north);
\end{tikzpicture}
\caption{Descendance graph of Fig. $\ref{fig:descendance2}$ unraveled into a tree. Note: n.d. labels are not shown}
\label{fig:descendanceTree}
\end{figure}
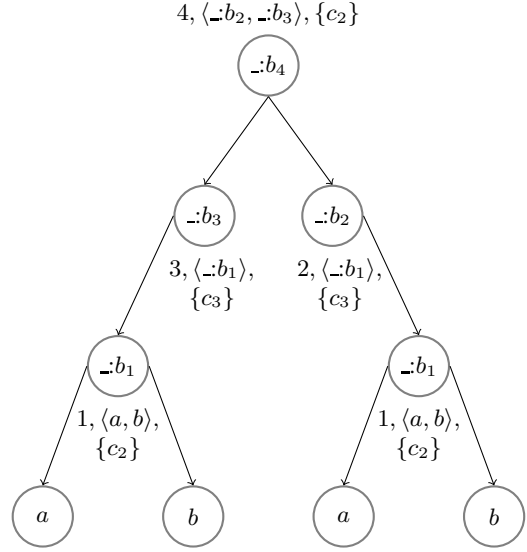

\noindent Since the descendance graph $G$ of any Skolem blank node $\_\colon b \in \bn_{sk}(dChase(QS_{\C}))$ is such that 
$G$ is rooted at $\_\colon b$ and is acyclic, any directed path from $\_\colon b$ terminates at some node. Hence, one can use a tree traversal technique, such as 
preorder (visit a node first and then its children) to sequentially traverse nodes in $G$. 
Algorithm \ref{alg:unRavel} takes a descendance graph 
$G$ and unravels it into a tree. The algorithm first removes all the transitive edges from $G$, i.e. if there are $v,v' \in V$ with $(v,v') \in E$ and 
$G$ contains a path of length greater than 1 from $v$ to $v'$, then it removes $(v,v')$. Note that, in the resulting graph, 
the presence of a path from $v$ to $v''$ still gives us the information that $v''$ is a descendant of $v$. 
The algorithm then traverses the graph in preorder fashion, as it encounters a node $v$, if $v$ has 
an indegree $k$ greater than one, it replaces $v$ with $k$ fresh nodes $v_1,...,v_k$, and distributes the set of edges 
incident to $v$ across $v_1,...,v_k$, such that (i) each $v_i$ has at-most one incoming edge (ii) all the edges incident to $v$
are incident to some $v_i$, $i \in \{1,\ldots, k\}$. Outgoing edges of $v$ are copied for each $v_i$. Hence,
after the above operation each $v_i$  has an indegree $1$, whereas outdegree of $v_i$ is same as the outdegree of $v$, 
$i \in \{1,\ldots, k\}$. Hence, after all the nodes are visited, every node except the root in the new graph $G$ has an indegree $1$. 
$G$ is still rooted, connected, acyclic,
and is hence a tree. The algorithm terminates as there are no cycles in the graph, and at some point reaches a node with no children.
For instance, the unraveling of the descendance graph of $\_\text{ :}b_4$ in
Fig. \ref{fig:descendance2} is shown in Fig. \ref{fig:descendanceTree}.
The following property holds for any Skolem blank node of a safe quad-system.
\begin{property} \label{property:descendanceTree}
For a safe quad-system $QS_{\C}$ $=$ $\langle Q_{\C}, R\rangle$, and any Skolem blank node
in $dChase(QS_{\C})$, the unraveling (Algorithm \ref{alg:unRavel}) of its descendance graph 
results in a tree $t$ $=$ $\langle V$, $E$, $\lambda_r$, $\lambda_v$, $\lambda_c\rangle$ s.t.:
 \begin{enumerate}
  \item any leaf node of $t$ is from the set $\const(QS_{\C})$,
  \item any non-leaf node of $t$ is from  $\bn_{sk}(dChase(QS_{\C}))$,
  \item $order(t)$ $\leq$ $w$, where $w$ $=$ $max_{r\in R} |\mathit{fr}(r)|$,
    \item there cannot be a path between $b \neq b'\in V$, with $\lambda_r(b)$ $=$ $\lambda_r(b')$
  and $\lambda_v(b) \cong \lambda_v(b')$, 
     \item there cannot be a path between $b \neq b'\in V$, with $\lambda_r(b)$ $=$ $\lambda_r(b')$, if $QS_{\C}$ is also msafe,
    \item there cannot be a path between $b \neq b'\in V$, with $\lambda_c(b)$ $=$ $\lambda_c(b')$, if $QS_{\C}$ is also csafe. 
 \end{enumerate}
\end{property}
\begin{proof}
\begin{enumerate}
 \item 
 Any node $n$ in the descendance graph is such that $n$ $\in$ $\const(dChase(QS_{\C}))$, and $\const(dChas$- $e(QS_{\C}))$ $=$
 $\const(QS_{\C})$ $\uplus$ $\bn_{sk}(dChase(QS_{\C}))$. Since any member $m \in \bn_{sk}(dChase(QS_{\C}))$ is generated from an application
 of a BR with an assignment $\mu$ such that its frontier variables are assigned by $\mu$ with a set of constants, 
 $m$ has at-least one child. But, since $n$ is a leaf node, $n \in \const(QS_{\C})$.
 
 \item Since no member $m \in \const(QS_{\C})$ can have descendants and any non-leaf node has children, 
 $m$ cannot be a non-leaf node. Hence, non-leaf nodes must be from $\bn_{sk}(dChase($ $QS_{\C}))$.
 
 \item The order of $t$ is the maximal outdegree among the nodes of $t$, 
 and outdegree of a node is the number of children it has. Since any node in $t$ with non-zero outdegree 
 is a Skolem blank-node $\_\colon b$ generated by 
 application of an assignment $\mu$ to $r$ $=$ $body(r)(\vec x$, $\vec z)$ $\rightarrow head(r)(\vec x$, $\vec y)\in R$,
 the number of children $\_\colon b$ has equals $\|\vec x\|$. Hence the order of $t$ is bounded by $w$. 
 
 \item Since any path from $b$ to $b'$ implies that $b'$ is a descendant of $b$, it must be the case that
 $\lambda_r(b) \neq \lambda_r(b')$ or $\lambda_v(b) \not \cong \lambda_v(b')$,  
 otherwise safety condition  would be violated.
 \item Similar as above, immediate by definition.
 \item Similar as above, immediate by definition.
 \end{enumerate}
\end{proof}
\noindent The property above is exploited to show that there exists a finite bound in the dChase size and its 
computation time. 
\begin{lemma}\label{lemma:safeChaseComputationTime}
 For any safe/msafe/csafe quad-system $QS_{\C}$ $=$ $\langle Q_{\C}$, $R\rangle$, the following holds: 
 (i) the dChase size $\|dChase(QS_{\C})\|$ $=$ $\bigO(2^{2^{\|QS_{\C}\|}})$,  
 (ii) $dChase(QS_{\C})$ can be computed in \textsc{2EXPTIME},
 (iii) if $\|R\|$ and the set of schema triples in $Q_{\C}$ is fixed
 to a constant, then $\|dChase(QS_{\C})\|$ is a polynomial in $\|QS_{\C}\|$ and
can be computed in \textsc{PTIME}.
\end{lemma}
\begin{proof}
The proofs are provided for safe quad-systems, but since \textsc{csafe} $\subset$ \textsc{msafe} $\subset$ \textsc{safe} and since we 
are giving upper bounds, they also propagate trivially to msafe and csafe quad-systems.

(i)
For any Skolem blank node in $dChase(QS_{\C})$,  the size of its originVector is upper bounded
by $w$ $=$ $max_{r \in R} |\mathit{fr}(r)|$. If $S$ is the set of all origin vectors of blank-nodes in $dChase(QS_{\C})$, then cardinality of the set 
$S'$ $=$ $S\setminus$ $\cong$ is upper bounded by $(|\uri(QS_{\C})|$ $+$ $|\lit(QS_{\C})|+w)^w$, which means that 
$|S'|$ $=$ $\bigO(2^{\|QS_{\C}\|})$. Also, since the set of origin ruleId labels, $Rids$, can at most be $|R|$, 
the cardinality of the set $Rids\times S'$ $=$ $\bigO(2^{\|QS_{\C}\|})$.   
 For the descendance tree $t$ of any Skolem blank node of $dChase(QS_{\C})$, 
since there cannot be paths in $t$ between distinct $b$ and $b'$, such that $originRuleId(b)$ $=$ $originRuleId(b')$
and $originVector(b)$ $\cong$ $originVector(b')$, the length of any such path is upper bounded by 
$|Rids\times S'|$ $=$ $\bigO(2^{\|QS_{\C}\|})$. However, it turns out that the above upper bound provided is loose, as there is the need of
additional \emph{filter} BRs to transform/back-propagate  vectors of constants associated with Skolem blank nodes generated by repetitive 
application of the same BR. For instance, consider the set of BRs in eg: \ref{eg:strictContainmentMsafeToSafe}. The BR $r_1$ transforms
the origin vector to a new vector each time during its application. BRs $r_2$ - $r_5$ deals with back propagation of these vectors
back to input origin vectors of BR $r_1$. Such filter BRs  
rule out the case of a BR being applied to a quad that contains a Skolem blank node that was generated using the same BR on
an isomorphic origin vector, ensuring that the safety criteria for Skolem blank-nodes generated is not violated. 
It turns out that the number of such filter BRs required is polynomial w.r.t. to 
the number of descendants with the same rule id, for a node in $t$. Hence, it turns out the 
depth of $t$ is polynomially bounded by $\|R\|$. (Note that depth of $t$ is bounded by $|R|$ for msafe quad-systems. Also
since, the set of origin context labels are bounded by the set of existential variables in $R$, depth of $t$ is bounded by $\|R\|$
for csafe quad-systems.)
 Also order of the tree is bounded by $w$. 
 Hence, any such tree can have at most
$\bigO(2^{\|QS_{\C}\|})$ leaf nodes, $\bigO(2^{\|QS_{\C}\|})$ inner nodes, and  $\bigO(2^{\|QS_{\C}\|})$ nodes.  
Since each of the leaf nodes can only be from
 $\const(QS_{\C})$ and each of the inner nodes  
 correspond to an existential variable in $R$, the number of such possible trees are 
 clearly bounded double exponentially in $\|QS_{\C}\|$, hence bounds the number of
Skolem blank nodes generated in the dChase.

(ii) From (i) $\|dChase(QS_{\C})\|$ is double exponential in $\|QS_{\C}\|$, and since
each iteration add at-least one quad to its dChase, the number of iterations are bounded double exponentially in
$\|QS_{\C}\|$. Also, by Lemma \ref{lemma:chaseSizeIncrease} any iteration $i$ can be done in
time $\bigO(\|dChase_{i-1}(QS_{\C})\|^{\|R\|})$. Hence, by using (i), we get 
$\|dChase_{i-1}(QS_{\C})\|$ $=$ $
\bigO(2^{2^{\|QS_{\C}\|}})$. Hence, we can infer that each iteration $i$ can be done in time 
$\bigO(2^{\|R\|*2^{\|QS_{\C}\|}})$. Also since 
the number of iterations is at most double exponential, computing $dChase(QS_{\C})$ is in 2EXPTIME.

(iii) Since $\|R\|$ is fixed to a constant, the set of  
existential variables is also a constant. 
In this case, since the size of the frontier of any $r\in R$ is also a constant, the order and depth of any descendant tree $t$ of a Skolem blank node
is a constant.  Hence, the number of (leaf) nodes of
 $t$ is bounded by a constant.  Also in this setting, 
the label of inner nodes of $t$, which correspond to existential variables, is also a constant, and the leaf nodes of $t$ can only be a constant
in $\const(QS_{\C})$. Hence, the number of descendant trees and consequentially,
the number of Skolem blank nodes generated is bounded by $\bigO(|\const(QS_{\C})|^z)$, where $z$ is a constant.
Hence, the set of constants generated in $dChase(QS_{\C})$ is a polynomial in $\|QS_{\C}\|$, and
so is $\|dChase(QS_{\C})\|$.

Since in any dChase iteration except the final one, at least one quad is added, and also since the final dChase
can have at most $\bigO(\|QS_{\C}\|^z)$ triples, the total number of iterations are bounded by $\bigO(\|QS_{\C}\|^z)$ $(\dagger)$.
By Lemma \ref{lemma:chaseSizeIncrease}, since any iteration $i$ can be computed in $\bigO(\|dChase_{i-1}(QS_{\C})\|^{\|R\|})$
time, and since $\|R\|$ is a constant, the time required for each iteration is  
a polynomial in $\|dChase_{i-1}(QS_{\C})\|$, which is at most a polynomial in $\|QS_{\C}\|$. 
Hence, any dChase iteration can be performed in polynomial time in size of $QS_{\C}$ $(\ddagger)$.
From $(\dagger)$ and $(\ddagger)$, it can be concluded that dChase can be computed in PTIME. 
\end{proof}

\begin{lemma}\label{lemma:safe-computational-properties}
For any safe/msafe/csafe quad-system, the following holds: 
 (i) data complexity of CCQ entailment is in \textsc{PTIME},
 (ii) combined complexity of CCQ entailment is in \textsc{2EXPTIME}. 
\end{lemma}
\begin{proof}
Note that the proofs are provided for safe quad-systems, but since \textsc{csafe} $\subset$ \textsc{msafe} $\subset$ \textsc{safe}
and since we are giving upper bounds, they also propagate trivially to msafe and csafe quad-systems.

Given a safe quad-system $QS_{\C}=\langle Q_{\C}, R\rangle$, since $dChase(QS_{\C})$ is finite, 
a boolean CCQ $CQ()$ can naively be evaluated by binding the set of constants in the dChase to the variables in the $CQ()$,
and then checking if any of these bindings are contained in $dChase(QS_{\C})$. The number of such 
bindings can at most be $\|dChase(QS_{\C})\|^{\|CQ()\|}$ ($\dagger$).

 (i) Since for data complexity, the size of the BRs $\|R\|$, the set of schema triples, and $\|CQ()\|$ is fixed to a constant. 
 From Lemma \ref{lemma:safeChaseComputationTime} (iii), we know that under the above mentioned settings the dChase can be computed
 in PTIME and is polynomial in the size of $QS_{\C}$. Since $\|CQ()\|$ is fixed to a constant, and 
 from ($\dagger$), binding the set of constants in $dChase(QS_{\C})$ on $CQ()$ still gives a number of bindings that is
 worst case polynomial in the size of $\|QS_{\C}\|$. Since membership of these bindings can checked in the polynomially sized 
 dChase in PTIME, the time required for CCQ entailment is in PTIME.
 
 (ii) Since in this case $\|dChase(QS_{\C})\|=\bigO(2^{2^{\|QS_{\C}\|}})$ $(\ddagger)$, from ($\dagger$) and $(\ddagger)$,  
 binding the set of constants in $dChase(QS_{\C})$ to $CQ()$ amounts to $\bigO(2^{\|CQ()\|*2^{\|QS_{\C}\|}})$ number of 
 bindings. Since the dChase is double exponential in $\|QS_{\C}\|$, checking the membership of each of these bindings
 can be done in 2EXPTIME. Hence, the combined complexity is in 2EXPTIME.
\end{proof}
\begin{theorem}\label{thoeorem:safe-computational-properties}
 For any safe/msafe/csafe quad-system, the following holds: 
 (i) The data complexity of CCQ entailment is \textsc{PTIME}-complete (ii) 
The combined complexity of CCQ entailment is  \textsc{2EXPTIME}-complete. 
\end{theorem}
\begin{proof} (i)(Membership) See Lemma \ref{lemma:safe-computational-properties} for the membership in PTIME. 
 
 \noindent(Hardness) Follows from the PTIME-hardness of data complexity of CCQ entailment for Range-Restricted quad-systems 
 (Theorem \ref{theorem:hornDataComplexity}), which are contained in safe/msafe/csafe quad-systems. 
 
 \noindent (ii) (Membership) See Lemma \ref{lemma:safe-computational-properties}.
 
 \noindent (Hardness) See following heading.
\end{proof}
\subsection{2EXPTIME-Hardness of CCQ Entailment}
\noindent In this subsection, we show that the combined complexity of the decision problem of CCQ entailment for context 
acyclic quad-systems is 2EXPTIME-hard. We show this
by reduction of the word-problem of a 2EXPTIME deterministic turing machine (DTM) 
to the CCQ entailment problem. We notify the reader that the technique we follow is, similar to works such as 
\cite{CaliGP2012AIJ,DBLP:conf/rr/CaliGP10}, to iteratively 
generate a doubly exponential number of objects that represent the configurations and cells of the tape of the DTM, 
and then simulate its working by appropriate BRs.
A DTM $M$ is a tuple $M=\langle Q, \Sigma, \Delta, q_0, q_A\rangle$, where 
\begin{itemize}
 \item $Q$ is a set of states,
 \item $\Sigma$ is a finite set of letters that includes the blank symbol $\Box$,
 \item $\Delta\colon (Q \times \Sigma) \rightarrow (Q \times \Sigma \times \{+1, -1\})$ is the transition function,
 \item $q_0 \in Q$ is the initial state.
 \item $q_A \in Q$ is the accepting state.
\end{itemize}
 W.l.o.g. we assume that there exists exactly one accepting state, which is also the lone halting state. 
A  configuration is a word $\vec \alpha \in \Sigma^*Q\Sigma^*$. A
configuration $\vec \alpha_2$ is a successor of the configuration $\vec \alpha_1$, iff one of the following holds:
\begin{enumerate}
 \item $\vec \alpha_1=\vec w_lq\sigma\sigma_r \vec w_r$ and $\vec \alpha_2=\vec w_l\sigma'q'\sigma_r \vec w_r$, 
 if $\Delta(q,\sigma) =(q',\sigma',R) $, or
  \item $\vec \alpha_1=\vec w_lq\sigma$ and $\vec \alpha_2=\vec w_l\sigma'q'\Box$, 
  if $ \Delta(q,\sigma)=(q',\sigma',R)$, or
  \item $\vec \alpha_1=\vec w_l\sigma_lq\sigma \vec w_r$ and $\vec \alpha_2=\vec w_lq'\sigma_l\sigma' \vec w_r$, 
  if $ \Delta(q,\sigma)=(q',\sigma',L)$.
\end{enumerate}
where $q,q' \in Q$, $\sigma,\sigma',\sigma_l, \sigma_r \in \Sigma$, and $\vec w_l, \vec w_r\in \Sigma^*$. 
Since number of configurations can at most be doubly exponential in the size of the input string, and since 2EXPTIME $\subseteq$ 2EXPSPACE,  
the number of tape cells traversed by the DTM tape head is also bounded double exponentially. A configuration $\vec c=\vec w_lq\vec w_r$ is an accepting configuration iff
$q=q_A$. A language $L \subseteq \Sigma^*$ is accepted by a 2EXPTIME bounded DTM $M$, iff 
for every $\vec w \in L$, $M$ accepts $\vec w$ in time $\bigO(2^{2^{\|\vec w\|}})$. 

\paragraph{Simulating DTMs using Safe Quad-Systems}\label{subsection:ATM simulation}
 Consider a double exponential time bounded DTM $M=\langle Q,\Sigma, \Delta, q_0, q_A\rangle$, and a string $\vec w$, with $\|\vec w\|=m$. 
 Suppose that $M$ terminates in $2^{2^n}$ time, where $n$ $=$ $m^k$, $k$ is a constant. 
 In order to simulate $M$,  
we construct a quad-system $QS^M_{\C}=\langle Q^M_{\C}, R\rangle$,
where  $\C=\{c_0,c_1,...,c_n\}$, whose various elements represents the constructs of $M$.  
Let $Q^M_{\C}$ be initialized with the following quads:
\begin{eqnarray}
&& c_0\colon(k_0,\texttt{rdf:type},R),c_0\colon(k_1,\texttt{rdf:type},R), c_0\colon(k_0,\texttt{rdf:type}, min_0), \nonumber \\
&&  c_0\colon(k_1,\texttt{rdf:type},  max_0),  c_0\colon(k_0,succ_0,k_1) \nonumber
\end{eqnarray}
Now for each pair of elements of type $R$ in $c_i$, a Skolem blank-node is generated
in $c_{i+1}$, and hence follows the recurrence relation $r(j+1)=[r(j)]^2$, with seed $r(0)=2$, which after $n$ iterations
yields $2^{2^n}$. In this way, a doubly exponentially long chain of elements is created in 
$c_n$ using the following set of rules:
\begin{eqnarray}
&& c_i\colon(x_0, \texttt{rdf:type},R),c_i\colon(x_1,\texttt{rdf:type},R) \rightarrow  \exists y \ c_{i+1}\colon(x_0,x_1,y),\nonumber \\
&&  
c_{i+1}\colon(y,\texttt{rdf:type},R) \hspace{0.3cm} (eBr)\nonumber 
\end{eqnarray}
The combination of the minimal element with the minimal element (elements of type $min_i$) in $c_i$
create the minimal element in $c_{i+1}$, and similarly the combination of the maximal element with the 
maximal element (elements of type $max_i$) in $c_i$
create the maximal element of $c_{i+1}$:
\begin{eqnarray}
&& c_{i+1}\colon(x_0,x_0,x_1), c_i\colon(x_0,\texttt{rdf:type},min_i) \rightarrow 
 c_{i+1}\colon(x_1,\texttt{rdf:type},min_{i+1}) \nonumber \\
&& c_{i+1}\colon(x_0,x_0,x_1), c_i\colon(x_0,\texttt{rdf:type},max_i) \rightarrow 
 c_{i+1}\colon(x_1,\texttt{rdf:type},max_{i+1}) \nonumber 
\end{eqnarray}
The successor relation $succ_{i+1}$ is created in $c_{i+1}$ using the following set of rules, using the well-known
integer counting technique:
\begin{eqnarray}
&& c_i\colon(x_1,succ_i,x_2), c_{i+1}\colon(x_0,x_1,x_3), 
 c_{i+1}\colon(x_0,x_2,x_4) \rightarrow c_{i+1}\colon(x_3, succ_{i+1},x_4) \nonumber \\
&& c_i\colon(x_1,succ_i,x_2), c_{i+1}\colon(x_1,x_3,x_5),c_{i+1}\colon(x_2, 
 x_4,x_6), c_i\colon(x_3, \texttt{rdf:type}, \nonumber \\
&&  max_i), c_i\colon(x_4, \texttt{rdf:type}, min_i) \rightarrow c_{i+1}\colon(x_5, succ_{i+1},x_6) \nonumber 
\end{eqnarray}
Each of the above set of rules are instantiated for $0\leq i <n$, and in this way after $n$ generating dChase iterations, 
$c_n$ has doubly exponential number of elements of type $R$, that are ordered linearly using the relation $succ_n$.
By virtue of the first rule below, each of the objects representing the cells of the DTM are 
linearly ordered by the relation $succ$. Also the transitive closure of $succ$ is defined as the relation
$succt$
\begin{eqnarray}
&& c_n\colon(x_0, succ_n, x_1) \rightarrow c_n\colon(x_0, succ, x_1) \nonumber \\
&& c_n\colon(x_0, succ, x_1) \rightarrow c_n\colon(x_0, succt,x_1) \nonumber \\
&& c_n\colon(x_0, succt, x_1), c_n\colon(x_1,succt,x_2) 
 \rightarrow c_n\colon(x_0, succt, x_2) \nonumber
\end{eqnarray}
Also using a similar construction, we can reuse the $2^{2^{n-1}}$ linearly ordered elements in $c_{n-1}$
to create another linearly ordered chain of a doubly exponential number of objects in $c_n$ that 
represents configurations of $M$, whose minimal element is of type $conInit$, and the linear order relation being $conSucc$.

 Various triple patterns that are used to encode the possible configurations, runs and their relations in $M$ are:
\begin{description}
 \item[$(x_0,head,x_1)$] denotes the fact that in configuration $x_0$, the head of the DTM is at cell $x_1$.
 \item[$(x_0,state, x_1)$] denotes the fact that in configuration $x_0$, the DTM is in state $x_1$.
 \item[$(x_0,\sigma, x_1)$] where $\sigma \in \Sigma$, denotes the fact that in configuration
 $x_0$, the cell $x_1$ contains $\sigma$.
 \item[$(x_0,succ, x_1)$] denotes the linear order between cells of the tape.
 \item[$(x_0,succt, x_1)$] denotes the transitive closure of $succ$.
 \item[$(x_0,conSucc,x_1)$] to denote the fact that $x_1$ is a successor configuration of $x_0$.
 \item$(x_0,\texttt{rdf:type},Accept)$ denotes the fact that the configuration $x_0$ is an accepting configuration.
\end{description}
Since in our construction, each $\sigma \in \Sigma$ is represented as a relation, we could constrain
that no two letters $\sigma \neq \sigma'$ are on the same cell  using the 
following axiom:
\begin{eqnarray}
  c_n\colon(z_1 ,\sigma, z_2), c_n\colon(z_1, \sigma', z_2) \rightarrow  \nonumber
\end{eqnarray}
for each $\sigma \neq \sigma' \in \Sigma$. Note that the above BR has an empty head, is equivalent to asserting the negation of its body. 

\paragraph{Initialization}
Suppose the initial configuration is $q_0\vec w\Box$, where $\vec w=\sigma_0...\sigma_{n-1}$, then we 
enforce this using the following BRs in our quad-system $QS^M_{\C}$ as:
\begin{eqnarray}
&& c_n\colon(x_0,\texttt{rdf:type}, conInit), c_n\colon(x_1,\texttt{rdf:type}, min_n) \rightarrow c_n\colon(x_0,head,x_1), \nonumber \\
&&  c_n\colon(x_0,state,q_0) \nonumber \\ 
&& c_n\colon(x_0,\texttt{rdf:type}, min_n) \wedge \bigwedge_{i=0}^{n-1} c_n\colon(x_i, succ, x_{i+1}) \ \wedge \ c_n\colon(x_j, \texttt{rdf:type}, \nonumber \\ 
&&   conInit)  \rightarrow 
\bigwedge_{i=0}^{n-1} c_n\colon(x_j, \sigma_i, x_i) \wedge c_n\colon(x_j,\Box, x_n) \nonumber \\
&& c_n\colon (x_j, \texttt{rdf:type}, conInit), c_n\colon(x_j,\Box,x_0), c_n\colon 
 (x_0, succt,x_1) \rightarrow c_n\colon(x_j, \Box, x_1) \nonumber
\end{eqnarray}
 The last BR copies the $\Box$ to every succeeding cell in the initial configuration. 
\paragraph{Transitions}
 For every left transition $\Delta (q,\sigma)$ $=$ $(q_j$, $\sigma'$, $-1)$, the following BR:
 \begin{eqnarray}
 && c_n\colon(x_0, head,x_i), c_n\colon(x_0, \sigma, x_i), c_n\colon(x_0, state, q),  c_n\colon(x_j, succ, x_i),  c_n\colon(x_0, \nonumber \\
 &&   conSucc, x_1) \rightarrow c_n\colon(x_1, 
   head, x_j), c_n\colon(x_1,\sigma',x_i), c_n\colon(x_1, state, q_j) \nonumber 
 \end{eqnarray}
 For every right transition $\Delta (q,\sigma)=(q_j,\sigma',+1) $, the following BR:
 \begin{eqnarray}
 && c_n\colon(x_0, head,x_i), c_n\colon(x_0, \sigma, x_i), c_n\colon(x_0, state, q), c_n\colon(x_i, succ, x_j), c_n\colon(x_0, \nonumber \\
 &&   conSucc, x_1) \rightarrow  c_n\colon(x_1, 
  head, x_j), c_n\colon(x_1,\sigma',x_i), c_n\colon(x_1, state, q_j) \nonumber 
 \end{eqnarray}
 \paragraph{Inertia}
 If in any configuration the head is at cell $i$ of the tape, then in every successor configuration, elements
 in preceding and following cells of $i$ in the tape are retained. The following two BRs ensures this:
 \begin{eqnarray}
  && c_n\colon (x_0, head, x_i), c_n\colon(x_0, conSucc, x_1), c_n\colon(x_j, succt, x_i), c_n\colon(x_0, \sigma, x_j) \nonumber \\
  &&     \rightarrow c_n\colon(x_1, \sigma, x_j) \nonumber \\
 && c_n\colon(x_0, head, x_i), c_n\colon(x_0, conSucc, x_1), c_n\colon(x_i,succt, x_j),  c_n\colon(x_0, \sigma, x_j)  \nonumber \\
  &&  \rightarrow c_n\colon(x_1, \sigma, x_j) \nonumber
 \end{eqnarray}
 The rules above are instantiated for every $\sigma \in \Sigma$.
 \paragraph{Acceptance}
 A configuration whose state is $q_A$ is accepting:
 \begin{eqnarray}
  c_n\colon(x_0, state,q_A) \rightarrow c_n\colon(x_0, \texttt{rdf:type}, Accept) \nonumber
 \end{eqnarray}
If a configuration of accepting type is reached, then it can be back propagated to the initial configuration, using the following BR:
  \begin{eqnarray}
&&   c_n\colon(x_0, conSucc, x_1), c_n\colon(x_1, \texttt{rdf:type}, Accept) 
    \rightarrow c_n\colon(x_0, \texttt{rdf:type}, Accept) \nonumber
 \end{eqnarray}
Finally $M$ accepts $\vec w$ iff the initial configuration is an accepting configuration. Let 
$CQ^M$ be CCQ: $\exists y \ c_n\colon (y$, \texttt{rdf:type}, $conInit)$, $c_n\colon(y$, \texttt{rdf:type}, $Accept)$. 
It can easily be verified that $QS^M_{\C} \models CQ^M$ iff the initial configuration is an accepting configuration. 
In order to prove the soundness and completeness of our simulation, we prove the following claims:
\begin{claim}(1)
 The quad-system $QS^M_{\C}$ in the aforementioned simulation is a csafe quad-system
\end{claim}
It can be noted that the only BRs in which existentials are present are the BRs used to generate the double exponential chain of tape cells and 
configurations, and are of the form (eBr). Note that in each of application of such a BR, a blank-node $\_\colon b$
generated in a context $c_i$, for any $i=1, \ldots, n$, is such that $originContexts(\_\colon b)=\{c_i\}$ and 
has exactly two child blank-nodes, each of whose origin contexts is $\{c_{i-1}\}$. Hence, any Skolem blank-node generated in any $c_i$, for $i=1\ldots n$
is such that its child blank-nodes has origin contexts $c_{i-1}$. Thanks to the above property, it turns out that there exists no two 
blank-nodes $\_\colon b, \_\colon b'$  in the dChase of  $QS^M_{\C}$ such that $\_\colon b$ is a descendant of $\_\colon b'$ and 
$originContexts(\_\colon b)$ $=$ $originContexts(\_\colon b')$. Therefore $QS^M_{\C}$ is csafe. 

\begin{claim}(2)
 $QS^M_{\C} \models CQ^M$ iff $M$ accepts $\vec w$.
\end{claim}
Suppose that $QS^M_{\C}$ $\models$ $CQ^M$, then by Theorem \ref{dChaseUniversalModelProperty},
 there exists an assignment $\mu\colon\var(CQ^M)$ $\rightarrow$ $\const$, with
$CQ^M[\mu]$ $\subseteq$ $dChase(QS_{\C})$. This implies that there exists a constant $o$ in $\const(dChase(QS_{\C}))$, 
with  $\{c_n\colon (o$, \texttt{rdf:type}, $Accept)$, $c_n\colon (o$, \texttt{rdf:type}, $conInit)\}$ $\subseteq$ $dChase(QS_{\C}$.
But thanks to the acceptance axioms it follows that there exists an constant $o'$ such that 
$\{c_n\colon(o$, $conSucc$, $o_1)$, $c_n\colon(o_1$, $conSucc$, $o_2)$, \ldots, $c_n\colon(o_n$, $conSucc$, $o')\}$ $\subseteq$ $dChase(QS_{\C})$,  
and $c_n\colon(o'$, \texttt{rdf:type}, $Accept)$ $\in$ $dChase(QS_{\C})$.
Also thanks to the initialization axioms, it can be seen 
that $o$ represents the initial configuration of $M$ i.e. it represents the configuration in which the initial state is $q_0$, 
and the left end of the read-write tape contains $\vec w$ followed by trailing $\Box$s, with the read-write head positioned
at the first cell of the tape. Also the transition axioms makes sure that if $c_n\colon (o, conSucc, o'')$ $\in$ $dChase(QS_{\C})$, then
$o''$ represents a successor configuration of $o$. That is, if $o$ represents the configuration in which
$M$ is at state $q$ with read-write head at position $pos$ of the tape that contains a letter $\sigma \in \Sigma$,
and if $\Delta(q, \sigma)$ $=$ $(q', \sigma', D)$, then $o''$ represents the configuration in which $M$ is at state $q'$,
in which read-write head is at the position $pos-1/pos+1$ depending on whether $D=-1/+1$, and $\sigma'$ is at the position $pos$
of the tape. As a consequence of the above arguments, it follows that $o'$ represents an accepting  configuration of $M$, i.e.
a configuration in which the state is $q_A$, the lone accepting, halting state. 
This means that $M$ accepts the string $\vec w$.

For the converse, we briefly show that if  $QS^M_{\C} \not \models CQ^M$ then $M$ does not accept $\vec w$. 
Suppose that  $QS^M_{\C} \not \models CQ^M$, then by Theorem \ref{dChaseUniversalModelProperty},
for every assignment $\mu\colon\var(CQ^M)$ $\rightarrow$ $\const$, it should be the case that
$CQ^M[\mu]$ $\not \subseteq$ $dChase(QS_{\C})$. By the initialization axioms, we know that there exists a constant $o\in \const(dChase(QS_{\C}))$ 
with $c_n\colon(o$, \texttt{rdf:type}, $conInit)$ $\in$ $dChase(QS_{\C})$.  
 We know that $o$ represents the initial configuration of $M$. Also by the initial construction axioms of $QS^M_{\C}$, we know that $o$ is the initial element of
a double exponential chain of objects that are linearly ordered by property symbol $conSucc$. From transition axioms we know that,
if, for any $o''$, $c_n\colon (o$, $conSucc$, $o'')$ $\in$ $dChase(QS_{\C})$, then $o''$ represents a valid successor configuration of $o$, 
which itself holds for $o''$, and so on.
This means that for none of the succeeding double exponential configurations of $M$, the accepting state $q_A$ holds. This means
that $M$ does not reach an accepting configuration with string $\vec w$, and hence rejects it.

Since the construction above shows the existence of a polynomial time reduction of the word problem of a 2EXPTIME DTM, 
which is a 2EXPTIME-hard problem, to the CCQ entailment problem over csafe quad-systems, 
it immediately follows that CCQ entailment over csafe/msafe/safe quad-systems is 2EXPTIME-hard. 

\subsection{Procedure for detecting safe/msafe/csafe quad-systems}
\noindent In this subsection, we present a procedure for deciding whether a given quad-system is safe (resp. msafe, resp. csafe) or not. If
the quad-system is safe (resp. msafe, resp. csafe), the result of the procedure is a \emph{safe dChase} (resp. \emph{msafe dChase}, \emph{csafe dChase})
that contains the standard dChase, and can be used for query answering.  
Since the safety (resp. msafety, resp. csafety) property of a quad-system is attributed to the 
dChase of the quad-system, the procedure nevertheless performs the standard operations for computing the dChase,
but also generate quads that indicate origin ruleIds and origin vectors (resp. origin ruleIds, resp. origin-contexts)
of each Skolem blank node generated. 
 In each iteration, a test for safety is performed, by checking the
 presence of  Skolem blank-nodes that violate the safety (resp. msafety, resp. csafety) condition. In
case a violation is detected, a distinguished quad is generated and the safe (resp. msafe, resp. csafe) dChase construction is aborted, prematurely.
 On the contrary, if there exists an iteration in which no new quad is generated,
the safe (resp. msafe, resp. csafe) dChase computation stops with a completed safe (resp. msafe, resp. csafe) dChase that 
contains the standard dChase. Since all the additional quads produced for accounting information
use a distinguished context identifier $c_c \not \in \C$, the computed safe (resp. msafe, resp. csafe) dChase itself can be used for
standard query answering. Before geting to the details of the procedure, 
we give a few necessary definitions. 
\begin{definition}[Context Scope]
The context scope of a term 
$t$ in a set of quad-patterns  
$Q$, 
denoted by $cScope(t,Q)$ is given as: 
$cScope(t,Q)$ $=$ $\{c \ | \ c\colon(s$, $p$, $o)$ $\in$ $Q,s=t \vee p=t \vee o=t \}$.
\end{definition} 
\noindent For any quad-system $QS_{\C}$ $=$ $\langle Q_{\C},R \rangle$, let 
$c_c$ be an ad hoc context identifier such that $c_c \not \in \C$, then for $r_i$ $=$ $body(r_i)(\vec x$, $\vec z)$ $\rightarrow$ $head(r_i)(\vec x$, 
$\vec y)$ $\in$ $R$, we define transformations $augS(r_i)$, $augM(r_i)$, $augC(r_i)$ as follows:
\begin{eqnarray}
&& augS(r_i)= body(r_i)(\vec x, \vec z) \rightarrow head(r_i)(\vec x, \vec y) \wedge  \forall y_j \in   \{\vec y\} \ [  
 \bigwedge_{x_k \in \{\vec x\}} c_c\colon(x_k, \nonumber \\
&& \text{descendantOf}, y_j) \ \wedge \ c_c\colon(y_j, 
  \text{descendantOf}, y_j) \ \wedge c_c\colon(y_j, \text{originRuleId},i) \ \wedge \hspace{0.5cm} \nonumber \\
&&   c_c\colon(y_j, \text{originVector},\vec x) ]  \hspace{3.7cm}\nonumber
\end{eqnarray}
It should be noted that $c_c\colon(y_j$, originVector, $\vec x)$ is not a valid quad pattern, 
 and is only used for notation brevity. In the actual implementation, vectors can be stored using an rdf container data structure
 such as $\texttt{rdf:List}$, $\texttt{rdf:Seq}$ or by typecasting it as a string. 
 \begin{eqnarray}
&& augM(r_i)= body(r_i)(\vec x, \vec z) \rightarrow head(r_i)(\vec x, \vec y) \wedge  \forall y_j \in \{\vec y\} \ [  \bigwedge_{x_k \in \{\vec x\}} c_c\colon(x_k,  \nonumber \\
&&  \text{descendantOf}, y_j) \ \wedge \ c_c\colon(y_j, 
  \text{descendantOf}, y_j) \ \wedge c_c\colon(y_j, \text{originRuleId},i)] \hspace{1.3cm} \nonumber 
\end{eqnarray}
\begin{eqnarray}
&& augC(r_i)= body(r_i)(\vec x, \vec z) \rightarrow head(r_i)(\vec x, \vec y) \wedge  \forall y_j \in   \{\vec y\} \ 
 [  \bigwedge_{x_k \in \{\vec x\}} c_c\colon(x_k,\nonumber \\
&& \text{descendantOf}, y_j) \wedge \ c_c\colon(y_j,  
  \text{descendantOf}, y_j) \ \wedge  \bigwedge_{ c \in cScope(y_j, head(r_i))} c_c\colon(y_j,  \hspace{1cm} \nonumber \\
&& \text{originContext}, c) ]  \hspace{5.5cm}\nonumber
\end{eqnarray}
Intuitively, the transformation $augS/augM/augC$ on a  BR $r_i$,
augments the head part of $r_i$ with additional types of quad patterns, which are the following:
\begin{enumerate}
 \item $c_c\colon(x_k, \text{descendantOf}, y_j)$, 
for every existentially quantified variable $y_j$ in $\vec y$ and universally quantified variable $x_k$ $\in$ 
$\{\vec x\}$. This is done because, during dChase computation any application of an assignment $\mu$ to $r_i$ such that $\vec x[\mu]$
$=\vec a$,  resulting 
in the generation of a Skolem blank node $\_\colon b$ $=$ $\mu^{ext(\vec y)}(y_j)$, any $a_i \in \{\vec a\}$ is a descendant of
$\_\colon b$. Hence, due to these additional quad-patterns,
quads of the form $c_c\colon(a_i$, \text{descendantOf}, $\_\colon b)$ are also produced, and
in this way, keeps track of the descendants of any Skolem blank node produced.
\item  $c_c\colon(y_j$, descendantOf, $y_j)$, in order
to maintain also the reflexivity of `descendantOf' relation.
\item  $c_c\colon(y_j, \text{originContext},c)$, 
for every existentially quantified variable $y_j$ in $\{\vec y\}$, every $c$ $\in$ $cScope(y_j$, $head(r_i))$.
This is done because during dChase computation, any application of an assignment $\mu$ on $r_i$, such that $\vec x[\mu]$
$=\vec a$, resulting 
in the generation of a Skolem blank node $\_\colon b$ $=$ $\mu^{ext(\vec y)}(y_j)$, $c$ is an origin context of $\_\colon b$. 
 Hence due to these additional quad-patterns,
quads of the form $c_c\colon(\_\colon b, \text{originContext},c)$ is also produced. In this way, we keep track
of the origin-contexts of any Skolem blank node produced.
\item $c_c\colon(y_j$, originVector, $\vec x)$, 
This is done because during the dChase computation, for any application of an assignment $\mu$ on $r_i$, such that $\vec x[\mu]$
$=\vec a$, resulting  in the generation of a Skolem blank node $\_\colon b$ $=$ $\mu^{ext(\vec y)}(y_j)$,
$\vec a$ is the origin vector of $\_\colon b$.  Hence, due to these additional quad-patterns,
quads of the form $c_c\colon(\_\colon b$, originVector, $\vec a)$ is also produced. In this way, we keep track
of the origin vector of any Skolem blank node produced.
\item $c_c\colon(y_j$, originRuleId, $i)$, 
for every  existentially quantified variable $y_j$ in $\{\vec y\}$, inorder to keep track of the ruleId of the BR used to create any Skolem blank node.
\end{enumerate}
It can be noticed that for any BR $r_i$ without existentially quantified variables, the transformations $augS/augM/augC$ leaves $r_i$ unchanged.
For any set of BRs $R$, let 
\begin{eqnarray}
&& augS(R) \ (\text{resp. } augM(R), \text{resp. } augC(R)) =\nonumber \bigcup_{r_i \in R} augS(r_i) \ (\text{resp. } augM(r_i), \\
&&  \text{resp. }augC(r_i)) \cup  \{c_c \colon (x_1, \text{descendantOf}, z_1) \wedge c_c \colon (z_1, \text{descendantOf},  x_2)  \nonumber \\
&& \rightarrow  c_c \colon (x_1, \text{descendantOf}, x_2)\} \nonumber
\end{eqnarray}
\noindent The function unSafeTest (resp. unMSafeTest, resp. unCSafeTest) defined below, given a  
BR $r_i$ $=$  $body(r_i)(\vec x$, $\vec z)$ $\rightarrow$ $head(r_i)(\vec x$, $\vec y)$, an assignment $\mu$, and a quad-graph $Q$ checks, if application
of $\mu$ on $r_i$  violates the safety (resp. msafety, resp. csafety) condition on $Q$. 

\noindent\textbf{unSafeTest}($r_i, \mu, Q$)$=$True iff $\exists \_\colon b$, $\_\colon b' \in \bn$, 
with all the following conditions being satisfied:
\begin{itemize}
\item 
$\_\colon b \in \{\vec x[\mu]\}$, and
\item 
$c_c\colon(\_\colon b',\text{descendantOf}, \_\colon b)\in  Q$, and 
\item 
$c_c\colon(\_\colon b', \text{originRuleId}, i) \in Q$, and 
\item 
$c_c\colon(\_\colon b', \text{originVector},\vec a) \in Q$, and 
$\vec a \cong \vec x[\mu]$.
\end{itemize}
Intuitively, unSafeTest returns True, if 
 $\mu$ applied to $r_i$ will produce a fresh Skolem blank node $\_\colon b''$, whose child $\_\colon b$ $\in$ $\{\vec x[\mu]\}$,
 and according to knowledge in $Q$,
  $\_\colon b'$ is a descendant of $\_\colon b$  
such that the origin ruleId of $\_\colon b'$ is $i$ (which is also the origin ruleId of $\_\colon b''$) and the origin vector of $\_\colon b'$ is isomorphic to 
the origin vector of $\vec x[\mu]$ (which is also the origin vector of $\_\colon b''$). The functions unMSafeTest and unCSafeTest are similarly defined 
as follows:

\noindent\textbf{unMSafeTest}($r_i$, $\mu$, $Q$)$=$True iff $\exists \_\colon b$, $\_\colon b' \in \bn$, 
with all the following conditions being satisfied:
\begin{itemize}
\item 
$\_\colon b \in \{\vec x[\mu]\}$, and
\item 
$c_c\colon(\_\colon b'$,\text{descendantOf}, $\_\colon b)\in  Q$, and 
\item 
$c_c\colon(\_\colon b'$, \text{originRuleId}, $i) \in Q$. 
\end{itemize}
\noindent\textbf{unCSafeTest}($r_i, \mu$, $Q$)$=$True iff $\exists \_\colon b$, $\_\colon b' \in \bn, 
 \ \exists y_j$ $\in$ $\{\vec y\}$, with all the following being satisfied:
\begin{itemize}
\item 
$\_\colon b \in \{\vec x[\mu]\}$, and
\item 
$c_c\colon(\_\colon b'$, \text{descendantOf}, $\_\colon b)\in  Q$, and 
\item 
$\{c \ | \ c_c\colon(\_\colon b'$, \text{originContext}, $c) \in Q\}$
$=$ $cScope($ $y_j$, $head(r_i)(\vec x$, $\vec y))\setminus \{c_c\}$.
\end{itemize}
For any BR $r_i$ and an assignment $\mu$, the \emph{safe}/\emph{msafe}/\emph{csafe application} of $\mu$ on $r_i$ w.r.t. a quad-graph $Q_{\C}$ is defined as follows:
\[
apply^{\text{safe}}(r_i, \mu, Q_{\C})= \left\{
\begin{array}{l}
 \textbf{unSafe,} \ \text{ If } \text{unSafeTest}(r_i, \mu, Q_{\C})=\text{True}; \\
apply(r_i, \mu), \ \text{ Otherwise};
\end{array}
\right.
\]
\[
apply^{\text{msafe}}(r_i, \mu, Q_{\C})= \left\{
\begin{array}{l}
 \textbf{unMSafe,} \ \text{ If } \text{unMSafeTest}(r_i,\mu, Q_{\C})=\text{True}; \\
apply(r_i, \mu), \ \text{ Otherwise};
\end{array}
\right.
\]
\[
apply^{\text{csafe}}(r_i, \mu, Q_{\C})= \left\{
\begin{array}{l}
 \textbf{unCSafe,} \ \text{ If } \text{unCSafeTest}(r_i,\mu, Q_{\C})=\text{True}; \\
apply(r_i, \mu), \ \text{ Otherwise};
\end{array}
\right.
\]
where \textbf{unSafe} $=$ $c_c\colon(\text{unsafe}$, \text{unsafe}, \text{unsafe}) 
(resp. \textbf{unMSafe} $=$ $c_c\colon(\text{unmsafe}$, \text{unmsafe}, \text{unmsafe}), 
resp. \textbf{unCSafe} $=$ $c_c\colon(\text{uncsafe}$, \text{uncsafe}, \text{uncsafe})$)$
is a distinguished quad
that is generated, if the prerequisites of safety (resp. msafety, resp. csafety) is violated.
For any quad-system $QS_{\C}$ $=$ $\langle Q_{\C},R \rangle$, we define its \emph{safe dChase} $dChase^{\text{safe}}(QS_{\C})$ as follows: 

 $dChase^{\text{safe}}_0(QS_{\C})$ $=$ $Q_{\C}$;  
 $dChase^{\text{safe}}_{m+1}(QS_{\C})$ $=$ $dChase^{\text{safe}}_m(QS_{\C})$
 $\cup$  $apply^{\text{safe}}(r_i$, $\mu$, $dChase^{\text{safe}}_m(QS_{\C}))$, if  $\exists \ r_i$  $\in$ $augS(R)$, assignment $\mu$ 
 such that $applicable_{augS(R)}(r_i$, $\mu$, $dChase^{\text{safe}}_m(QS_{\C}))$; 
 
 $dChase^{\text{safe}}_{m+1}(QS_{\C})$ $=$ $dChase^{\text{safe}}_m(QS_{\C})$, \text{otherwise}; for any $m \in \mathbb{N}$.

$dChase^{\text{safe}}(QS_{\C})=\bigcup_{m\in \mathbb{N}} dChase^{\text{safe}}_m(QS_{\C})$ 

The termination condition for safe dChase computation can be implemented using the following conditional:
If there exists $m$ such that 

$dChase^{\text{safe}}_m(QS_{\C})$ $=$ $dChase^{\text{safe}}_{m+1}(QS_{\C})$; then

$dChase^{\text{safe}}(QS_{\C})=dChase^{\text{safe}}_m(QS_{\C})$.

\noindent The dChases $dChase^{\text{msafe}}(QS_{\C})$ and $dChase^{\text{csafe}}(QS_{\C})$ are defined, similarly, for 
msafe and csafe quad-systems, respectively. We bring to the notice of the reader that although application 
of any $augS(r)$ (resp. $augM(r)$, resp. $augC(r)$) produces quad-patterns of the form $c_c\colon (\_\colon b$, descendantOf, $\_\colon b)$, 
for any Skolem blank node $\_\colon b$ generated, there is no raise of a false alarm in the 
unSafeTest (resp. unMSafeTest, resp. unCSafeTest). This is because unSafeTest (resp. unMSafeTest, resp. unCSafeTest) 
on a bridge rule $r$ $=$ $body(r)(\vec x$, $\vec z)$ $\rightarrow$ $head(r)(\vec x$, $\vec y)$ and assignment $\mu$ checks
if the application of $\mu$ of $r$ with the fresh $\_\colon b''$ assigned to a $y_i \in \{\vec y\}$ by $\mu^{ext(\vec y)}$
would have a child  $\_\colon b \neq b''$ assigned to some $x_i \in \{\vec x\}$ by $\mu$, such that
there exists a  quad of the form $c_c\colon (\_\colon b'$, descendantOf, $\_\colon b)$
in the safe (resp. msafe, resp. csafe) dChase constructed so far, and  $\_\colon b''$ and $\_\colon b'$ 
have the same origin ruleId and originVector (resp. originRuleId, resp. originContexts). Note 
that in the above $\_\colon b'$ should also be distinct from $\_\colon b''$, and hence rules out the 
case in which unSafeTest (resp. unMSafeTest, resp. unCSafeTest) returns True because 
of the detection of a blank node as a self descendant of itself.

\noindent The following theorem shows that the procedure above described for detecting unsafe quad-systems is sound and complete:
\begin{theorem}\label{theorem:soundnessCompleteness}
  For any quad-system $QS_{\C}$ $=$ $\langle Q_{\C},R \rangle$, the quad \textbf{unSafe} (resp. \textbf{unMSafe}, resp. \textbf{unCSafe}) 
  $\in dChase^{\text{safe}}(QS_{\C})$ (resp. $dChase^{\text{msafe}}(QS_{\C})$, resp. $dChase^{\text{csafe}}(QS_{\C})$),
iff $QS_{\C}$  is unsafe (resp. unmsafe, resp. uncsafe). 
\end{theorem}
\noindent It should be noted that for any quad-system $QS_{\C}$ $=$ $\langle Q_{\C}$, $R\rangle$,  
$dChase^{\text{safe}}(QS_{\C})$ (resp. $dChase^{\text{msafe}}(QS_{\C})$, resp. $dChase^{\text{csafe}}(QS_{\C})$) is a finite set and hence 
the iterative procedure which we described earlier terminates, regardless of whether $QS_{\C}$ is safe (resp. msafe, resp. csafe) or not. 
This is because if $QS_{\C}$ is safe (resp. msafe, resp. csafe), then, as we have seen before, there exists a double exponential bound on 
number of quads in its dChase. Hence, there is an iteration in which no new quad is generated, which leads to stopping of computation. Otherwise,
if $QS_{\C}$ is unsafe (resp. msafe, resp. csafe), then from Theorem~\ref{theorem:soundnessCompleteness}, we know that the 
quad \textbf{unSafe} (resp. \textbf{unMSafe}, resp. \textbf{unCSafe}) gets generated in $dChase^{\text{safe}}(QS_{\C})$ (resp. $dChase^{\text{msafe}}(QS_{\C})$, 
resp. $dChase^{\text{csafe}}(QS_{\C})$) in not more than $\bigO(2^{2^{\|QS_{\C}\|}})$ iterations. 
This implies that there exists an iteration $m$ such that the quad  \textbf{unSafe} (resp. \textbf{unMSafe}, 
resp. \textbf{unCSafe}) is in $dChase_m^{\text{safe}}(QS_{\C})$ (resp. $dChase_m^{\text{msafe}}(QS_{\C})$, 
resp. $dChase_m^{\text{csafe}}(QS_{\C})$). W.l.o.g, let $m$ be the first such iteration. This means that there exists a BR $r_i \in R$ with 
head $head(r_i)(\vec x$, $\vec y)$, 
assignment $\mu$ such that $applicable_{augS(R)}(r_i$, $\mu$, $dChase_{m-1}^{\text{safe}}(QS_{\C}))$  (resp. $applicable_{augM(R)}(r_i$, $\mu$, 
$dChase_{m-1}^{\text{msafe}}(QS_{\C}))$, resp. $applicable_{augC(R)}(r_i$, $\mu$, $dChase_{m-1}^{\text{csafe}}(QS_{\C}))$ holds. 
By construction, since $head(r_i)[\mu^{ext(\vec y)}]$ is not generated,
and instead the  quad \textbf{unSafe} (resp. \textbf{unMSafe}, resp. \textbf{unCSafe}) is generated, 
$applicable_{augS(R)}(r_i$, $\mu$, $dChase_{m}^{\text{safe}}(QS_{\C}))$  (resp. $applicable_{augM(R)}(r_i$, $\mu$, $dChase_{m}^{\text{msafe}}(QS_{\C}))$, 
resp. $applicable_{augC(R)}($ $r_i$, $\mu$, $dChase_{m}^{\text{csafe}}(QS_{\C}))$ holds yet again. This means that the termination condition is satisfied at 
iteration $m+1$, and hence computation stops. Note that regardless of whether a given quad-system is safe (resp. msafe, resp. csafe) or not, the 
number of safe (resp. msafe, resp. csafe) dChase iterations is double exponentially bounded in the size of the quad-system.
Consequently, we derive the following theorem.
\begin{theorem}
Recognizing whether a quad-system is safe/msafe/csafe is in 2EXPTIME.
\end{theorem}
\noindent Also notice that after running procedure described above, if the quad \textbf{unSafe} (resp. \textbf{unMSafe}, resp. \textbf{unCSafe}) is not generated, then
 its safe (resp. msafe, resp. csafe) dChase itself can be used for CCQ answering, as in such a case the standard dChase is contained in safe (resp. msafe,
 resp. csafe) dChase,  and all the quads generated
 for accounting information have the  context identifier $c_c$. Hence, for any safe (resp. msafe, resp. csafe) quad-system, for any boolean CCQ 
 that does not contain quad patterns of the form $c_c\colon(s,p,o)$, the dChase entails CCQ iff the safe (resp. msafe, resp. csafe) dChase entails CCQ.
 
 A set of BRs $R$ is said to be \emph{universally safe} (resp. \emph{msafe}, resp. \emph{csafe}) iff, for any quad-graph $Q_{\C}$, 
the quad-system $\langle Q_{\C}, R\rangle$ is safe (resp. msafe, resp. csafe). 
 For any set of BRs $R$, whose set of context identifiers is $\C$, also let $U_R$ be the set of URIs that occur in 
 the triple patterns of $R$ plus an additional ad hoc blank node $\_\colon b^{crit}$, the \emph{critical quad-graph} of $R$ is defined 
 as the set $\{c\colon (s,p,o)|c \in \C, \{s,p,o\}\subseteq U_R\}$. The following property illustrates how the critical quad-graph of
 a set of BRs $R$ can be used to determine, whether or not $R$ is universally safe/msafe/csafe. 
 \begin{property}\label{prop:UniversalSafety}
  A set of BRs $R$ is universally safe (resp. msafe, resp. csafe) iff $\langle Q^{crit}_{\C}, R\rangle$ is safe (resp. msafe, resp. csafe), 
  where $Q^{crit}_{\C}$ is the critical quad-graph of $R$.
 \end{property}

\section{Range Restricted Quad-Systems: Restricting to Range Restricted BRs}\label{sec:horn}
\noindent In this section, we investigate the complexity of CCQ entailment over quad-systems, whose
BRs do not have existentially quantified variables. Such BRs are of the form: 
\begin{eqnarray}
&& c_1:t_1(\vec x, \vec z) \wedge ... \wedge c_n:t_n(\vec x, \vec z) \rightarrow 
 c'_1:t'_1(\vec x) \wedge ... \wedge c'_m:t'_m(\vec x) \nonumber
\end{eqnarray}
Note that any set of BRs $R$ of the form above can be replaced by semantically equivalent set $R'$, such that
each $r \in R'$ is the form:
\begin{equation}\label{eqn:horn-intraContextualRule}
 c_1:t_1(\vec x, \vec z), ... ,c_n:t_n(\vec x, \vec z) \rightarrow  \ c'_1:t'_1(\vec x)
\end{equation}
Also $\|R'\|$ is at most quadratic in $\|R\|$, and hence, w.l.o.g, we assume that each $r \in R$ is of the 
form (\ref{eqn:horn-intraContextualRule}).
 Borrowing the parlance from the $\forall \exists$ rules setting, 
where rules whose variables in the head part are contained in the variables in the body part are called range restricted rules~\cite{BagetLMS11},
we call such BRs \emph{range restricted} (RR) BRs. 
We call a quad-system whose BRs are all of RR-type, a \emph{RR quad-system}. 
Since there exists no existentially quantified variable in the BRs of a RR quad-system, no Skolem blank node is produced during 
dChase computation. Hence, there can be no violation of the safety/msafety/csafety condition in section \ref{sec:safe}, and hence, 
the class of RR quad-systems are contained in the class of safe/msafe/csafe quad-systems, and is also a FEC. Of course, this containment is
strict as any quad-system that contains a BR with an existential variable is not RR. Since one can determine whether 
or not a given quad-system is RR or not by simply iterating through set of BRs and checking their syntax, the following holds:
\begin{theorem}
Recognizing whether a quad-system is RR can be done in linear time.
\end{theorem}
In the following, we see that restricting to RR BRs,  size of the dChase becomes polynomial w.r.t. size of the 
input quad-system, and the complexity of CCQ entailment further reduces compared to safe/msafe/csafe quad-systems.
\begin{lemma}\label{lemma:chase-computation-horn}
For any RR quad-system $QS_{\C}=\langle Q_{\C}, R\rangle$,  the
following holds:  (i) $\|dChase(QS_{\C})\|$ $=$ $\bigO(\|QS_{\C}\|^4)$ 
(ii) $dChase(QS_{\C})$ can be computed in \textsc{EXPTIME}
(iii) If $\|R\|$ is fixed to be a constant, $dChase(QS_{\C})$ can be computed in \textsc{PTIME}.
\end{lemma}
\begin{proof}

(i)  
Note that the number of constants in $QS_{\C}$ is roughly equal to  $\|QS_{\C}\|$.
As no existential variable occurs in any BR in a RR quad-system $QS_{\C}$, the
set of constants $\const(dChase(QS_{\C}))$ is contained in $\const(QS_{\C})$. Since
each $c\colon(s,p,o) \in dChase(QS_{\C})$ is such that $c,s,p,o \in \const(QS_{\C})$, 
$|dChase(QS_{\C})|$ $=$ $\bigO(|\const(QS_{\C}$ $)|^4)$. Hence $\|dChase(QS_{\C})\|$ $=$ $\bigO(|\const(QS_{\C})|^4)$
$=$ $\bigO(\|QS_{\C}\|^4)$.

(ii) Since from (i) $|dChase(QS_{\C})|$ $=$ $\bigO(\|QS_{\C}\|^4)$, and in each iteration of the dChase
at least one new quad is added, the number of iterations cannot exceed $\bigO(\|QS_{\C}\|^4)$.
 Since by Lemma \ref{lemma:chaseSizeIncrease}, computation of each iteration $i$ of the dChase
 requires $\bigO(|R|*\|dChase_{i-1}(QS_{\C})\|^{rs})$ time, where $rs=max_{r \in R} \|r\|$, and  
 $rs\leq \|QS_{\C}\|$,  time required for each iteration is of the order $\bigO(2^{\|QS_{\C}\|})$ time. Although 
 the number of iterations is a polynomial,
 each iteration requires an exponential amount of time w.r.t $\|QS_{\C}\|$. Hence time complexity of dChase computation
 is in EXPTIME.

(iii)  As we know that the time taken for application of a BR $R$ is $\bigO(\|dChase_{i-1}(QS_{\C})\|^{\|R\|})$.
Since  $\|R\|$ is fixed to a constant, application of $R$ can be done in PTIME.
Hence, each dChase iteration can be computed in PTIME.
Also since the number of iterations is a polynomial in $\|QS_{\C}\|$, computing dChase is in PTIME.
\end{proof}
\begin{theorem}\label{theorem:hornDataComplexity}
 Data complexity of CCQ entailment over RR quad-systems is \textsc{PTIME}-complete.
\end{theorem}
\begin{proof}
(Membership) Follows from the membership in P of data complexity of CCQ entailment
for safe quad-systems, whose expressivity subsumes the expressivity of 
RR quad-systems (Theorem \ref{thoeorem:safe-computational-properties}). 

 (Hardness)
In order to prove P-hardness, we reduce a well known P-complete problem, 3HornSat, i.e. the satisfiability of 
propositional Horn formulas with at most 3 literals. Note that a (propositional) Horn formula is a propositional formula of the form:
\begin{eqnarray}\label{eqn:PropHornFormula}
 P_1 \wedge \ldots \wedge P_n \rightarrow P_{n+1}  
\end{eqnarray}
where $P_i$, for  $1 \leq i \leq n+1$, are either propositional variables or constants $t$, $f$, that represents
true and false, respectively. Note that for any propositional variable $P$, the fact that ``$P$ holds'' is 
represented by the formula $t \rightarrow P$, and ``$P$ does not hold'' is 
represented by the formula $P \rightarrow f$.
A 3Horn formula is a formula of the form~(\ref{eqn:PropHornFormula}), where $1\leq n\leq 2$. Note that any (set of) 
Horn formula(s) $\Phi$ can be transformed in polynomial time to a polynomially sized set $\Phi'$ of 3Horn formulas,
by introducing auxiliary propositional variables such that $\Phi$ is satisfiable iff $\Phi'$ is satisfiable.
A pure 3Horn formula is a 3Horn formula of the form~(\ref{eqn:PropHornFormula}), where $n=2$. Any 3Horn formula $\phi$ that is not pure
can be trivially converted to equivalent pure form by appending a $\wedge \ t$ on the body part of $\phi$. For instance, 
 $P\rightarrow Q$,  can be converted to  $P \wedge t \rightarrow Q$. Hence, w.l.o.g. we assume that any set of 3Horn formulas is pure, and is of the form:
 \begin{eqnarray}\label{eqn:PropPure3HornFormula}
  P_1 \wedge P_2 \rightarrow P_3  
 \end{eqnarray}
In the following, we reduce the satisfiability problem of pure 3Horn formulas to CCQ 
entailment problem over a quad-system whose set of schema triples, the set of BRs, and the CCQ $CQ$ are all fixed. 

 For any set of pure Horn formulas $\Phi$, we construct the quad-system $QS_{\C}=\langle Q_{\C}, R\rangle$, where
 $\C=\{c_t, c_f\}$. For any formula $\phi \in \Phi$ of the form (\ref{eqn:PropPure3HornFormula}), $Q_{\C}$ contains
 a quad $c_f \colon (P_1, P_2, P_3)$. In addition $Q_{\C}$ contains a quad $c_t\colon (t$, \texttt{rdf:type}, $T)$. 
 $R$ is the singleton that contains only the following fixed BR:
 \begin{eqnarray}
&&  c_t\colon (x_1, \texttt{rdf:type}, T), c_t\colon (x_2,  \texttt{rdf:type}, T), 
  c_f\colon (x_1, x_2, x_3) \rightarrow c_t\colon (x_3,  \nonumber \\
&&  \texttt{rdf:type}, T) \nonumber 
 \end{eqnarray}
Let the $CQ$ be the fixed query $c_t:(f, \texttt{rdf:type}, T)$. 

Now, it is easy to see that $QS_{\C}$ $\models$ $CQ$, iff $\Phi$ is not satisfiable.
\end{proof}

\begin{theorem}
Combined complexity of CCQ entailment over  RR quad-systems is in \textsc{EXPTIME}.
\end{theorem}
\begin{proof} 
(Membership)
By Lemma \ref{lemma:chase-computation-horn}, for any RR quad-system $QS_{\C}$, its dChase $dChase(QS_{\C})$
can be computed in EXPTIME. Also by Lemma \ref{lemma:chase-computation-horn}, its dChase size $\|dChase(QS_{\C})\|$
is a polynomial w.r.t to $\|QS_{\C}\|$.  
A boolean CCQ $CQ()$ can naively be evaluated by grounding the set of constants in the dChase to the variables in the $CQ()$,
and then checking if any of these groundings are contained in $dChase(QS_{\C})$. The number of such 
groundings can at most be $\|dChase(QS_{\C})\|^{\|CQ()\|}$ ($\dagger$). Since $\|dChase(QS_{\C})\|$ is a polynomial in
$\|QS_{\C}\|$, there are an exponential number of groundings w.r.t $\|CQ()\|$. Since containment of each of these 
groundings can be checked in time polynomial w.r.t. the size of $dChase(QS_{\C})$, and since $\|dChase(QS_{\C})\|$ is a polynomial w.r.t. $\|QS_{\C}\|$,
 the time complexity of CCQ entailment is in EXPTIME. 
\end{proof}
\noindent Concerning the combined complexity of CCQ entailment of RR quad-systems, we leave the lower bounds open.
\subsection{Restricted RR Quad-Systems}
\noindent We call those quad-systems with BRs of form (\ref{eqn:horn-intraContextualRule})
with a fixed bound on $n$ as \emph{restricted RR quad-systems}. They can be further
classified as linear, quadratic, cubic,..., quad-systems, when $n=1,2,3,...$, respectively.
\begin{theorem}
 Data complexity of CCQ entailment over restricted RR quad-systems is \textsc{P}-complete.
\end{theorem}
\begin{proof}
 The proof is same as in Theorem \ref{theorem:hornDataComplexity}, since the size of BRs are fixed to constant.
\end{proof}

\begin{theorem}
Combined complexity of CCQ entailment over restricted RR quad-systems is \textsc{NP}-complete.
\end{theorem}
\begin{proof}
 Let the  problem of deciding if $QS_{\C} \models CQ()$ be called DP'.

(Membership) 
for any $QS_{\C}$ whose rules are of restricted RR-type, the size of any $r \in R$ is a constant. Hence,
by Lemma \ref{lemma:chaseSizeIncrease}, any dChase iteration can be computed in PTIME. Since the number of iterations is also polynomial in 
$\|QS_{\C}\|$, $dChase(QS_{\C})$ can be computed in PTIME in the size of $QS_{\C}$ and $dChase(QS_{\C})$ has a polynomial number of 
constants.
Hence, we can guess an assignment $\mu$ for all the existential variables in CCQ $CQ()$, 
to the set of constants in $dChase(QS_{\C})$. Then, one can evaluate the CCQ, by checking if $c\colon(s,p,o) \in dChase(QS_{\C})$,
for each $c\colon(s,p,o) \in CQ()[\mu]$, which can be done in time $\bigO(\|CQ\|*\|dChase(QS_{\C})\|)$, and 
is hence is in non-deterministic PTIME, which implies that DP' is in NP.

(Hardness)
We show that DP' is NP-hard, by reducing the well known NP-hard problem of 3-colorability to DP'.
Given a graph $G$ $=$ $\langle V$, $E\rangle$, where $V$ $=$ $\{v_1$, ..., $v_n\}$ is the set of nodes, 
$E \subseteq V \times V$ is the set of edges,
the 3-colorability problem is to decide if there exists a labeling function 
$l:V \rightarrow \{r,b,g\}$ that assigns each $v \in V$ to an element in $\{r,b,g\}$ 
such that the condition: $(v,v') \in E \rightarrow l(v) \neq l(v')$, for each $(v, v') \in E$,
is satisfied.

One can construct a quad-system $QS_c=\langle Q_c, \emptyset\rangle$, 
where $graph_{Q_c}(c)$ has the following triples:


$\{(r,edge,b),(r,edge,g),(b,edge,g),(b,edge,r)$,
$(g,edge,r),(g,edge,b)\}$

Let $CQ$ be the boolean CCQ: 
$\exists v_1,....,v_n \bigwedge_{(v,v')\in E}$ $[$
$c\colon(v, edge, v')$ $\wedge$ $c\colon(v', edge, v)]$. 
Then, it can be seen that $G$ is 3-colorable, iff $QS_c \models  CQ$.
\end{proof}

\section{Quad-Systems and Forall-Existential rules: A formal comparison}\label{sec:Comparison}
\noindent In this section, we formally compare the formalism of quad-systems with forall-existential ($\forall \exists$) rules, which are also called 
Tuple generating dependencies (Tgds)/Datalog+- rules. $\forall\exists$ rules is a fragment of first order logic in which every formula 
is restricted to a certain syntactic form. A $\forall\exists$ rule  is a first order formula of the form:
\begin{eqnarray}\label{eqn:datalog+-rule}
 \forall \vec x \forall \vec z \ [p_1(\vec x,\vec z) \wedge ... \wedge p_n(\vec x,\vec z) \rightarrow 
 \exists \vec y \ p'_1(\vec x,\vec y) \wedge ... \wedge p'_m(\vec x, \vec y)] 
\end{eqnarray}
where $\vec x, \vec y, \vec z$ are vectors of variables such that $\{\vec x\}, \{\vec y\}$ and $\{\vec z\}$ are pairwise 
disjoint, $p_i(\vec x, \vec z)$, for $1 \leq i \leq n$ are predicate atoms whose variables are 
from $\vec x$ or $\vec z$, $p'_1(\vec x, \vec y)$, for $1 \leq i \leq m$ are predicate atoms whose variables are from $\vec x$ or $\vec y$.
We, for short, occasionally note a $\forall\exists$ rule of the form (\ref{eqn:datalog+-rule}) as $\phi(\vec x$, $\vec z)$ $\rightarrow$ $\psi(\vec x$, $\vec y)$, 
where $\phi(\vec x$, $\vec z)$ $=$ $\{p_1(\vec x,\vec z), ..., p_n(\vec x,\vec z)\}$, $\psi(\vec x$, $\vec y)$ $=$ $\{p'_1(\vec x,\vec y)$, 
... $p'_m(\vec x, \vec y)\}$. A set of $\forall \exists$ rules is called a $\forall \exists$ rule set.
In the realm of $\forall\exists$ rule sets, a conjunctive query (CQ)  is an expression of the form:
\begin{equation}\label{eqn:CQ}
\exists \vec y \ p_1(\vec x,\vec y) \wedge ... \wedge p_r(\vec x, \vec y) 
\end{equation}
where $p_i(\vec x,\vec y)$, for $1 \leq i \leq r$ are predicate atoms over vectors $\vec x$ or $\vec y$. A boolean CQ is defined as 
usual. The decision problem of whether, for a $\forall\exists$ rule set $\mathbb{P}$ and a CQ $Q$, if $\mathbb{P}$ $\models_{\text{fol}} Q$ is called the \emph{CQ EP}, 
 where $\models_{\textrm{fol}}$ is the standard first order logic entailment relation.

For any quad-graph $Q_{\C}=\{c_1\colon(s_1$, $p_1$, $o_1)$, \ldots, $c_n\colon (s_r$, $p_r$, $o_r)\}$, let $r_{Q_{\C}}$ be the BR
\begin{eqnarray}
&& \rightarrow \vec \exists y_{b_1}, \ldots, y_{b_q} \ c_1\colon (s_1, p_1, o_1)[\mu_{\bn}] 
 \wedge \ldots  \wedge c_r\colon (s_r, p_r, o_r)[\mu_{\bn}], \nonumber 
\end{eqnarray}
\noindent where $\{\_ \colon b_1$, \ldots, $\_ \colon b_q\}$ is the set of blank nodes in $Q_{\C}$, and 
$\mu_{\bn}$ is the substitution function $\{\_ \colon b_i \rightarrow y_{b_i}\}_{i=1, \ldots, q}$ that assigns each blank-node
to a fresh existentially quantified variable. It can be noted that the quad-systems $\langle Q_{\C}, R\rangle$ and 
$\langle \emptyset, R \cup \{r_{Q_{\C}}\}\rangle$ are semantically equivalent. The following definition gives
the translation functions that will be necessary to establish the relation between quad-systems and $\forall\exists$
rule sets.
\begin{definition}[Translations $\tau_q$, $\tau_r$, $\tau_{ccq}$, $\tau$]\label{def:TranslationQuadsystemsToForall} 
The translation function $\tau_q$ from the 
set of quad patterns to the set of ternary atoms is defined as: for any quad-pattern $c\colon (s,p,o)$, $\tau_{q}(c\colon (s,p,o))$ $=$ $c(s,p,o)$. 

The translation function $\tau_{br}$ from the set of BRs to the set of $\forall \exists$ rules is defined as:
for any BR $r$ of the form (\ref{eqn:intraContextualRuleWithQuantifiers}):
 \begin{eqnarray}
&&   \tau_{br}(r) = \forall \vec x \forall \vec z \ [\tau_{q}(c_1\colon t_1(\vec x, \vec z)) \wedge ... \wedge \tau_{q}(c_n\colon t_n(\vec x,  \nonumber \\ 
&& \vec z)) \rightarrow \exists \vec y \ \tau_{q}(c'_1\colon t'_1(\vec x, \vec y)) \wedge ... \wedge \tau_{q}(c'_m \colon t'_m(\vec x, \vec y))], \nonumber 
 \end{eqnarray} 
 
The translation function $\tau$ from the set of quad-systems to forall-existential rule sets is defined as: 
for any quad-system $QS_{\C}$ $=$ $\langle Q_{\C},R \rangle$, 
$\tau(QS_{\C})$ $=$ $\tau_{br}(R) \cup \{\tau_{br} (r_{Q_{\C}})\}$, where $\tau_{br}(R)$ $=$ $\bigcup_{r\in R} \tau_{br}(r)$. 

The translation function $\tau_{ccq}$ from the set of boolean CCQs to the set of boolean CQs is defined as:  
for any boolean CCQ $CQ$ $=$ $\exists \vec y$ $c_1\colon t_1(\vec a, \vec y)$ $\wedge$ \ldots $\wedge$ $c_r\colon t_r(\vec a, \vec y)$, 
$\tau_{ccq}(CQ)$ is: 
\begin{eqnarray}
\exists \vec y \ \tau_q(c_1\colon t_1(\vec a, \vec y)) 
 \wedge \ldots  \wedge \tau_q(c_r\colon t_r(\vec a, \vec y)). \nonumber   
\end{eqnarray}
\end{definition}
\noindent The following property gives the relation between CCQ entailment of unrestricted quad-systems and standard first order 
CQ entailment of $\forall \exists$ rule sets.
\begin{property}\label{prop:TranslationQuadsystemsToForall}
For any quad-system $QS_{\C}$, CCQ $CQ$, $QS_{\C} \models CQ$ iff $\tau(QS_{\C})$ $\models_{\textrm{fol}}$ 
$\tau_{ccq}(CQ)$.
\end{property}
\begin{proof}
\noindent  
Notice that every context $c\in \C$ becomes a ternary predicate symbol in the resulting translation.
Also, $\tau(QS_{\C})$ is a $\forall\exists$ rule set, and for any CCQ $CQ$, $\tau_{ccq}(CQ)$ is a CQ.  
 
In order to construct the restricted chase for $\tau(QS_{\C})$, suppose that $\prec_q$ is also extended to set of instances such that 
for any two quad-graphs $Q_{\C}$, $Q'_{\C'}$, $Q_{\C} \prec_q Q'_{\C'}$ iff 
$\tau_q(Q_{\C})$ $\prec_q$ $\tau_q(Q'_{\C'})$. Suppose $\prec$ is extended similarly
to set of instances. Also assume that during the construction of standard chase $chase(\tau(QS_{\C}))$ of $\tau(QS_{\C})$,
for any application of a $\tau_{br}(r)$ with existential variables, with $r \in R$, suppose that the Skolem blank nodes generated 
in $chase(\tau(QS_{\C}))$ follow the same order as they are generated in $dChase(QS_{\C})$. 
Also let us extend the rule applicability function to the $\forall\exists$ rules settings such that for any set of BRs $R$, for 
any $r\in R$, quad-graph $Q'_{\C'}$, assignment $\mu$, 
$applicable_R(r, \mu, Q'_{\C'})$ iff $applicable_{\tau_{br}(R)}(\tau_{br}(r)$, $\mu$, $\tau_q(Q'_{\C'}))$. 

Now $dChase_0(\langle \emptyset$, $R$ $\cup$ $\{r_{Q_{\C}}\}\rangle)$ $=$ $\emptyset$, and also $chase_0(\tau(QS_{\C}))$ $=$ $\emptyset$,  
 $dChase_1(QS_{\C})$ $=$ $apply(r_{Q_{\C}}$, $\mu_{\emptyset})$, where $\mu_{\emptyset}$ is the empty function,  
 $chase_1(\tau(QS_{\C}))$ $=$ $apply(\tau_{br}(r_{Q_{\C}})$, $\mu_{\emptyset})$, and so on. It is straightforward to see that
 for any $m\in \mathbb{N}$, $\tau_q($ $dChase_m(\langle \emptyset$,  $R \cup \{r_{Q_{\C}}\}\rangle))$ $=$ $chase_m(\tau(QS_{\C}))$. As a consequence, 
 $\tau_q(dChase(QS_{\C}))$ $=$ $chase(\tau(QS_{\C}))$, and $\{CQ\}[\sigma]\subseteq dChase(QS_{\C})$ iff
 $\{\tau_{ccq}(CQ)\}[\sigma]$ $\subseteq$ $chase(\tau(QS_{\C}))$.
 
Hence, applying Theorem ~\ref{dChaseUniversalModelProperty} and the analogous theorem 
for $\forall\exists$ rulesets from Deutch et al.~\cite{Deutsch:2008}, it follows that 
for any quad-system $QS_{\C}$ $=$ $\langle Q_{\C}$, $R \rangle$ and a boolean CCQ $CQ$, 
$QS_{\C} \models CQ$ iff $\tau(QS_{\C})$  $\models_{\textrm{fol}} \tau_{ccq}(CQ)$.
\end{proof}
\begin{theorem}\label{theorem:PolynomialTranslationQuadsystemsToForall}
There exists a polynomial time  
translation function $\tau$ (resp. $\tau_{ccq}$) from the set of unrestricted quad-systems (resp. CCQs) to the 
set of $\forall \exists$ rule sets (resp. CQs), such that for any unrestricted quad-system $QS_{\C}$ and a CCQ $CQ$, $QS_{\C} \models CQ$ iff 
$\tau(QS_{\C})$ $\models_{\textrm{fol}}$ 
$\tau_{ccq}(CQ)$.
\end{theorem}
\begin{proof}  
It is easy to see that $\tau_q$, $\tau_{br}$, $\tau$, and $\tau_{ccq}$ in Definition~\ref{def:TranslationQuadsystemsToForall}  
can be implemented using simple syntax transformation, by iterating through the respective components of a quad-system/CCQ,  
and the time complexity of these functions are linear w.r.t their inputs. 
\end{proof}
\noindent Notice that for any CCQ $CQ$ (resp. CQ $Q$), $\rightarrow CQ$ (resp. $\rightarrow Q$) is a bridge (resp. $\forall\exists$) rule, with an empty body. 
Also, since for any quad-graph $Q_{\C}$, the translation function $\tau_{br}$ defined above can directly be applied on $r_{Q_{\C}}$ 
to obtain a $\forall \exists$ rule, the following theorem immediately follows:
\begin{theorem}
 For quad-systems, the EPs: (i) quad EP, (ii) quad-graph EP, (iii) BR EP, (iv) BRs EP, (v) Quad-System EP, and (vi) CCQ EP are polynomially reducible
 to entailment of $\forall \exists$ rule sets.
\end{theorem}
\noindent 
A $\forall \exists$ rule set $\mathbb{P}$ is said to be a \emph{ternary} $\forall \exists$ rule set, iff all the 
predicate symbols in the vocabulary of $\mathbb{P}$ are of arity less than or equal to three. $\mathbb{P}$ is a \emph{purely ternary} rule set,
iff all the predicate symbols in the vocabulary $\mathbb{P}$ is of arity three. Similarly, 
a (purely) ternary CQ is defined.
The following property gives the relation between the CQ entailment problem of $\forall \exists$ rule sets and CCQ EP of unrestricted quad-systems. 
\begin{theorem}\label{theorem:PolynomialTranslationForallToQuadsystems}
 There exists a polynomial time tranlation function $\nu$ (resp. $\nu_{cq}$) from ternary $\forall\exists$ rule sets (resp. ternary CQs)
 to unrestricted quad-systems (resp. CCQs) such that for any ternary $\forall\exists$ rule set $\mathbb{P}$ and a ternary CQ $Q$, 
  $\mathbb{P}$ $\models_{\text{fol}} CQ$ iff $\langle \emptyset$, $\nu(\mathbb{P}) \rangle$ $\models \nu_{cq}(Q)$.
\end{theorem}
 \begin{proof}
Note that the CQ EP of  any ternary $\forall \exists$ rule set  $\mathbb{P}$, whose set of predicate symbols is $P$,  
and CQ $Q$ over $P$, can polynomially reduced to the CQ EP of a purely ternary rule set $\mathbb{P}'$ and purely ternary CQ $Q'$,
by the following transformation function $\chi$. Let $\Box$ be an adhoc fresh URI; $\chi$ is such that  
for any ternary atom $c(s,p,o)$, $\chi(c(s,p,o))$ $=$ $c(s,p,o)$. For any binary atom $c(s,p)$, $\chi(c(s,p))$ $=$ $c(s,p,\Box)$, 
and for any unary atom $c(s)$, $\chi(c(s))=c(s,\Box,\Box)$. For any $\forall\exists$ rule $r$ of the form (\ref{eqn:datalog+-rule}), 
\begin{eqnarray}
&& \chi(r)= \forall \vec x \forall \vec z \ [\chi(p_1(\vec x, \vec z)) \wedge \ldots \wedge  \chi(p_n(\vec x, \vec z))  \nonumber \\
&& \rightarrow \exists \vec y \ \chi(p'_1(\vec x, \vec y)) \wedge \ldots \wedge \chi(p'_m(\vec x, \vec y)) ] \nonumber
\end{eqnarray}
And, for any $\forall\exists$ rule set $\mathbb{P}$, $\chi(\mathbb{P})$ $=$ $\bigcup_{r \in \mathbb{P}}$ $\chi(r)$. 
For any CQ $Q$, $\chi(Q)$ is similarly defined. Note that for any ternary $\forall\exists$ rule set $\mathbb{P}$, 
ternary CQ $Q$, $\chi(\mathbb{P})$ (resp. $\chi(Q)$) is purely ternary, and $\mathbb{P} \models_{\text{fol}} Q$
iff $\chi(\mathbb{P}) \models_{\text{fol}} \chi(Q)$.

Also, it can straightforwardly seen that $\tau^{-1}_{br}(\chi(\mathbb{P}))$  
(resp. $\tau^{-1}_{ccq}(\chi(Q))$) is a set of BRs (resp. CCQ).
Suppose, $\nu(\mathbb{P})$ is  such that 
$\nu(\mathbb{P})$ $=$ $QS_{\C}$ $=$ $\langle \emptyset$, $\tau^{-1}_{br}(\chi(\mathbb{P}))\rangle$.
Intuitively, $\C$ contains a context identifier $c$, for each predicate symbol $c\in P$.
Also suppose, $\nu_{cq}(Q)=$ $\tau^{-1}_{ccq}(\chi(Q))$. Notice that $\nu_{cq}(Q)$ is CCQ.
It can straightforwardly seen that $\nu$ and $\nu_{cq}$ can be computed in polynomial time, and 
$\mathbb{P} \models_{\text{fol}} Q$ iff $\nu(\mathbb{P})$ $\models$ $\nu_{cq}(Q)$.
 \end{proof}
\noindent Thanks to Theorem~\ref{theorem:PolynomialTranslationQuadsystemsToForall} and Theorem~\ref{theorem:PolynomialTranslationForallToQuadsystems}, 
the following theorem immediately holds: 
 \begin{theorem}
  The CCQ EP over quad-systems is polynomially equivalent to CQ EP over ternary $\forall\exists$ rule sets.
 \end{theorem}
\noindent By virtue of the theorem above,  we derive the following property: 
\begin{property}\label{prop:QuadsystemEPsReductiontoCCQeps}
For quad-systems, the Quad EP, Quad-graph EP, BR(s) EP, and Quad-system EP are polynomially reducible to CCQ EP.
\end{property}
\begin{proof}
The following claim is a folklore in the realm of $\forall\exists$ rules.
\begin{claim}(1)
The $\forall\exists$ rule set EP is polynomially reducible to CQ EP.
\end{claim}
Reducibility of $\forall\exists$ rule EP to CQ EP is a folklore in the realm of $\forall\exists$ rules. 
 For a formal proof, we refer the reader to Baget et al.~\cite{BagetLMS11}, where it is shown that the $\forall \exists$ rule EP is 
polynomially reducible to fact (a set of instances) EP, and fact EP are equivalent to CQ EP. Also, Cali et al~\cite{CaliGK08} show that CQ containment problem, which is
 equivalent to $\forall\exists$ rule EP, is reducible to CQ EP.  Since a $\forall\exists$ rule set is a set of $\forall\exists$ rules, 
 by using a series of oracle calls to a function that solves the $\forall\exists$ rule EP, we can define a function for deciding
 $\forall\exists$ rule set entailment. Hence, the claim holds. 
 
(a)  Thanks to translation functions $\tau$, $\tau_{br}$ defined earlier, such that for any quad-system $QS_{\C}$, quad-graph $Q'_{\C'}$,
$QS_{\C} \models Q'_{\C'}$ iff $\tau(QS_{\C}) \models_{\text{fol}} \tau_{br}(r_{Q'_{\C'}})$, we can infer that
quad-graph EP is polynomially reducible to $\forall\exists$ rule set EP. 
 Applying claim 1, it follows the quad-graph EP over quad-systems is polynomially reducible to CQ EP over $\forall\exists$ rule sets. By Theorem
\ref{theorem:PolynomialTranslationForallToQuadsystems}, we can deduce that quad-graph EP is polynomially reducible to CCQ EP.
 
(b) By the translation functions $\tau$ and $\tau_{br}$, defined earlier, such that for any quad-system $QS_{\C}$, a set of BRs $R$,
$QS_{\C} \models R$ iff $\tau(QS_{\C}) \models_{\text{fol}} \tau_{br}(R)$, we can infer that
BRs EP is polynomially reducible to $\forall\exists$ rule set EP. Similar to (a) above, we deduce that
BRs EP is polynomially reducible to CCQ EP.

From (a) and (b), it follows that Quad-system EP is reducible to CCQ EP. 
\end{proof}
\noindent Having seen that the CCQ EP over quad-systems is polynomially equivalent to CQ EP over 
ternary $\forall\exists$ rule sets, we now compare  some of the well known techniques used to ensure 
decidability of CQ entailment in the $\forall\exists$ rules settings to the decidability techniques for quad-systems that we saw earlier in the 
previous sections. Note that since all the quad-system classes we proposed in this paper are FECs, 
for a judicious comparison, the $\forall\exists$ rule classes to which we compare are classes which have a finite chase property.
We compare to the following three well known classes: (i) Weakly Acyclic rule sets (WA),  (ii) Jointly Acyclic rule sets (JA), and (iii) Model 
Faithful Acyclic $\forall\exists$ rule sets (MFA). The following property is well known in the realm of $\forall\exists$ rules:
\begin{property}
For the any $\forall\exists$ rule set $\mathbb{P}$, the following holds:
\begin{enumerate}
 \item If $\mathbb{P} \in$ WA, then $\mathbb{P} \in $ JA (from \cite{KR11jointacyc}),
 \item If $\mathbb{P} \in$ JA, then $\mathbb{P} \in $ MFA (from \cite{Bernardo:dlacyclicity:2012}),
 \item WA $\subset$ JA $\subset$ MFA (from \cite{KR11jointacyc} and \cite{Bernardo:dlacyclicity:2012}).
\end{enumerate}
\end{property}
\noindent Note that a description of few other $\forall\exists$ rule classes that do not have the finite chase property, but still enjoy decidability
of CQ entailment are given in the related work. 
\subsection{Weak Acyclicity}
\noindent Weak acyclicity~\cite{Fagin05dataexchange,Deutsch03reformulationof} is a popular technique used to detect whether a  
 $\forall\exists$ rule set has a finite chase, thus ensuring decidability of query answering. The set \textsc{wa} represents class of ternary $\forall\exists$ rule sets
 that have the weak acyclicity property. 
 
For any predicate atom $p(t_1, \ldots, t_n)$,  an expression $\langle p,i \rangle$, for $i=1,\ldots, n$  is called
a position of $p$. In the above case, $t_1$ is said to occur at position $\langle p, 1\rangle$, $t_2$ at $\langle p, 2\rangle$, 
and so on.
For a set of $\forall\exists$ rules $\mathbb{P}$, its dependency graph is a
graph whose nodes are positions of predicate atoms in $\mathbb{P}$; for each  $r \in \mathbb{P}$ of the form (\ref{eqn:datalog+-rule}),
and for any variable $x$ occurring in position $\langle p,i\rangle$ in head of $r$:
\begin{enumerate}
 \item  if $x$ is universally quantified and $x$ occurs in the body of $r$ at position $\langle p',j\rangle$, 
 then there exists an edge from $\langle p',j\rangle$ to $\langle p,i\rangle$ 
\item if $x$ is existentially quantified, then for any universally quantified variable $x'$ occurring in the head of $r$, 
with $x'$ also occurring in the body of $r$ at position $\langle p',j\rangle$, 
there exists a special edge from $\langle p',j\rangle$ to $\langle p,i\rangle$.
\end{enumerate}
$\mathbb{P}$ is called weakly acyclic, iff its dependency graph does not contain cycles going through a special edge. For
any $\forall\exists$ rule set $\mathbb{P}$, if $\mathbb{P}$ is WA, then its chase is finite, and hence CQ EP is 
decidable. Note that the nodes in the dependency graph that has incoming special edges corresponds to the positions of predicates
where new values are created due to existential variables, and the normal edges capture the propagation of constants from one 
predicate position to another predicate position. In this way, absence of cycles involving special edges ensures that newly created Skolem blank nodes are not
recursively used to create other new Skolem blank nodes in the same position, leading to termination of chase computation.
\begin{example}\label{eg:comparisonWAcsafe}
 Let us revisit the quad-system $QS_{\C}=\langle Q_{\C}$, $R\rangle$ mentioned in example~\ref{eg:descendance}, whose dependency graph is shown in 
 Fig.~\ref{fig:dependencyGraph}. Note that the $QS_{\C}$ is uncsafe, 
 since its dChase contains a Skolem blank-node $\_\colon b_4$, which has as descendant another Skolem blank node 
 $\_\colon b_1$, with the same origin context $c_2$ (see Fig.~\ref{fig:descendance2}). However, it can be seen 
 from  Fig.~\ref{fig:dependencyGraph} that the dependency graph of $\tau(QS_{\C})$ does not contain 
 any directed cycle involving special edges. Hence $\tau(QS_{\C})$ is weakly acyclic. 
\end{example}
\begin{figure}[t]
\centering
\begin{tikzpicture}
[place/.style={circle,draw=black!50,fill=white!20,thick,
inner sep=0pt,minimum size=8mm}
]
\node(c11)[place] at ( -2,2) [place] {$\langle c_1, 1 \rangle$};
\node(c12)[place] at ( -2,-2) [place] {$\langle c_1, 2 \rangle$};
\node(c21)[place] at ( 1,2) [place] {$\langle c_2, 1 \rangle$};
\node(c23)[place] at ( 0.5,0) [place] {$\langle c_2, 3 \rangle$};
\node(c22)[place] at ( 1,-2) [place] {$\langle c_2, 2 \rangle$};
\node(c33)[place] at ( 3.5,0) [place] {$\langle c_3, 3 \rangle$};
\node(c32)[place] at ( -2,0) [place] {$\langle c_3, 2 \rangle$};
\draw [->] (c11.east) to node[auto,swap]{} (c21.west);
\draw [->] (c12.east) to node[auto,swap]{} (c22.west);
\draw [->] (c11.south) to node[auto,swap,label=below:$*$]{} (c23.north);
\draw [->] (c12.north) to node[auto,swap,label=below:$*$]{} (c23.south);
\draw [->] (c23.west) to node[auto,swap]{} (c32.east);
\draw [->] (c23.east) to node[auto,swap,label=below:$*$]{} (c33.west);
\draw [->] (c33.north) to node[auto,swap]{} (c21.east);
\draw [->] (c33.south) to node[auto,swap,label=below:$*$]{} (c22.east);
\end{tikzpicture}
\caption{Dependency graph of the quad-system in Example \ref{eg:descendance}.}
\label{fig:dependencyGraph}
\end{figure}
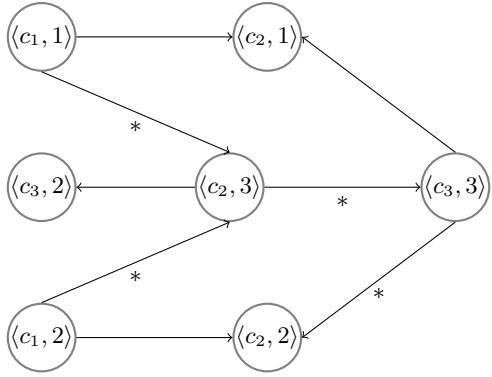
It turns out that there exists no inclusion relationship between the classes \textsc{wa} and \textsc{csafe} in either directions, i.e. 
\textsc{wa} $\not \subseteq $ $\textsc{csafe}$ (from example~\ref{eg:comparisonWAcsafe}),
and $\textsc{csafe} \not \subseteq$ \textsc{wa} (from the fact that \textsc{wa} $\subset$ \textsc{ja}, and example~\ref{eg:comparisonJAcsafe} below). 
Whereas \textsc{wa} $\subset$ \textsc{msafe}, since \textsc{wa} $\subset$ \textsc{mfa} and \textsc{mfa} $\equiv$ \textsc{msafe} (Theorem~\ref{msafeMFAequivalence}).
\subsection{Joint Acyclicity}
\noindent
Joint acyclicity~\cite{KR11jointacyc} extends weak acyclicity, by also taking into consideration the join between variables in body 
of $\forall\exists$ rules while analyzing the rules for acyclicity. The set \textsc{ja} represents the class of all ternary $\forall\exists$ rule sets
 that have the joint acyclicity property.
 A $\forall\exists$ rule set $\mathbb{P}$ is said to be \emph{renamed apart}, if for any $r \neq r'\in R$, $\var(r) \cap \var(r')=\emptyset$. 
Since any set of rules can be converted to an equivalent renamed apart one by simple variable renaming, we assume that any rule set $\mathbb{P}$ is renamed apart. 
Also for any $r \in \mathbb{P}$ and for a variable $y$, let $Pos_H^r(y)$ ($Pos_B^r(y)$) be the set of positions in which $y$ occurs in the head (resp. body) of $r$. 
For any $\forall\exists$ rule set $\mathbb{P}$ and an existentially quantified variable $y$ occurring in a rule in $\mathbb{P}$, we define $Mov_{\mathbb{P}}(y)$ 
as the least set with: 
\begin{itemize}
 \item $Pos_H^r(y) \subseteq Mov_{\mathbb{P}}(y)$, if $y$ occurs in $r$;
 \item $Pos_H^r(x) \subseteq Mov_{\mathbb{P}}(y)$, if $x$ is a universally quantified variable and $Pos_B^r(x) \subseteq Mov_{\mathbb{P}}(y)$;
\end{itemize}
for any $r \in \mathbb{P}$. 
The existential dependency graph of a (renamed apart) set of rules $\mathbb{P}$ is a graph whose nodes are the existentially quantified variables in $\mathbb{P}$.
There exists an edge from a variable $y$ to $y'$, if there is a rule $r \in \mathbb{P}$ in which $y'$ occurs and there exists a universally quantified variable $x$ in the 
head (and body) of $r$ such that $Pos_B^r(x)\subseteq Mov_{\mathbb{P}}(y)$. A $\forall\exists$ rule set $\mathbb{P}$ is jointly acyclic, iff its 
existential dependency graph is acyclic. Analyzing the containment relationships, it happens to be the case that \textsc{ja} $\not \subseteq$ \textsc{csafe} (since 
\textsc{wa} $\subset$ \textsc{ja}, and eg.~\ref{eg:comparisonWAcsafe}). Also example~\ref{eg:comparisonJAcsafe} shows us that 
\textsc{csafe} $\not \subseteq $ \textsc{ja}. However \textsc{ja} $\subset$ \textsc{msafe}, 
since \textsc{ja} $\subset$ \textsc{mfa} and \textsc{mfa} $\equiv$ \textsc{msafe} (Theorem~\ref{msafeMFAequivalence}).
\begin{example}\label{eg:comparisonJAcsafe}
 Consider the quad-system $QS_{\C}$ $=$ $\langle Q_{\C}, R \rangle$, where $Q_{\C}=\{c_1\colon(a,b,c)\}$.
  Suppose $R$ is the following set:
\[  R= \left \{
  \begin{array}{r}
     c_1\colon (x_{11},x_{12},z_1) \rightarrow c_2\colon(x_{11},x_{12},y_1) \hspace{0.4cm} (r_1)\\
     c_1\colon(x_{21}, x_{22}, z_2), c_2\colon (x_{22},x_{21},x_{23})  \rightarrow 
     c_3\colon(x_{21},x_{22},x_{23}) \hspace{.3cm} (r_2) \\
     c_3\colon (x_{31},x_{32}, x_{33}) \rightarrow c_1\colon(x_{33},x_{31},x_{32}) \hspace{0.1cm} (r_3)
  \end{array}
  \right \}\]
   Iterations during the dChase construction are:
  \vspace{-2pt}
 \begin{eqnarray}
  dChase_0(QS_{\C})=\{c_1\text{:}(a,b,c)\} \nonumber \\
  dChase_1(QS_{\C})=\{c_1\colon (a,b,c), c_2\colon (a,b,\_\colon b_1)\} \nonumber \\
  dChase(QS_{\C})=dChase_1(QS_{\C}) \nonumber 
 \end{eqnarray}
 Note that the lone Skolem blank node generated is $\_\colon b_1$, which do not have any descendants. Hence, 
 by definition $QS_{\C}$ is csafe (msafe/safe). Now analyzing the BRs for joint acyclicity, we note that
 for the only existentially quantified variable $y_1$, 
 
 $Mov_R(y_1)=\{\langle c_2, 3\rangle, \langle c_3, 3 \rangle, \langle c_1,1 \rangle \}$
 
\noindent\ Since the BR $r_1$ in which $y_1$ occurs contains the universally quantified variable $x_{11}$ in the head of $r_1$
 such that $Pos_B^{r_1}(x_{11}) \subseteq Mov_R(y_1)$, there exists a cycle from $y_1$ to $y_1$ itself in the existential dependency graph of $\tau(QS_{\C})$.
 Hence, by definition $\tau(QS_{\C})$ is not joint acyclic. Also since the class of weakly acyclic rules are contained in the class of jointly 
 acyclic rule, it follows that $\tau(QS_{\C})$ is also not weakly acyclic. 
\end{example}

\subsection{Model Faithful Acyclicity (MFA)}
\noindent MFA, proposed in Cuenca Grau et al.~\cite{Bernardo:dlacyclicity:2012}, is an acyclicity technique
that guarantees finiteness of chase and decidability of query answering, in the realm of $\forall\exists$ rules. The set \textsc{mfa}
denotes the class of all ternary $\forall\exists$ rule sets that are model faithfully acyclic. As far as we know, the 
MFA technique subsumes almost all other known techniques 
that guarantee a finite chase, in the $\forall\exists$ rules setting. Obviously,
 \textsc{wa} $\subset$ \textsc{ja} $\subset$ \textsc{mfa}.

For any $\forall\exists$ rule $r=\phi(r)(\vec x, \vec z)\rightarrow \psi(r)(\vec x, \vec y)$, for each $y_j \in \{\vec y\}$,
let $Y^j_r$ be a fresh unary predicate unique for $y_j$ and $r$; furthermore, let $S$ be a fresh binary predicate.
The transformation $\mathbf{mfa}$ of $r$ is defined as:
\begin{eqnarray}
\mathbf{mfa}(r)=\phi(r)(\vec x, \vec z)\rightarrow \psi(r)(\vec x, \vec y) \wedge  
\bigwedge_{y_j \in \{\vec y\}} \ [ Y^j_r(y_j) \wedge \bigwedge_{x_k \in \{\vec x\}} S(x_k, y_j)] \nonumber 
\end{eqnarray}
Also let $r_1$ and $r_2$ be two additional rules defined as:
\begin{eqnarray}
 S(x_1,z) \wedge S(z,x_2) \rightarrow S(x_1,x_2) \hspace{1cm} (r_1)\nonumber \\
 Y^j_r(x_1) \wedge  S(x_1, x_2) \wedge Y^j_r(x_2) \rightarrow \mathfrak{C} \hspace{1cm} (r_2) \nonumber 
\end{eqnarray}
where $\mathfrak{C}$ is a fresh nullary predicate. For any set of $\forall \exists$ rules $\mathbb{P}$, let $ad(\mathbb{P})$ be the union of
$r_1$ with the set  of rules obtained by instantiating $r_2$, for each $r \in \mathbb{P}$, for each existential variable $y_j$ in $r$.
For a set of $\forall\exists$ rules $\mathbb{P}$, $\mathbf{mfa}(\mathbb{P})$ $=$ $\bigcup_{r \in \mathbb{P}} \mathbf{mfa}(r) \cup ad(\mathbb{P})$. 
A $\forall\exists$ rule set $\mathbb{P}$ is said to be MFA, iff $\mathbf{mfa}(\mathbb{P}) \not \models_{\text{fol}} \mathfrak{C}$.
It was shown in Cuenca Grau et al.~\cite{Bernardo:dlacyclicity:2012} that if $\mathbb{P}$ is MFA, then $\mathbb{P}$ has
a finite chase, thus ensuring decidability of query answering.
The following theorem establishes the fact that the notion of msafety is equivalent to MFA, thanks to 
the polynomial time translations between quad-systems and ternary $\forall\exists$ rule sets.
\begin{theorem}\label{msafeMFAequivalence} Let $\tau$ be the translation function
from the set of unrestricted quad-systems to the set of ternary $\forall\exists$ rule sets, as defined in Definition~\ref{def:TranslationQuadsystemsToForall},
then, for any quad-system $QS_{\C}$ $=$ $\langle Q_{\C}$, $R \rangle$, $QS_{\C}$ is msafe iff $\tau(QS_{\C})$ is MFA. 
\end{theorem}
\begin{proof}(outline)
Recall that $\tau=\langle \tau_q$, $\tau_{br}\rangle$, where $\tau_q$ is the quad translation function and $\tau_{br}$ is 
the translation function from BRs to $\forall\exists$ rules. Also, $\tau(QS_{\C})$ $=$ $\tau_{br}(\{r_{Q_{\C}}\}$ $\cup$ $R)$.
Also, recall that for every blank node $b$ in $Q_{\C}$, the BR $r_{Q_{\C}}$ contains a corresponding existentially quantified variable $y_b$. 
We already saw that for such a transformation, the following property holds:
for any $m \in \mathbb{N}$, $\tau_q(dChase_m(QS_{\C}))$ $=$ $chase_m(\tau(QS_{\C}))$, and for any BR $r \in$ $R$ $\cup$ $\{r_{Q_{\C}}\}$, an assignment $\mu$, 
$applicable_{R \cup \{r_{Q_{\C}}\}}(r$, $\mu$, $dChase_m(QS_{\C}))$ iff $applicable_{\tau(QS_{\C})}($ $\tau_{br}(r)$, $\mu$, $chase_m(\tau(QS_{\C})))$. Also notice that
for any two blank nodes $\_\colon b_1$, $\_\colon b_2$, $S(\_\colon b_1$, $\_\colon b_2) \in$ $chase(\tau(QS_{\C}))$, iff $\_\colon b_1$ 
is a descendant of $\_\colon b_2$ in $dChase(QS_{\C})$. Hence, the relations $S$ and descendantOf are identical.

 Intuitively, MFA looks for cyclic
creation of a Skolem blank-node whose descendant is another Skolem blank-node that is generated by the same rule $r$ $=$ $body(r)(\vec x$, $\vec z)$ $\rightarrow$
$head(r)(\vec x$, $\vec y)$, by the same existential variable in $y_j \in \{\vec y\}$ of $r$. Wheras, msafety looks only  for generation 
of a Skolem blank-node $\_\colon b'$ whose descendant is another Skolem $\_\colon b$ using the same rule $r$. Hence, 
if $\tau(QS_{\C})$ is not MFA, then $QS_{\C}$ is not msafe, and consequently 
 onlyIf part of the theorem trivially holds.

(If part)
Suppose $QS_{\C}$ is unmsafe, and  
$\mu$ and $\mu'$ are the assignments applied on $r\in R$ to create Skolem blank nodes $\_\colon b$ and $\_\colon b'$, respectively, and suppose 
$\_\colon b$ is a descendant of $\_\colon b'$  in the $dChase(QS_{\C})$. That is 
$\_\colon b$ $=$ $\mu(y_j)$ and $\_\colon b'$ $=$ $\mu'(y_k)$, for $y_j, y_k \in \{\vec y\}$ of $r$. Suppose $j=k$, then the prerequisite of 
non-MFA is trivially satisfied. 
Suppose if $j\neq k$ is the case, then there exists $\_\colon b''$ in $dChase(QS_{\C})$ such that
$\_\colon b''$ $=$ $\mu'(y_j)$, since $\mu'$ is applied on $r$ and $y_j$ $\in$ $\{\vec y\}$. This means that also in this case, the prerequisite of 
non-MFA is satisfied. As a consequence $\tau(QS_{\C})$ is not MFA. Hence it follows that, $QS_{\C}$ is msafe iff $\tau(QS_{\C})$ is MFA.
\end{proof}
\noindent Let us revisit the quad-system $QS_{\C}$ in example~\ref{eg:strictContainmentMsafeToSafe}, it can be easily seen that $\tau(QS_{\C})$ is not MFA.
 Recall that we have seen that $QS_{\C}$ is safe but not msafe. 
We consider the Theorem~\ref{msafeMFAequivalence} to be of importance, as it not only establishes the equivalence of MFA and msafety, but
 thanks to it and the translation $\tau$, it can be deduced that the technique of safety, which we presented earlier, (strictly) extends the MFA technique. 
 As far as we know, the MFA class of
$\forall\exists$ rule sets is one of the most expressive class in the realm of $\forall \exists$ rule sets which allows a finite chase. 
Hence, the notion of safety that we propose can straightforwardly be ported to $\forall\exists$ settings. 
 The main difference between MFA and safety is that MFA only looks for cyclic
creation of two distinct Skolem blank-nodes $\_\colon b, \_\colon b'$ that are generated by the same rule $r$, by the same
existential variable in $r$. Whereas safety also takes into account the origin vectors $\vec a$ and $\vec a'$ used during rule application
to create $\_\colon b$ and $\_\colon b'$, 
respectively, and only raises an alarm if $\vec a \cong \vec a'$. Although, equivalence holds only between quad-systems and ternary $\forall\exists$ rule sets,
it can easily be noticed that the technique of safety can be applied to $\forall\exists$ rule sets of arbitrary arity, and can be used to extend currently established
tools and systems that work on existing notions of acyclicity such as WA, JA, or MFA.

\section{Related Work}\label{sec:related work}

\paragraph{Contexts and Distributed Logics} Work on contexts gained its attention as early as in the 80s, as 
McCarthy~\cite{mccarthy_generality_87} proposed context as a 
solution to the generality problem in AI. 
After this, various studies about logics of contexts mainly in the field of KR were done by Guha~\cite{guhas_thesis}, 
\emph{Distributed First Order Logics} by Ghidini et al.~\cite{Ghidini98distributedfirst} and 
\emph{Local Model Semantics} by Giunchiglia et al.~\cite{Ghidini01localmodels}. 
Primarily in these works, contexts were formalized as a first order/propositional
theory and bridge rules were provided to inter-operate the various theories of contexts. Some of the initial works on contexts
relevant to semantic web were the ones like \emph{Distributed Description Logics}~\cite{DDL} by Borgida et al.,
 and \emph{Context-OWL}~\cite{cowl} by Bouquet et al., and the work of 
 CKR~\cite{serafini2012contextualized,josephWomo} by Serafini et al. 
These were mainly logics based on DLs, which formalized contexts as OWL KBs, whose semantics is given using a 
distributed interpretation structure with additional semantic conditions that suits varying requirements. Compared to these works, 
the bridge rules we consider are much more expressive with conjunctions and existential variables that supports 
value/blank-node creation. 

\paragraph{Temporal RDF/Annotated RDF}
Studies in extending standard RDF with dimensions such as time and annotations have already been 
accomplished. Gutierrez et al. in \cite{GutierrezHV05}
tried to add a temporal extension to RDF and defines the notion of a `temporal rdf graph', in which
a  triple is augmented to a quadruple of the form $t\colon(s, p, o)$, where $t$ is a time point.
Whereas annotated extensions to RDF and querying annotated graphs
have been studied in  Udrea et al.~\cite{udrea-annotated-rdf-2010} and Straccia et al.~\cite{stra-etal-2010}. 
Unlike the case of time, here the quadruple has the form: $a\colon(s,p,o)$,
where $a$ is an annotation. The authors provide  semantics, inference rules and 
query language that allows for expressing temporal/annotated
queries. Although these approaches, in a way address contexts by means of time and annotations, 
the main difference in our work is that we provide the means to specify expressive
bridge rules for inter-operating the reasoning between the various contexts.

\paragraph{DL+rules} Works on extending DL KBs with Datalog like rules was studied by Horrocks et al.~\cite{SWRL} giving 
rise to the  SWRL~\cite{SWRL} language.  Related initiatives propose a formalism using which one can mix a DL ontology with
 the Unary/Binary Datalog RuleML sublanguages of the Rule Markup Language, and hence enables Horn-like rules
to be combined with an OWL KB. Since SWRL is undecidable in general, studies on 
computable sub-fragments gave rise to works like Description Logic Rules~\cite{DLrules}, where the authors
deal with rules that can be totally internalized by a DL knowledge base, and hence if the DL considered is 
decidable, then also is a DL+rules KB. The authors give various fragments of the rule bases like SROIQ rules, EL++ rules
etc. and show that certain new constructs that are not expressible by plain DL can be expressed using
rules, although they are finally internalized into DL KBs. Unlike in our scenario, these works consider only 
horn rules without existential variables.

 
\paragraph{$\forall\exists$ rules, TGDs, Datalog+- rules}
Query answering over rules with universal-existential quantifiers
in the context of databases, where these rules are called Datalog+- rules/tuple generating dependencies (TGDs),
was done by Beeri and Vardi~\cite{BeeriVImplicationProblem81} even in the early 80s, where the authors
show that the query entailment problem, in general, is undecidable. However, recently many classes of such rules 
have been identified for which query answering is decidable. These classes (according to \cite{{BagetLMS11}}) can broadly be divided into the following 
three categories: (i) \emph{bounded treewidth sets} (BTS), (ii) finite unification sets (FUS), and (iii) finite extension sets (FES). 
BTS contains the classes of $\forall\exists$ rule sets, whose models have bounded treewidth. Some of the important classes 
of these sets are the linear $\forall\exists$ rules~\cite{JohnsonK84},  (weakly) guarded rules~\cite{CaliGK08}, (weakly) frontier guarded rules~\cite{BagetLMS11}, and
jointly frontier guarded rules~\cite{KR11jointacyc}. BTS classes in general need not have a finite chase, and query answering is done by exploiting the 
fact that the chase is tree shaped, whose nodes (which are sets of instances)  start replicating (up to isomorphism) after a while. Hence, one 
could stop the computation of the chase, once it can be made sure that any future iterations of chase can only produce nodes that are isomorphic 
to existing nodes. A deterministic algorithm for deciding query entailment for the greedy BTS, which is a subset of this class is 
provided in Thomazo et al.~\cite{thomazo:lirmm-00763518}.

FUS classes include the class of `sticky' rules~\cite{DBLP:conf/rr/CaliGP10,CaliGP2012AIJ}, atomic hypothesis rules in which the body of each rule
contains only a single atom, and also the class of linear $\forall\exists$ rules. 
The approach used for query answering in  FUS classes is to rewrite the input query w.r.t. to the $\forall\exists$ rule sets
to another query  that can be evaluated directly on the set of instances, such that the answers for the former query
and latter query coincides. The approach is called the \emph{query rewriting approach}.  
Compared to approaches proposed in this paper, 
these approaches do not enjoy the finite chase property, and are hence not conducive to materialization/forward chaining 
based query answering.

Unlike BTS and FUS, the FES classes are characterized by the finite chase property, and hence are most related to the techniques proposed in our work.  
Some of the classes in this set employ termination guarantying checks called  `acyclicity tests' that  analyze 
the information flow between rules to check whether cyclic dependencies exists that can lead to infinite chase.  
\emph{Weak acyclicity}~\cite{Fagin05dataexchange,Deutsch03reformulationof},
was one of the first such notions, and was extended to \emph{joint acyclicity}~\cite{KR11jointacyc} and 
\emph{super weak acyclicity}~\cite{Marnette2009}.  The main approach used in these techniques is to exploit the 
structure of the rules and use a dependency graph that models the propagation path of constants across various predicates
in the rules, and restricting the dependency graph to be acyclic. The main drawback of these approaches is that
they only analyze the schema/Tbox part of the rule sets, and ignore the instance part, and hence produce a large number of false alarms, i.e. it is often the 
case that although dependency graph is cyclic, the chase is finite. Recently, a more dynamic approach, called the MFA technique, that also takes into account
the instance part of the rule sets was proposed in Cuenca grau et al.~\cite{{Bernardo:dlacyclicity:2012}}, where existence of cyclic 
Skolem blank-node/constant generations in the chase is detected by augmenting the rules with extra information that keeps track of the Skolem function 
used to generate each Skolem blank-node. As shown in section \ref{sec:Comparison}, our technique of safety subsumes the MFA technique,  
and supports for much more expressive rule sets, by also keeping track of the vectors used by rule bodies
while Skolem blank-nodes are generated.


\paragraph{Data integration} Studies in query answering on integrated heterogeneous databases with expressive integration
rules in the realm of data integration is primarily studied in the following two settings: 
(i) Data exchange~\cite{Fagin05dataexchange}, in which there is a source database and target database that are connected with 
existential rules, and 
(ii) Peer-to-peer data management systems (PDMS)~\cite{Halevy03schemamediation}, 
where there are an arbitrary number of peers that are interconnected using existential rules. 

The approach based on dependency graphs, for instance, is used by 
Halevi et al. in the context of peer-peer data management systems~\cite{Halevy03schemamediation}, and
decidability is attained by not allowing any kind of cycles in the peer topology.
Whereas in the context of Data exchange, 
WA is used in \cite{Fagin05dataexchange,Deutsch03reformulationof} to assure decidability, and the recent work by Marnette~\cite{Marnette2009}
employs the super weak acyclicity (SWA) to ensure decidability. It was shown in Cuenca Grau et al~\cite{Bernardo:dlacyclicity:2012}
that their MFA technique strictly subsumes both WA and SWA techniques in expressivity. Since we saw in section \ref{sec:Comparison}
that our technique of safety subsumes the MFA technique and allows the representation of much more expressive rule sets, the safety 
technique can straightforwardly be employed in the above mentioned systems with decidability guarantees for query answering.

\section{Summary and Conclusion}\label{section:conclusion}
\begin{table*}[t]
\centering
\begin{tabular}{| c | c | c | c |c |}
\hline
&&&&\\
Quad-System &Chase size w.r.t &  Data Complexity of \ & Combined Complexity & Complexity of \\
 Fragment & input quad-system \ & CCQ entailment & of CCQ entailment & Recognition  \\
\hline
Unrestricted Quad-Systems & Infinite & Undecidable & Undecidable & PTIME\\ 
Safe Quad-Systems & Double exponential  & PTIME-complete  & 2EXPTIME-complete & 2EXPTIME\\ 
MSafe Quad-Systems & Double exponential  & PTIME-complete  & 2EXPTIME-complete & 2EXPTIME\\ 
CSafe Quad-Systems & Double exponential  & PTIME-complete  & 2EXPTIME-complete & 2EXPTIME\\ 
RR Quad-Systems & Polynomial  & PTIME-complete & \ \ EXPTIME & PTIME\\ 
Restricted RR Quad-Systems & Polynomial & PTIME-complete & \ \ NP-complete & PTIME  \\
\hline
\end{tabular}
\caption{Complexity info for various quad-system fragments} 
\label{tab:compResults}
\end{table*}
\noindent In this paper, we study the problem of query answering over contextualized RDF knowledge in the presence of 
forall-existential bridge rules.
We show that the problem, in general, is undecidable, and present a few decidable classes of quad-systems.
Table \ref{tab:compResults} displays the complexity results of chase computation and query entailment for the various classes of 
quad-systems we have derived.  
Classes csafe, msafe, and safe, 
ensure decidability by restricting the structure of Skolem blank-nodes generated in the dChase. Briefly, the above classes do not allow an infinite descendant
chain for Skolem blank-nodes generated, by constraining each Skolem blank-node in a descendant chain to have a different value for 
certain attributes, whose value sets are finite. RR and restricted RR quad-systems, do not allow the generation of Skolem blank nodes, thus constraining
the dChase to have only constants from the initial quad-system. The above classes which suit varying situations, can be used to extend the currently established tools  
for contextual reasoning  to give support for expressive
bridge rules with conjunctions and existential quantifiers with decidability guarantees. From an expressivity point of view,
the class of safe quad-systems subsumes all the above classes, and other well known classes in the realm of $\forall\exists$ rules with finite chases.
We view the  results obtained in this paper as a general foundation 
for contextual reasoning and query answering over contextualized RDF knowledge formats such as quads, 
and can straightforwardly be used to extend existing quad stores.
\section{Acknowledgements}
\noindent We sincerely thank  Loris Bozzatto (FBK-IRST, Italy), and Francesco Corcoglionitti (FBK-IRST, Italy),
 and Prof. Roberto Zunino (DISI, University of Trento, Italy) for all their helpful technical feedbacks on an initial version of this paper.
  We also thank Dr. Christoph Lange (School of Computer Science, University of 
 Birmingham, UK),  Prof. Sethumadhavan (Center for Cyber Security, Amrita University, India), 
 and  Prof. Padmanabhan T.R. (Dept. of Computer Science, Amrita University, India), for their time and motivating discussions.

\appendix
\section{Proofs for Section \ref{sec:Query}}
\begin{proof}[Property \ref{prop:precLinearOrder}]
Note that a strict linear order is a relation that is irreflexive, transitive, and linear.

Irreflexivity: By contradiction, suppose $\prec_q$ is not irreflexive, then there exists $Q \in \mathcal{Q}$ such that 
$Q \prec_q Q$ holds. This means that neither of the conditions (i) and (ii) of $\prec_q$ definition holds for $Q$. Hence, due to condition
(iii) $Q \not \prec_q Q$, which is a contradiction.

Linearity: Note that for any two distinct $Q, Q' \in \mathcal{Q}$, one of the following holds: (a) $Q \subset Q'$, (b) $Q' \subset Q$, or 
(c) $Q\setminus Q'$ and $Q'\setminus Q$ are non-empty and disjoint. Suppose (a) is the case, then $Q \prec_q Q'$ holds.
Similarly, if (b) is the case then $Q' \prec_q Q$ holds. Otherwise if (c) is the case, then by condition (ii), either $Q \prec_q Q'$ or $Q' \prec_q Q$ should hold.
Hence, $\prec_q$ is a linear order over $\mathcal{Q}$.

Transitivity: Suppose there exists $Q, Q', Q'' \in \mathcal{Q}$ such that $Q \prec_q Q'$ and $Q' \prec_q Q''$. Then, one of the following
four cases hold: (a) $Q \prec_q Q'$ due to (i) and $Q' \prec_q Q''$ due to (i), (b) $Q \prec_q Q'$ due to (i) and $Q' \prec_q Q''$ due to (ii),
(c) $Q \prec_q Q'$ due to (ii) and $Q' \prec_q Q''$ due to (i), (d) $Q \prec_q Q'$ due to (ii) and $Q' \prec_q Q''$ due to (ii). 

Suppose if (a) is the case, then trivially $Q \subset Q''$, and hence by applying condition (i) $Q \prec_q Q''$. Otherwise if
(b) is the case, then either (1) $Q \subset Q''$ or (2) $Q \not \subset Q''$. Suppose, (1) is the case then, by (i) $Q \prec_q Q''$.
Otherwise, if (2) is the case, then since, $Q \subset Q'$,  it cannot be the case that greatestQuad$_{\prec_l}(Q''\setminus Q)$ $\prec_l$ 
greatestQuad$_{\prec_l}(Q''\setminus Q')$, and it cannot be the case that greatestQuad$_{\prec_l}(Q'\setminus Q'')$ $\prec_l$ 
greatestQuad$_{\prec_l}(Q\setminus Q'')$. Hence, it should be the case that
greatestQuad$_{\prec_l}(Q''\setminus Q')$ $\preceq_l$ greatestQuad$_{\prec_l}(Q''\setminus Q)$ and 
greatestQuad$_{\prec_l}(Q\setminus Q'')$ $\prec_l$ greatestQuad$_{\prec_l}(Q'\setminus Q'')$. But since, greatest-\\
Quad$_{\prec_l}(Q'\setminus Q'')$ $\prec_l$ greatestQuad$_{\prec_l}(Q''\setminus Q')$, it allows us to derive
greatestQuad$_{\prec_l}(Q$ $\setminus$ $Q'')$ $\prec_l$ greatestQuad$_{\prec_l}(Q''$ $\setminus$ $Q)$, and hence by condition (ii), $Q \prec_q Q''$. 
Hence, if (b) is the case, then in both possible cases (1) or (2), it should be the case that $Q \prec_q Q''$.  
Otherwise if (c) is the case, then similar to the arguments in (b), by condition (i) or (ii), it can easily be seen that $Q \prec_q Q''$.

Otherwise, if (d) is the case, then the following must hold: greatestQuad$_{\prec_l}(Q\setminus Q')$ $\prec_l$ 
greatestQuad$_{\prec_l}(Q'\setminus Q)$ ($\dagger$) and greatestQuad$_{\prec_l}(Q'\setminus Q'')$ $\prec_l$ 
greatestQuad$_{\prec_l}($ $Q''\setminus Q')$ ($\ddagger$).
Suppose by contradiction $Q''$ $\prec_q$ $Q$, then one of the following holds:
(1) $Q''$ $\prec_q$ $Q$ by condition (i) or (2) $Q''$ $\prec_q$ $Q$ by condition (ii). Suppose, (1) is the case, then 
it should be the case that $Q'' \subset Q$. Hence, it should  not be the case that greatestQuad$_{\prec_l}($ $Q \setminus Q')$ $\prec_l$ 
greatestQuad$_{\prec_l}(Q''\setminus Q')$ and it should not be the case that greatestQuad$_{\prec_l}(Q'\setminus Q'')$ $\prec_l$ 
greatestQuad$_{\prec_l}(Q' \setminus Q)$.  Hence, it should be the case that 
greatestQuad$_{\prec_l}(Q''\setminus Q')$ $\preceq_l$ greatestQuad$_{\prec_l}(Q \setminus Q')$ ($\heartsuit$), 
and it should be the case that 
greatestQuad$_{\prec_l}($ $Q' \setminus Q)$ $\preceq_l$ greatestQuad$_{\prec_l}(Q'\setminus Q'')$ ($\spadesuit$). 
Applying ($\ddagger$) in $(\heartsuit)$, we get 
greatestQuad$_{\prec_l}(Q'\setminus Q'')$ $\prec_l$ greatestQuad$_{\prec_l}(Q \setminus Q')$,
and Applying ($\dagger$) in $(\spadesuit)$, we get 
greatestQuad$_{\prec_l}($ $Q \setminus Q')$ $\prec_l$ greatestQuad$_{\prec_l}(Q'\setminus Q'')$, which is a contradiction. 
Suppose if (2) is the case, then greatestQuad$_{\prec_l}(Q''\setminus Q)$ $\prec_l$ greatestQuad$_{\prec_l}($ $Q\setminus Q'')$. 
The above can be written as: 
greatestQuad$_{\prec_l}($ $Q''\setminus (Q \cap Q''))$ $\prec_l$ greatestQuad$_{\prec_l}(Q\setminus (Q \cap Q''))$.
Using $Q \cap Q' \cap Q'' \subseteq Q \cap Q'$, it follows that 
greatestQuad$_{\prec_l}(Q''\setminus (Q \cap Q' \cap Q''))$ $\preceq_l$ greatestQuad$_{\prec_l}($ $Q\setminus (Q \cap Q' \cap Q''))$ ($\clubsuit$).
Also applying similar transformation in $(\dagger)$ and $(\ddagger)$, we get 
greatestQuad$_{\prec_l}(Q\setminus (Q \cap Q' \cap Q''))$ $\preceq_l$ greatestQuad$_{\prec_l}($ $Q'\setminus (Q \cap Q' \cap Q''))$, 
and 
greatestQuad$_{\prec_l}(Q'\setminus (Q \cap Q' \cap Q''))$ $\preceq_l$ greatestQuad$_{\prec_l}(Q''\setminus (Q \cap Q' \cap Q''))$.
From which, it follows that 
greatestQuad$_{\prec_l}(Q\setminus (Q \cap Q' \cap Q''))$ $\preceq_l$ greatestQuad$_{\prec_l}(Q''\setminus (Q \cap Q' \cap Q''))$.
Using ($\clubsuit$) in the above, we get 
greatestQuad$_{\prec_l}(Q\setminus (Q \cap Q' \cap Q''))$ $=$ greatestQuad$_{\prec_l}(Q'\setminus (Q \cap Q' \cap Q''))$ $=$
greatestQuad$_{\prec_l}(Q''\setminus (Q \cap Q' \cap Q''))$, which is a contradiction. Hence, it should be the case that 
$Q \prec_q Q''$.
\end{proof}
\begin{proof}[Theorem \ref{theorem:undecidable}]
We show that CCQ entailment is undecidable for
unrestricted quad-systems, by showing that the well known undecidable problem of 
``non-emptiness of intersection of context-free grammars'' is reducible to the CCQ answering
problem. 

Given an alphabet $\Sigma$,  string $\vec w$ is a sequence of symbols from $\Sigma$. A language $L$ is a 
subset of $\Sigma^*$, where $\Sigma^*$ is the set of all strings that can be constructed
from the alphabet $\Sigma$, and also includes the empty string $\epsilon$. Grammars are machineries that generate a particular language.
A grammar $G$ is a quadruple $\langle V, T, S, P \rangle$, where $V$ is the set of variables, $T$, the set of terminals,
$S \in V$ is the start symbol, and $P$ is a set of production rules (PR), in which each PR $r \in P$, is of the form:
\[
\vec w \rightarrow \vec w'
\]
where $\vec w, \vec w' \in \{T \cup V\}^*$. Intuitively application of a PR $r$ of the form above on a string $ \vec w_1$, 
replaces every occurrence of the sequence $\vec w$ in $\vec w_1$ with $\vec w'$. PRs are applied starting from the 
start symbol $S$ until it results in a string $\vec w$, with $\vec w \in \Sigma^*$ or no more production 
rules can be applied on $\vec w$. In the former case, we say that $\vec w \in L(G)$, the language generated by grammar $G$. For
a detailed review of grammars, we refer the reader to Harrison et al. ~\cite{Harrison-formal-language}.
A \emph{context-free grammar} (CFG) is a grammar, whose set of PRs $P$, have the following property: 
\begin{property}\label{prop:CFG}
For a CFG, every PR is of the form $v \rightarrow \vec w$, where $v \in V$, $\vec w \in \{T \cup V\}^*$.
\end{property}
 Given two CFGs, $G_1=\langle V_1, T, S_1, P_1 \rangle$ and $G_2=\langle V_2, T, S_2, P_2 \rangle$, 
where $V_1, V_2$ are the set of variables, 
$T$ such that $T \cap (V_1 \cup V_2)=\emptyset$ is the set of terminals. $S_1 \in V_1$ is the start symbol of $G_1$,
and $P_1$ are the set of PRs of the form $v \rightarrow \vec w$, where $v \in V$, $\vec w$ is a sequence of the form
$w_1...w_n$, where $w_i \in V_1 \cup T$. $S_2, P_2$ are defined similarly.  Deciding whether the language
generated by the grammars $L(G_1)$ and $L(G_2)$ have non-empty intersection is known 
to be undecidable~\cite{Harrison-formal-language}.

Given two CFGs, $G_1=\langle V_1, T, S_1, P_1 \rangle$ and $G_2=\langle V_2, T, S_2, P_2 \rangle$,
we encode grammars $G_1, G_2$ into a quad-system of the form $QS_c=\langle Q_{c},R\rangle$, with a single context identifier $c$.
Each PR $r= v \rightarrow \vec w \in P_1 \cup P_2$, with $\vec w=w_1w_2w_3..w_n$, is encoded as a BR of the form:
\begin{eqnarray}\label{eqn:PRtoICR}
&& c\colon(x_1,w_1,x_2), c\colon(x_2, w_2, x_3),...,c\colon(x_n,w_n,x_{n+1}) 
 \rightarrow c\colon(x_1,v,x_{n+1}) 
\end{eqnarray}
where $x_1,..,x_{n+1}$ are variables. W.l.o.g. we assume that the set of terminal symbols $T$ is
equal to the set of terminal symbols occurring in $P_1 \cup P_2$. For each terminal symbol $t_i \in T$, $R$ contains
a BR of the form:
\begin{eqnarray}\label{eqn:existential rule}
&& c\colon(x, \texttt{rdf:type}, C) \rightarrow \exists y \ c\colon(x,t_i,y), 
 c\colon(y, \texttt{rdf:type},C) 
\end{eqnarray}
and $Q_c$ contains only the triple:
\[
 c\colon(a,\texttt{rdf:type},C)
\]
We in the following show that:
\begin{eqnarray}
QS_c \models \exists y \ c\colon(a , S_1, y) \wedge c\colon(a, S_2, y) \leftrightarrow 
L(G_1)\cap L(G_2)\neq\emptyset 
\end{eqnarray}

\begin{claim}(1)
 For any $\vec w=t_1,...,t_p \in T^*$, there exists $b_1,...b_p$, such that $c\colon(a,t_1,b_1)$, $c\colon(b_1,t_2,b_2)$, ..., 
 $c\colon(b_{p-1}, t_p, b_p)$, $c\colon(b_p,\texttt{rdf:type},C)$ $\in$ $dChase($ $QS_c)$.
\end{claim}
we proceed by induction on $|\vec w|$.
\begin{description}
 \item [base case] suppose if $|\vec w|=1$, then $\vec w=t_i$, for some $t_i \in T$. But
 by construction $c\colon(a$, \texttt{rdf:type}, $C)$ $\in$ $dChase_0(QS_c)$, on which rules of the 
 form (\ref{eqn:existential rule}) is applicable. Hence, there exists an $i$ such that
 $dChase_i(QS_c)$ contains $c\colon(a,t_i,b_i)$, $c\colon(b_i,\texttt{rdf:type},C)$, for each $t_i \in T$. 
 Hence, the base case.
 \item [hypothesis] for any $\vec w=t_1...t_p$, if $|\vec w|\leq p'$, then there exists $b_1,...,b_p$, such that 
 $c\colon(a,t_1,b_1)$, $c\colon(b_1,t_2,b_2)$, ..., $c\colon(b_{p-1}, t_p, b_p)$, $c\colon(b_p$, \\
 \texttt{rdf:type}, $C)$ $\in$ $dChase(QS_c)$.
 \item [inductive step] suppose $\vec w=t_1...t_{p+1}$, with $|\vec w|\leq p'+1$. Since $\vec w$ can be written
 as $\vec{w'}t_{p+1}$, where $\vec w'=t_1...t_p$, and by hypothesis, there exists $b_1,...,b_p$ such that 
   $c\colon(a,t_1,b_1)$, $c\colon(b_1,t_2,b_2)$, $...$, $c\colon(b_{p-1}, t_p, b_p)$, $c\colon(b_p,\texttt{rdf:type},C)$ $\in$ 
 $dChase(QS_c)$. Also since rules of the form (\ref{eqn:existential rule}) are applicable on $c\colon(b_p$, \texttt{rdf:type}, $C)$,
 and hence produces triples of the form $c\colon(b_p,t_i,b_{p+1}^i)$, $c\colon(b_{p+1}^i$, \texttt{rdf:type}, $C)$, 
 for each $t_i \in T$. Since $t_{p+1} \in T$, the claim follows.
\end{description}
 For a grammar $G=\langle V,T,S,P\rangle$, whose start symbol is $S$, and for any $\vec w \in \{V \cup T\}^*$,
 for some $V_j \in V$, we denote by $V_j \rightarrow^i \vec w$, 
 the fact that $\vec w$ was derived from $V_j$ by $i$ production steps, i.e. there exists steps 
 $V_j \rightarrow r_1, ...,r_i \rightarrow \vec w$, which lead to the production of $\vec w$. 
For any $\vec w$, $\vec w \in L(G)$, iff 
there exists an $i$ such that $S \rightarrow^i \vec w$. For any $V_j \in V$, we use $V_j \rightarrow^* \vec w$ to denote
the fact that there exists an arbitrary $i$, such that $V_j \rightarrow^i \vec w$.
\begin{claim}(2)
 For any $\vec w=t_1...t_p \in \{V \cup T\}^*$, and for any $V_j \in V$, if $V_j \rightarrow^* \vec w$
and there exists $b_1,...,b_{p+1}$, with $c\colon(b_1,t_1,b_2), ..., c\colon(b_p,t_p,b_{p+1}) \in dChase(QS_c)$,
then $c\colon(b_1,V_j,b_{p+1}) \in dChase(QS_c)$. 
\end{claim}
We prove this by induction on the size of $\vec w$.
\begin{description}
 \item[base case] Suppose $|\vec w|=1$, then $\vec w=t_k$, for some $t_k \in T$. If there exists
 $b_1,b_2$ such that $c\colon(b_1,t_k,b_2)$. But since there exists a PR $V_j \rightarrow t_k$,  by 
 transformation given in (\ref{eqn:PRtoICR}), there exists a BR $c\colon(x_1,t_k,x_2) \rightarrow c\colon(x_1,V_j,x_2)
 \in R$, which is applicable on $c\colon(b_1,t_k,b_2)$ and hence the quad $c\colon(b_1,V_j,b_2) \in dChase(QS_c)$.
 \item[hypothesis] For any $\vec w=t_1...t_p$, with 
 $|\vec w|\leq p'$, and for any $V_j \in V$, if $V_j \rightarrow^* \vec w$
and there exists $b_1,...b_p,b_{p+1}$, such that $c\colon(b_1,t_1,b_2)$, $...$, $c\colon(b_p,t_p,b_{p+1})$ $\in$ $dChase(QS_c)$,
then $c\colon(b_1$, $V_j$, $b_{p+1})$ $\in$ $dChase(QS_c)$.
 \item[inductive step] Suppose if $\vec w=t_1...t_{p+1}$, with $|\vec w|\leq p'+1$, and  $V_j \rightarrow^i \vec w$, 
and there exists $b_1,...b_{p+1}$, $b_{p+2}$, such that $c\colon(b_1,t_1,b_2)$, $...$, 
$c\colon(b_{p+1}$, $t_{p+1}$, $b_{p+2})$ 
 $\in$ $dChase(Q_c)$.  Also, one of the following holds (i) $i=1$, or (ii) $i>1$. Suppose (i) is the case, 
 then it is trivially the case that $c\colon(b_1,V_j,b_{p+2}) \in dChase(QS_c)$. Suppose if (ii) is the case,
 one of the two sub cases holds (a) $V_j \rightarrow^{i-1} V_k$, for some $V_k \in V$ and $V_k \rightarrow^1 \vec w$ or (b)
   there exist a $V_k \in V$, such that $V_k \rightarrow^*  t_{q+1}...t_{q+l}$, with $2\leq l \leq p$, where
 $V_j \rightarrow^* t_1...t_qV_kt_{p-l+1}...t_{p+1}$. If (a) is the case, trivially
 then $c\colon(b_1,V_k,b_{q+2}) \in dChase(QS_c)$, and since by construction there exists $c\colon(x_0, V_k, x_1)$ $\rightarrow$ 
 $c\colon(x_0,V_{k+1},x_1)$, $...$, $c\colon(x_0,V_{k+i},x_1)$ $\rightarrow$ $c\colon(x_0,V_j,x_1)$ $\in$ $R$,
$c\colon(b_1,V_j,b_{q+2}) \in dChase($ $QS_c)$.
If (b) is the case, then since $|t_{q+1}...t_{q+l}|\geq 2$, 
 $|t_1...t_qV_2t_{p-l+1}...t_{p+1}|\leq p'$. This implies that  
 $c\colon(b_1,V_j,b_{p+2}) \in dChase(QS_c)$.
\end{description}
Similarly, by construction of $dChase(QS_c)$, the following claim can straightforwardly be shown to hold:
\begin{claim}(3)
 For any $\vec w=t_1...t_p \in \{V \cup T\}^*$, and for any $V_j \in V$, if there exists 
 $b_1,...,b_p,b_{p+1}$, with $c\colon(b_1,t_1,b_2)$, ..., $c\colon(b_p,t_p,b_{p+1})$ $\in$ $dChase(QS_c)$ 
 and $c\colon(b_1$, $V_j$, $b_{p+1})$ $\in$ $dChase(QS_c)$, then $V_j \rightarrow^* \vec w$. 
\end{claim}
(a) For any $\vec w=t_1...t_p \in T^*$, if $\vec w \in L(G_1) \cap L(G_2)$, then
by claim 1, since there exists $b_1,...,b_p$, such that $c\colon(a,t_1,b_1),...,c\colon(b_{p-1},t_p,b_p) \in dChase(QS_c)$. But
since $\vec w \in L(G_1)$ and $\vec w \in L(G_2)$, $S_1 \rightarrow \vec w$ and $S_2 \rightarrow \vec w$. Hence 
by claim 2, $c\colon(a,S_1,b_p),c\colon(a,S_2,b_p)$ $\in$ $dChase(QS_c)$, which implies that $dChase(QS_c)$ $\models$ 
$\exists y$  $c\colon(a,s_1,y)$ $\wedge$ $c\colon(a,s_2,y)$. Hence, by Theorem \ref{dChaseUniversalModelProperty}, $QS_c$ $\models$  
$\exists y$ $c\colon(a,s_1,y)$ $\wedge$ $c\colon(a,s_2,y)$.  \\
(b) Suppose if $QS_c \models  \exists y \ c\colon(a,S_1,y) \wedge c\colon(a,S_2,y)$, then applying Theorem \ref{dChaseUniversalModelProperty},
it follows that 
there exists $b_p$ such that $c\colon(a$, $S_1$, $b_p)$, $c\colon(a,S_2,b_p) \in dChase(QS_{\C})$. 
Then it is the case that there exists $\vec w=t_1...t_p \in T^*$, and
$b_1,...,b_p$ such that $c\colon(a,t_1,b_1)$, ..., $c\colon(b_{p-1},t_p,b_p)$, $c\colon(a,S_1,b_p)$, 
$c\colon(a,S_2,b_p)$ $\in$ $dChase(QS_c)$. Then by claim 3,
$S_1 \rightarrow^* \vec w$, $S_2 \rightarrow^* \vec w$. Hence, $w \in L(G_1) \cap L(G_2)$. 

By (a),(b) it follows that there exists $\vec w \in L(G_1) \cap L(G_2)$ 
iff $QS_c \models \exists y \ c\colon(a,s_1,y) \wedge c\colon(a,s_2,y)$.
As we have shown that the intersection of CFGs, which is an undecidable problem, is reducible to the problem of
query entailment on unrestricted quad-system, the latter is undecidable.
\end{proof}

\section{Proofs for Section \ref{sec:safe}}

\begin{proof}[Theorem \ref{theorem:soundnessCompleteness}]
We in the following show the case of $dChase^{\text{csafe}}(QS_{\C})$, i.e. \textbf{unCSafe} $\in dChase^{\text{csafe}}(QS_{\C})$ iff 
$QS_{\C}$ is uncsafe. The proof follows from Lemma \ref{lemma:Soundness} and Lemma \ref{lemma:Completeness} below. 

The proofs for the case of $dChase^{\text{safe}}(QS_{\C})$ and $dChase^{\text{msafe}}(QS_{\C})$ is similar, and is omitted.

\end{proof}

\begin{lemma}[Soundness]\label{lemma:Soundness}
 For any quad-system $QS_{\C}=\langle Q_{\C}, R\rangle$, if the quad \textbf{unCSafe} $\in dChase^{\text{csafe}}(QS_{\C})$,
then $QS_{\C}$  is uncsafe. 
\end{lemma}
\begin{proof}
Note that $augC(R)= \bigcup_{r \in R} augC(r)$ $\cup$ $\{brTR\}$, where $brTR$ is 
the range restricted BR $c_c\colon (x_1$, descendantOf, $z)$, $c_c\colon (z$, descendantOf, $x_2)$ $\rightarrow$
$c_c\colon (x_1$, descendantOf, $x_2)$. Also for each $r \in R$, $body(r)=body(augC(r))$, and
for any $c \in \C$, $c\colon (s,p,o) \in head(r)$ iff $c\colon (s,p,o) \in head(augC(r))$. That 
is, $head(r)$ $=$ $head(au$- $gC(r))(\C)$, where $head($ $r)(\C)$ denotes the quad-patterns in $head(r)$, whose 
context identifiers is in $\C$. Also, $head(augC(r))$ $=$ $head(augC(r))(\C)$ $\cup$ $head(augC(r))(c_c)$, and also
the set of existentially quantified variables in $head(augC(r))(c_c)$ is contained in the set of existentially quantified 
variables in $head(augC(r))(\C)$ ($\dagger$). We first prove the following claim:
\begin{claim}(0)
For any quad-system $QS_{\C}$ $=$ $\langle Q_{\C}, R \rangle$, let $i$ be a csafe dChase iteration, let $j$ be the number of csafe dChase iterations before $i$ in 
which $brTR$ was applied, then $dChase_{i-j}(QS_{\C})$ $=$ $dChase^{\text{csafe}}_{i}(QS_{\C})(\C)$. 
\end{claim}
We approach the proof of the above claim by induction on $i$.
\begin{description}
 \item[base case] If $i=1$, then $dChase^{\text{csafe}}_0(QS_{\C})(c_c)$ $=$ $\emptyset$ and $dChase^{\text{csafe}}_0(QS_{\C})(\C)$ $=$
$dChase^{\text{csafe}}_0(QS_{\C})$ $=$ $dChase_0(QS_{\C})$. Hence, it should be the case that $applicab$- $le_{augC(R)}(brTR$, $\mu$, $dChase^{\text{csafe}}_0(QS_{\C}))$
does not hold, for any $\mu$. Hence, $applicab$- $le_R($ $r$, $\mu$, $dChase_0(QS_{\C}))$ iff
$applicable_{augC(R)}($ $augC(r)$, $\mu$, $dChase^{\text{csafe}}_0(QS_{\C}))$, for any $r \in R$, assignment $\mu$. Also using $(\dagger)$, 
it follows that $dChase_1(QS_{\C})$ $=$ $dChase^{\text{csafe}}_{1-0}(QS_{\C})(\C)$.
\item[hypothesis] for any $i \leq k$, if $i$ is a csafe dChase iteration, and $j$ be the number of csafe dChase iterations
before $i$ in which $brTR$ was applied, then $dChase_{i-j}(QS_{\C})$ $=$ $dChase^{\text{csafe}}_{i}(QS_{\C})(\C)$.
\item[inductive] suppose $i=k+1$, then one of the following three cases should hold: (a) $applicable_{augC(R)}(r$, $\mu$, $dChase^{\text{csafe}}_k(QS_{\C}))$
does not hold for any $r \in augC(R)$, assignment $\mu$, and $dChase^{\text{csafe}}_{k+1}(QS_{\C})$ $=$ 
$dChase^{\text{csafe}}_k(QS_{\C})$, or (b) $applicable_{augC(R)}($ $brTR$, $\mu$, $dChase^{\text{csafe}}_k(QS_{\C}))$ holds, for some assignment $\mu$, 
or (c) $applicable_{augC(R)}(r$, $\mu$, $dChase^{\text{csafe}}_k(QS_{\C}))$ holds, for some $r$ $\in$ $augC(R)$ $\setminus$ $\{brTR\}$, for some assignment $\mu$.
If (a) is the case, then it should be the case that $applicable_R(r'$, $\mu$, $dChase_{k-j}(QS_{\C}))$ does not hold, 
for any $r' \in R$, assignment $\mu$. As a result $dChase_{k+1-j}(QS_{\C})$ $=$ $dChase_{k-j}(QS_{\C})$, and hence,
$dChase_{k+1-j}($ $QS_{\C})$ $=$ $dChase^{\text{csafe}}_{k+1}(QS_{\C})(\C)$. If (b) is the case, then since $dChase^{\text{csafe}}_{k+1}(QS_{\C})(\C)$
$=$ $dChase^{\text{csafe}}_{k}($ $QS_{\C})(\C)$, $dChase^{\text{csafe}}_{k+1}(QS_{\C})(\C)$ $=$ $dChase_{k+1-j-1}($ $QS_{\C})$ $=$ \linebreak
$dChase_{k-j}(QS_{\C})$. If (c) is the case, then it should the case that $applicable_R(r'$, $\mu$, $dChase_{k-j}(QS_{\C})$, 
where $r$ $=$ $augC(r')$ and $head(r)(\C)$ $=$ $head(r)$. Hence, it should be the case that $dChase^{\text{csafe}}_{k+1}(QS_{\C})(\C)$
$=$ $dChase_{k+1-j}($ $QS_{\C})$.
\end{description}
The following claim, which straightforwardly follows from claim 0, shows that any quad $c\colon(s,p,o)$, 
with $c \in \C$  derived in csafe dChase, is also derived in its standard dChase. In this way, csafe dChase do not
generate any unsound triples in any context $c \in \C$. 
\begin{claim}(1)
 For any quad $c\colon(s,p,o)$, where $c \in \C$, 
if $c\colon(s,p,o)$ $\in$ $dChase^{\text{csafe}}(QS_{\C})$, then $c\colon(s,p,o) \in dChase(QS_{\C})$.
\end{claim}
The following claim shows that the set of origin context quads are also sound. 
\begin{claim}(2) If there exists quad $c_c\colon(b, \text{originContext}, c)$ 
$\in$ $dChase^{\text{csafe}}(QS_{\C})$, then $c$ $\in$ $originContexts(b)$.  
\end{claim}
If $c_c\colon(b$, originContext, $c)$ $\in$ $dChase^{\text{csafe}}(QS_{\C})$, there exists $i\in \mathbb{N}$,
such that  $c_c\colon(b$, originContext, $c)$ $\in$ $dChase^{\text{csafe}}_i($ $QS_{\C})$ and there exists no $j<i$ with 
$c_c\colon(b$, originContext, $c)$ $\in$ $dChase^{\text{csafe}}_j(QS_{\C})$. But if 
$c_c\colon(b$, originContext, $c)$ $\in$ $dChase^{\text{csafe}}_i(QS_{\C})$ implies that there exists an 
 $augC(r)$ $=$ $body(\vec x$, $\vec z)$ $\rightarrow$ $head(\vec x, \vec y)$ $\in$ $augC(R)$, 
with $c_c\colon(y_j$, originContext, $c)\in$ $head(\vec x$, $\vec y)$, $y_j \in \{\vec y\}$, 
such that  $c_c\colon(b$, originContext, $c)$ was generated due to application of an assignment $\mu$ on $augC(r)$, with 
$b=y_j[\mu^{ext(\vec y)}]$. 
This implies that there exists $c\colon(s,p,o)$ $\in$ $head(\vec x,\vec y)$, with $s$ $=$ $y_j$ or $p$ $=$ $y_j$
or $o$ $=$ $y_j$, $c \in \C$. Since according to our assumption, $i$ is the first iteration in which $c_c\colon(b$, originContext, $c)$
is generated, it follows that $i$ is the first iteration in which $c\colon(s,p,o)[\mu^{ext(\vec y)}]$ is also generated. 
Let $k$ be the number of iterations before $i$ in which $brTR$ was applied. By applying claim 0, it should be the case 
that $c\colon (s,p,o)[\mu^{ext(\vec y)}]$ $\in dChase_{i-k}(QS_{\C})$, and $i-k$ should be the first such dChase iteration. 
Hence, $c \in orginContexts(b)$.

\noindent In the following claim, we prove the soundness of the descendant quads generated in a safe dChase.
\begin{claim}(3)
 For any two distinct blank nodes $b, b'$ in $dChase^{\text{csafe}}(QS_{\C})$, if
$c_c\colon$ $(b'$, \text{descendantOf}, $b)$ $\in$ $dChase^{\text{csafe}}(QS_{\C})$ then $b'$ is a descendant of $b$.
\end{claim}
Since any quad of the form $c_c\colon(b'$, descendantOf, $b)$ $\in$ $dChase^{\text{csafe}}(QS_{\C})$
is not an element of $Q_{\C}$, and can only be introduced by an application of a BR $r\in augC(R)$,
any quad of the form $c_c\colon(b'$, descendantOf, $b)$ can only be introduced, earliest in the first iteration of 
$dChase^{\text{csafe}}(QS_{\C})$. 
Suppose $c_c\colon(b'$, descendantOf, $b)$ $\in$ $dChase^{\text{csafe}}(QS_{\C})$, then there exists an iteration 
$i\geq 1$ such that  $c_c\colon(b'$, $\text{descendantOf}$, $b)$ $\in$ $dChase^{\text{csafe}}_j(QS_{\C})$, for any $j\geq i$, 
and $c_c\colon(b'$, $\text{descendantOf}$, $b)$ $\not \in$ $dChase^{\text{csafe}}_{j'}(QS_{\C})$, for any $j'<i$. 
We apply induction on $i$ for the proof. 
\begin{description}
 \item[base case] suppose $c_c\text{:}(b'$, $\text{descendantOf}$, $b)$ $\in$ $dChas$- -$e^{\text{csafe}}_1($ $QS_{\C})$ and since $b$ $\neq$ $b'$,
 then there exists a BR $r$ $\in$ $augC(R)$, $\exists \mu$ such that  $applicable_{augC(R)}($ $r$, $\mu$, $dChase^{\text{csafe}}_0(QS_{\C}))$, i.e.
 $body(r)(\vec x$, $\vec z)[\mu]$ $\subseteq$ $dChase^{\text{csafe}}_0(QS_{\C})$ and  $c_c\colon(b'$, descendantOf, $b)$ 
 $\in$ $head(r)(\vec x, \vec y)[\mu^{ext(\vec y)}]$.
Then by construction of $augC(r)$, it follows that $b=y_j[\mu^{ext(\vec y)}]$, for some 
$y_j \in \{\vec y\}$   
and $b'=\mu(x_i)$, for some $x_i \in \{\vec x\}$. Since $dChase_0(QS_{\C})$ $=$ $dChase^{\text{csafe}}_0(QS_{\C})$, it follows using ($\dagger$) that  
$applicable_R(r', \mu, dChas$- -$e_0(QS_{\C}))$ holds, for $r'$ $=$ $body(r')(\vec x$, $\vec z)$ $\rightarrow$ $head(r')(\vec x$, $\vec y)$, 
with $augC(r')$ $=$ $r$. Hence, by construction, it follows that $b=y_j[\mu^{ext(\vec y)}] \in \const(dChase_1(QS_{\C}))$, for 
$y_j \in \{\vec y\}$   
and $b'=\mu(x_i)$, for $x_i \in \{\vec x\}$. Hence $b'$ is a descendant of $b$ (by definition).
 \item[hypothesis] if $c_c\colon(b'$, $\text{descendantOf}$, $b)$ $\in$ $dChase^{\text{csafe}}_i($ $QS_{\C})$, 
 for $1 \leq i \leq k$, then $b'$ is a descendant of $b$.
 \item[inductive step] suppose $c_c\colon(b'$, descendantOf, $b)$ $\in$ $dChase^{\text{csafe}}_{k+1}(QS_{\C})$, then 
 either (i) $c_c\colon(b'$, descendantOf, $b)$
$\in$ $dChase^{\text{csafe}}_{k}(QS_{\C})$ or (ii) $c_c\colon(b'$, descendantOf, $b)$ $\not\in$ 
$dChase^{\text{csafe}}_{k}(QS_{\C})$. Suppose (i) is the case, then 
by hypothesis, $b'$ is a descendant  of $b$. If (ii) is the case, then either (a) $c_c\colon(b', \text{descendantOf}, b)$ is 
the result of the application of a $brTR \in augC(R)$ on $dChase^{\text{csafe}}_k(QS_{\C})$ or
(b) $c_c\colon(b', \text{descendantOf}, b)$ is 
the result of the application of a $r \in augC(R)\setminus \{brTR\}$ on $dChase^{\text{csafe}}_k(QS_{\C})$. 
If (a) is the case, then there exists a $b'' \in \const(dChase^{\text{csafe}}_k(QS_{\C}))$ such that  
$c_c\colon(b'$, descendantOf, $b'')$ $\in $ $dChase^{\text{csafe}}_k(QS_{\C})$ and
$c_c\colon(b'', \text{descendantOf}, b) \in $ $dChase^{\text{csafe}}_k($ $QS_{\C})$. Hence, by hypothesis 
$b'$ is a descendantOf $b''$ and $b''$ is a descendantOf $b$. Since `descendantOf' relation is transitive, 
$b'$ is a descendantOf $b$.  Otherwise if (b) is the case then similar to the arguments used
 in the base case, it can easily be seen that $b'$ is a descendant of $b$.
\end{description}
Suppose if the quad \textbf{unCSafe} $\in$ $dChase^{\text{csafe}}(QS_{\C})$, then this 
implies that there exists an iteration $i$ such that 
the function unCSafeTest on $augC(r)$, with $r=$ $body(r)(\vec x,\vec z)$ $\rightarrow$ $head(r)(\vec x$, $\vec y)$ 
$\in R$, assignment $\mu$,
and $dChase^{\text{csafe}}_i(QS_{\C})$ returns True. This implies that, there exists 
   $b, b' \in \bn$, $y_j$  $\in$ $\{\vec y\}$ such that 
 $body(r)(\vec x, \vec z)[\mu]$ $\subseteq$ $dChase^{\text{csafe}}_i(QS_{\C})$, $b$ $\in$ $\{\mu(\vec x)\}$, $c_c\colon(b'$, descendantOf, $b)$ $\in$ 
 $dChase^{\text{csafe}}_i(QS_{\C})$
 and $\{c \ | \ c_c\colon(b'$, originContext, $c)$ $\in$ $dChase^{\text{csafe}}_i(QS_{\C})\}=
 cScope(y_j$,  $head(r)(\vec x$, $\vec y))$.
Suppose $k$ be the number of csafe dChase iterations
before $i$, in which $brTR$ was applied. Hence, by claim 0, $dChase_{i-k-1}(QS_{\C})$ $=$ $dChase^{\text{csafe}}_{i-1}(QS_{\C})(\C)$, 
and consequently $applicable_R($ $r$, $\mu$, $dChase_{i-k-1}(QS_{\C}))$ holds. Hence, as a result of $\mu$ being applied on $r$, 
there exists $b''=y_j[\mu^{ext(\vec y)}]$ $\in$ $\bn(dChase_{i-k}(QS_{\C})))$, with $b$ $\in$ $\{\mu(\vec x)\}$. Hence, by 
definition $originContext(b'')$ $=$ $cScope(y_j, head(r))$, and
$b$ is a descendantOf $b''$. If $b \neq b'$, then by Claim 2, $b'$ is a descendantOf $b$, otherwise $b'=b$ and 
hence $b'$ is a descendantOf $b''$. Consequently, $b'$ is a descendantOf $b''$. 
Also, applying claim 3, we get that $originContexts(b')$ $=$ $originContexts(b'')$, which means that
prerequisites of uncsafety is satisfied, and hence, $QS_{\C}$ is uncsafe. 
\end{proof}

\begin{lemma}[Completeness]\label{lemma:Completeness}
 For any quad-system, $QS_{\C}=\langle Q_{\C}, R\rangle$, if $QS_{\C}$ is uncsafe then 
$\textbf{unCSafe} \in dChase^{\text{csafe}}(QS_{\C})$.
\end{lemma}
\begin{proof}
We first prove a few supporting claims in order to prove the theorem. 
\begin{claim}(0)
 For any quad-system $QS_{\C}$ $=$ $\langle Q_{\C}$, $R\rangle$,  suppose $\textbf{unCSafe}$ $\not\in$ $dChase^{\text{csafe}}(QS_{\C})$, 
 then for any dChase iteration $i$, there exists a $j \geq 0$ such that  $dChase_i(QS_{\C})$ $=$ $dChase^{\text{csafe}}_{i+j}(QS_{\C})(\C)$.
\end{claim}
We approach the proof by induction on $i$.
\begin{description}
 \item[base case] for $i=0$, we know that $dChase_0(QS_{\C})=dChase^{\text{csafe}}_0(QS_{\C})$ $=$ $Q_{\C}$. Hence, the base case trivially holds.
 \item[hypothesis] for $i\leq k\in \mathbb{N}$, there exists $j\geq 0$ such that  $dChase_i(QS_{\C})$ $=$ $dChase^{\text{csafe}}_{i+j}($ $QS_{\C})$
 \item[step case] for $i$ $=$ $k+1$, one of the following holds: (a) $dChase_{k+1}(QS_{\C})$ $=$ $dChase_k($ $QS_{\C})$ or (b)
 $dChase_{k+1}(QS_{\C})$ $=$ $dChase_{k}(QS_{\C})$ $\cup$ $head(r)($ $\vec x$, $\vec y)[\mu^{ext(\vec y)}]$ and 
 $applicable_R(r$, $\mu$, $dChase_k($ $QS_{\C}))$ holds, for some $r$ $=$ $body(r)(\vec x$, $\vec z)$ $\rightarrow$ $head(r)($ $\vec x$, $\vec y)$, 
 assignment $\mu$. If (a) is the case, then trivially the claim holds. Otherwise, if (b) is the case, then let $j\in \mathbb{N}$ be such that  
 $dChase_k($ $QS_{\C})$ $=$ $dChase^{\text{csafe}}_{k+j}(QS_{\C})(\C)$. Let $j'\geq j, l \in \mathbb{N}$ be such that 
 $applicable_{augC(R)}(brTR$, $\mu$, $dChase^{\text{csafe}}_{k+l}(QS_{\C}))$, for any $j'\geq l \geq j$, and 
 $applicable_{augC(R)}(brTR$, $\mu$, $dChase^{\text{csafe}}_{k+j'+1}(QS_{\C}$ $))$ does not hold. By construction, it should be 
 the case that $applicable(r'$, $\mu$, \linebreak
 $dChase^{\text{csafe}}_{k+j'+1}(QS_{\C}))$ holds, where
 $r'$ $=$ $augC($ $r)$. Also since no new Skolem blank node was introduced in any csafe dChase iteration $k+l$, for any $j \leq l \leq j'$. 
 It should be the case that $head(r)[\mu^{ext(\vec y)}]$ $=$ $head(r')[\mu^{ext(\vec y)}](\C)$. Since, 
 $dCha$- $se^{\text{csafe}}_{k+l}(QS_{\C})(\C)$ $=$ $dChase_{k}(QS_{\C})$, for any $j \leq l \leq j'$, and 
 $dChase^{\text{csafe}}_{k+j'+1}(QS_{\C})$ $=$ $dChase^{\text{csafe}}_{k+j'}(QS_{\C})$ $\cup$ $head(r')[\mu^{ext(\vec y)}]$, 
 $dChase^{\text{csafe}}_{k+j'+1}(QS_{\C})(\C)$ $=$ $dChase_{k+1}($ $QS_{\C})$. Hence, 
 the claim follows.
\end{description}
The following claim, which straightforwardly follows from claim 0, shows that, for csafe quad-systems
its standard dChase is contained in its safe dChase.
\begin{claim}(1)
 Suppose $\textbf{unCSafe}$ $\not\in$ $dChase^{\text{csafe}}(QS_{\C})$, then $dChase(QS_{\C})$ $\subseteq$ 
 $dChase^{\text{csafe}}($ $QS_{\C})$. 
 \end{claim}
 Claim below shows that the generation of originContext quads in csafe dChase is complete.
\begin{claim}(2)
 For any quad-system $QS_{\C}$, if $\textbf{unCSafe}$ $\not\in$ $dChase^{\text{csafe}}(QS_{\C})$, then for any
 Skolem blank-node $b$ generated in $dChase(QS_{\C})$, and for any $c \in \C$, 
 if $c$ $\in$ $originCon$- $texts(b)$,
then there exists a quad $c_c\colon(b$, originContext, $c)$ $\in$ $dChase^{\text{csafe}}(QS_{\C})$.
\end{claim}
Since the only way a Skolem blank node $b$ gets generated in any iteration $i$ 
of $dChase($ $QS_{\C})$ is by the application of a BR $r\in R$, i.e. when there
 $\exists r$ $=$ $body(r)(\vec x$, $\vec z)$ $\rightarrow$ $head(r)(\vec x$, $\vec y) \in R$, assignment $\mu$, 
such that  $applicable_R(r$, $\mu$, $dChase_{i-1}(QS_{\C}))$, and $b$ $=$ $y_j[\mu^{ext(\vec y)}]$, for some $y_j$ $\in$ $\{\vec y\}$, 
and $dChase_{i}(QS_{\C})$ $=$ $dChase_{i-1}(QS_{\C})$ $\cup$ $head(r)(\vec x$, $\vec y)[\mu^{ext(\vec y)}]$. 
Also since $c$ $\in$ $originContexts(b)$, it should be the case that $c$ $\in$ $cScope(y_j$, $head(r))$. 
From claim 0, we know that there exists $j\geq 0$, such that  $dChase_i(QS_{\C})$ $=$ $dChase^{\text{csafe}}_{i+j}(QS_{\C})(\C)$. 
W.l.o.g, assume that $i+j$ is the first such csafe dChase iteration. Hence, it follows that 
$applicable_{augC(R)}(r'$, $\mu$, $dChase^{\text{csafe}}_{i+j-1}($ $QS_{\C}))$, where $r'=augC(r)$. Since,
$head(r) \subseteq head(r')$, it should be the case that $c \in$ $cScope(y_j$, $head(r'))$. Hence, by construction of $augC$, 
$c_c\colon (y_j$, originContext, $c)$ $\in$ $head(r')$, and as a result of application of $\mu$ on $r'$ 
in iteration $i+j$, $c_c\colon (b$, originContext, $c)$ gets generated in $dChase^{\text{csafe}}_{i+j}(QS_{\C})$.
Hence, the claim holds.

For the claim below, we introduce the concept of the sub-distance.
For any two blank nodes, their sub-distance is inductively defined as:
\begin{definition}
 For any two blank nodes $b, b'$, sub-distance$(b,b')$ is defined inductively as:
\begin{itemize}
 \item sub-distance$(b,b')=0$, if $b'=b$;
 \item sub-distance$(b,b')=\infty$, if $b\neq b'$ and $b$ is not a descendant of $b'$;
 \item sub-distance$(b,b')$ $=$ $min_{t \in \{\vec x[\mu]\}}\{$ sub-distance$(b$, $t)\}$ $+$ $1$, if $b'$ was generated by application 
 of $\mu$ on $r=body(r)(\vec x, \vec z)$ $\rightarrow$ $head(r)(\vec x, \vec y)$, i.e.
 $b'$ $=$ $y_j[\mu^{ext(\vec y)}]$, for some $y_j \in \{\vec y\}$,  and $b$ is a descendant of $b'$.
\end{itemize}
\end{definition}
\begin{claim}(3)
For any quad-system $QS_{\C}$ $=$ $\langle Q_{\C}$, $R\rangle$, if $\textbf{unCSafe}$ $\not \in$ 
$dChase^{\text{csafe}}(QS_{\C})$, then for any  two Skolem blank nodes $b, b'$ in $dChase(QS_{\C})$, if $b$ is a descendant of $b'$ then 
there exists a quad of the form $c_c\colon(b$, descendantOf, $b')$ $\in$ $dChase^{\text{csafe}}(QS_{\C})$.
\end{claim}
Note by the definition of sub-distance that if $b$ is a descendant of $b'$, 
then sub-distance$(b$, $b')$ $\in$ $\mathbb{N}$. Assuming
$\textbf{unCSafe}$ $\not \in$  $dChase^{\text{csafe}}(QS_{\C})$, and $b$ is a descendant of $b'$,
we approach the proof by induction on sub-distance$(b,b')$.
\begin{description}
 \item [base case] Suppose sub-distance$(b, b')=1$, then this implies that there exists $r$ $=$ $body(\vec x$, $\vec z)$ 
 $\rightarrow$ $head(r)(\vec x$, $\vec y)$, assignment $\mu$ such that  $b'$ was generated due to application of $\mu$ on $r$, i.e. 
 $b'$ $=$ $y_j[\mu^{ext(\vec y)}]$, for some $y_j \in \{\vec y\}$, and $b \in \{\vec x[\mu]\}$. This implies that there exists a dChase iteration $i$ 
 such that  $applicable_R(r$, $\mu$, $dChase_i(QS_{\C}))$ and $dChase_{i+1}(QS_{\C})$ $=$ $dChase_i(QS_{\C})$ $\cup$ $apply(r$, $\mu)$.
Since $\textbf{unCSafe}$ $\not \in$  $dChase^{\text{csafe}}(QS_{\C})$, using claim 0,  $\exists \ k\geq i$ such that  
$dChase_i(QS_{\C})$ $=$ $dChase^{\text{csafe}}_k(QS_{\C})(\C)$. W.l.o.g., let $k$ be the first such csafe dChase iteration. 
This means that $applicable_{augC(R)}(r'$, $\mu$, $dChase^{\text{csafe}}_k(QS_{\C}))$, where $r'=augC(r)$, and 
$dChase^{\text{csafe}}_{k+1}$ $=$ $dChas$ $e^{\text{csafe}}_k(QS_{\C})$ $\cup$ $head(r')[\mu^{ext(\vec y)}]$, 
and $b$, $b'$ $\in$ $head(r'$ $)[\mu^{ext(\vec y)}]$, $b \in \{\vec x[\mu]\}$, $b'=y_j[\mu^{ext(\vec y)}]$.  
By construction of $augC()$, since there exists a quad-pattern $c_c\colon (x_l$, descendantOf, $y_j) \in head(r')$, 
for any $x_l \in \{\vec x\}$, $y_j \in \{\vec y\}$, it follows that $c_c\colon (b$, descendantOf, $b') \in$ $dChase^{\text{csafe}}_{k+1}(QS_{\C})$.
 \item [hypothesis] Suppose sub-distance$(b,b') \leq k$,  $k \in \mathbb{N}$,
 then $c_c\colon(b$, \text{descendantOf}, $b')$ $\in$ $dChase^{\text{csafe}}(QS_{\C})$.
 \item [inductive step] Suppose sub-distance$(b,b')=k+1$, then there exists a $b''\neq b$, assignment $\mu$, and BR $r$ $=$ $body(r)(\vec x, \vec z)$ $\rightarrow$
 $head(r)(\vec x, \vec y) \in R$ such that  $b'$ was generated due to the application of $\mu$ or $r$ with $b''\in \{\vec x[\mu]\}$, i.e. 
 $b'=y_j[\mu^{ext(\vec y)}]$, for $y_j \in \{\vec y\}$, and  $b$ is a descendant of $b''$. This implies that sub-distance$(b'',b')=1$,
and sub-distance$(b,b'')$ $=$ $k$, and hence by hypothesis $c_c\colon(b$, descendantOf, $b'')$ $\in$ 
$dChase^{\text{csafe}}(QS_{\C})$, 
and $c_c\colon(b''$, descendantOf, $b')$ $\in$ $dChase^{\text{csafe}}(QS_{\C})$.  
Hence, by construction of csafe dChase, $c_c\colon (b$, descendantOf, $b')$ $\in$ $dChase^{\text{csafe}}($ $QS_{\C})$. 
\end{description}
\noindent Suppose $QS_{\C}$ is uncsafe, then by definition, there exists a blank nodes $b$, $b'$ in $\bn_{sk}($ $dChase(QS_{\C}))$,
such that  $b$ is descendant of $b'$, and $originContexts(b)$ is equal to $originContexts(b')$.
By contradiction, if $\textbf{unCSafe} \not \in dChase^{\text{csafe}}(QS_{\C})$, then by claim 1, 
$dChase(QS_{\C}) \subseteq dChase^{\text{csafe}}(QS_{\C})$. 
Since by claim 2, for any $c \in originContexts(b)$, there exists quads of the form 
$c_c\colon(b$, originContext, $c)$ $\in$ $dChase^{\text{csafe}}(QS_{\C})$ and for every $c' \in originContexts(b')$, 
there exists  $c_c\colon(b'$, originContext, $c')$  $\in$ $dChase^{\text{csafe}}(QS_{\C})$. 
Since $originContexts(b)$ $=$ $originContexts(b')$,
it follows that $\{c$ $|$ $c_c\colon(b$, originContext, $c)$ $\in$ $dChase^{\text{csafe}}($ $QS_{\C})\}$
$=$ $\{c'$ $|$  $c_c\colon(b'$, originContext, $c')$ $\in$ $dChase^{\text{csafe}}($ $QS_{\C})\}$
Also by claim 3, since $b$ is a descendant of $b'$, there exists a quad of the form 
$c_c\colon(b$, \text{descendantOf}, $b')$ in $dChase^{\text{csafe}}(QS_{\C})$. 
But, by construction of $dChase^{\text{csafe}}(QS_{\C})$, it should be the case that there exist a 
$b''\in \bn_{sk}(dChase^{\text{csafe}}(QS_{\C}))$, $r$ $=$ $body(r)(\vec x$, $\vec z)$ $\rightarrow$ $head(r)(\vec x$, $\vec y)$ 
$\in$ $augC(R)$, assignment $\mu$ such that  $b'$ was generated due to the application of $\mu$ on $r$, i.e. $b'$ $=$ $y_j[\mu^{ext(\vec y)}]$
with $b''$ $\in$ $\{\vec x[\mu]\}$, and 
$c_c\colon(b$, \text{descendantOf}, $b'')$
$\in$ $dChase^{\text{csafe}}(QS_{\C})$. But, since 
$\{c \ | \ c_c\colon(b$, originContext, $c)\in$ $dChase^{\text{csafe}}(QS_{\C})\}$
$=$ $cScope(y_j$, $head(r_i))$,
the method \textbf{unCSafeTest}$(r$, $\mu$, $dChase^{\text{csafe}}_l(QS_{\C}))$
should return True, for some $l \in \mathbb{N}$. Hence, it should be the case that
$\textbf{unCSafe}$ $\in$ $dChase^{\text{csafe}}(QS_{\C})$, which is a contradiction to our assumption.
Hence $\textbf{unCSafe}$ $\in$ $dChase^{\text{csafe}}(QS_{\C})$, if $dChase(QS_{\C})$ is uncsafe. 
\end{proof}

\begin{proof}[Property~\ref{prop:UniversalSafety}]
 (Only If)  By definition, $R$ is universally safe (resp. msafe, resp csafe) iff $\langle Q_{\C}, R\rangle$ is safe (resp. msafe, resp. csafe),
 for any quad-graph $Q_{\C}$. Hence, $\langle Q^{crit}_{\C}, R\rangle$ is safe (resp. msafe, resp. csafe).
 
 (If part) We give the proof for the case of safe quad-systems. The proof for the msafe and csafe case can be obtained by 
 slight modification.  In order to show that if $\langle Q^{crit}_{\C}$, $R\rangle$ is safe, then $R$ is universally safe, we prove the contrapositive. 
 That is we show that if there exists $Q_{\C}$ such that  $\langle Q_{\C}$, $R\rangle$ is unsafe, then $QS^{crit}_{\C}$ $=$ $\langle Q^{crit}_{\C}$, $R\rangle$ is unsafe. 
 Suppose, there exists such an unsafe quad-system $QS_{\C}$ $=$ $\langle Q_{\C}$, $R\rangle$, we show how to incrementally construct a homomorphism $h$ from 
 constants in $dChase(QS_{\C})$ to the constants in $dChase(QS^{crit}_{\C})$ such that  for any Skolem blank node $\_\colon b$ in 
 $dChase(QS_{\C})$,  there exists a homomorphism from 
 descendance graph of  $\_\colon b$  to  the descendance graph of $h(\_\colon b)$ in $dChase(QS^{crit}_{\C})$. Suppose $h$ is initialized as: 
 for any constant $c\in \const(QS_{\C})$, $h(c)=\_\colon b^{crit}$, if $c \in \const(QS_{\C})$ $\setminus$ $\const(QS^{crit}_{\C})$; 
 and  $h(c)=c$ otherwise . It can be noted that for any BR $r$ $=$ $body(r)(\vec x, \vec z)$
 $\rightarrow$ $head(r)(\vec x, \vec y)$ $\in R$, if $body(r)[\mu]$ $\subseteq$ 
 $dChase_0(QS_{\C})$ then $body(r)[\mu][h]$ $\subseteq$ $dChase_0(QS^{cric}_{\C})$. 
 Now it follows that for any $i \in \mathbb{N}$, $level(body(r)[\mu])=0$ if $applicable(r$, $\mu$, $dChase_i(QS_{\C}))$, 
 then there exists $j \leq i$ such that  $applicable(r$, $h\circ \mu$, $dChase_j(QS^{crit}_{\C}))$.
 Let $h$ be extended so that for any $i \in \mathbb{N}$, for any Skolem blank node $\_\colon b$ 
 introduced in $dChase_{i+1}(QS_{\C})$ while applying $\mu$ on $r$, for existential variable $y\in \{\vec y\}$, let $h(\_\colon b)$ be the blank node introduced in 
 $dChase_{j+1}(QS^{crit}_{\C})$, for the existential variable $y$ while applying $h\circ \mu$ on $r$. Hence, it follows that, for any $i \in \mathbb{N}$,
  $applicable_R(r$, $\mu$, $dChase_i(QS_{\C}))$ implies there exists $j \leq i$ such that $applicable(r$, $h\circ \mu$, $dChase_j(QS^{crit}_{\C}))$, 
  for any $r, \mu$.
   Also note that, for any Skolem blank node $\_\colon b$ generated in $dChase_i(QS_{\C})$, 
 it can be noted that  $\lambda_r(\_\colon b)$ $=$ $\lambda_r(h(\_\colon b))$ and 
 $\lambda_c(\_\colon b)$ $=$ $\lambda_c(h(\_\colon b))$ and $\lambda_v(\_\colon b)[h]$ $=$ 
 $\lambda_v(h(\_\colon b))$. Hence, it follows that for any Skolem blank node $\_\colon b$ in $dChase(QS_{\C})$, $h$ 
 is a homomorphism from descendance graph of $\_\colon b$ to the descendance graph of $h(\_\colon b)$ in 
 $dChase(QS^{crit}_{\C}$. Hence, 
 if there exists two Skolem blank nodes $\_\colon b$, $\_\colon b'$ in $dChase(QS_{\C})$, 
 with $\_\colon b'$ a descendant of $\_\colon b$ and $originRuleId(\_\colon b)$ $=$ $originRuleId(\_\colon b')$ and
 $originVec$- $tor(\_\colon b)$ $\cong$ $originVector(\_\colon b')$, then it follows that
 there exists $h(\_\colon b)$, $h(\_\colon b')$ in $dChase($ $QS^{crit}_{\C})$, with
 $h(\_\colon b')$  descendant of $h(\_\colon b)$ and $originRuleId(h(\_\colon b))$ $=$ $originRuleId(h(\_\colon b'))$ and
 $originVector(h(\_\colon b))$ $\cong$ $originVector(h(\_\colon b'))$. Hence, it follows from the definition that $QS^{critic}_{\C}$ is unsafe. 
\end{proof}
\end{document}